\documentclass[10pt]{article}
\usepackage{enumerate}
\usepackage{amsmath,amssymb}
\usepackage{caption}
\usepackage{subcaption}
\usepackage[usenames]{color}
\usepackage{bm}
\usepackage{ying}
\usepackage{multirow} 
\usepackage{rotating}
\usepackage{fullpage}
\usepackage{comment}
\usepackage{authblk}
\usepackage{enumerate}
\usepackage{geometry}
\usepackage{float}
\usepackage{parskip}
\usepackage{mathtools}
\mathtoolsset{showonlyrefs}
\allowdisplaybreaks

\newcolumntype{g}{>{\columncolor{red}}c}

\usepackage{graphicx}
\usepackage{xcolor}

\usepackage[colorlinks,
            linkcolor=red,
            anchorcolor=blue,
            citecolor=blue
            ]{hyperref}
\usepackage{algorithm}
\usepackage{algorithmic}

\let\hat\widehat
\let\tilde\widetilde

\def\asto{{\stackrel{\textnormal{a.s.}}{\to}}}
\def \pto {{\stackrel{\textrm{P}}{\to}}}
 
\def \iid {\stackrel{\textnormal{i.i.d.}}{\sim}}

\def \calib {\textnormal{calib}}
\def \train {\textnormal{train}}
\def \test {\textnormal{test}}

\def \TPP {\tilde{\PP}}
\def \pcausal {\cP}

\def \fdp {\textnormal{FDP}}
\def \fdr {\textnormal{FDR}} 
\def \bh {\textnormal{BH}}
\def \cc {\textnormal{CC}}
\def \homo {\textnormal{homo}}
\def \hete {\textnormal{hete}}
\def \dtm {\textnormal{dtm}}
\def \nw {\textnormal{nw}}
\def \ebh {\textnormal{eBH}}
 
\usepackage{paralist}

\definecolor{ForestGreen}{RGB}{34,139,34}
\definecolor{quote}{RGB}{103,148,54}
\def\##1\#{\begin{align}#1\end{align}}
\def\$#1\${\begin{align*}#1\end{align*}}
 
\newcommand{\revisea}[1]{{\color{black} #1}}

\makeatletter
\newcommand{\printfnsymbol}[1]{%
  \textsuperscript{\@fnsymbol{#1}}%
}
\makeatother

\title{Model-free selective inference under covariate shift\\ via weighted conformal p-values} 
\author[1]{Ying Jin} 
\author[1,2]{Emmanuel J. Cand\`es}
\affil[1]{Department of Statistics, Stanford University} 
\affil[2]{Department of Mathematics, Stanford University}
\date{}

\begin{document}

\maketitle

\begin{abstract}

  This paper introduces novel weighted conformal p-values and methods for model-free selective inference. The problem is as follows: given test units with covariates $X$ and missing responses $Y$, how do we select units for which the responses $Y$ are larger than user-specified values while controlling the proportion of false positives? Can we achieve this without any modeling assumptions on the data and without any restriction on the model for predicting the responses? Last, methods should be applicable when there is a covariate shift between training and test data, which commonly occurs in practice. 

  We answer these questions by first leveraging any prediction model to produce a class of well-calibrated weighted conformal p-values, which control the type-I error in detecting a large response. These p-values cannot be passed onto classical multiple testing procedures since they may not obey a well-known positive dependence property. Hence, we introduce weighted conformalized selection (WCS), a new procedure which controls false discovery rate (FDR) in finite samples. Besides prediction-assisted candidate selection, WCS (1) allows to infer multiple individual treatment effects, and (2) extends to outlier detection with inlier distributions shifts. We demonstrate performance via simulations and applications to causal inference, drug discovery, and outlier detection datasets.
\end{abstract}

\section{Introduction}


Many scientific discovery and decision making tasks concern the selection 
of promising candidates from a potentially large pool. In drug discovery, 
scientists aim to find
drug candidates---from a huge 
chemical space of molecules or compounds---with sufficiently 
high binding affinity to a target.  
College admission officers or hiring managers seek applicants with a high potential 
of excellence. In healthcare, one would like to allocate resources to individuals 
who would benefit the most. The list of such examples is endless.  
In all these problems, a decision maker would like to identify those units for which   
an unknown outcome  of interest (affinity for a target molecule/job performance/reduction in stress level) takes on a large value. 

Machine learning has shown tremendous potential for accelerating these 
resource- and time-consuming processes whereby 
predictions for the unknown outcomes are used to 
shortlist promising candidates before any further in-depth investigation~\cite{lavecchia2013virtual,carpenter2018machine,sajjadiani2019using}.  
This is a sharp departure from traditional strategies that either 
directly evaluate the outcomes via a physical screening to determine drug binding 
affinities, or ``predict'' the unknown outcomes via human judgement 
in healthcare or job hiring applications. 
In the new paradigm, a prediction model draws evidence from training data
to inform values the unknown outcome may take on given the features of a unit. 

If the intent is thus to use machine learning to screen candidates, 
we would need to understand and control 
the uncertainty in black-box prediction models as a selection tool. 
First, gathering evidence for new discoveries from past instances 
is challenging as they are not necessarily similar; i.e.~they may be drawn from distinct distributions.
Second, since a subsequent investigation applied to the shortlisted candidates 
is often resource intensive, it is meaningful to limit the \emph{error rate on the selected}. It is not so much the accuracy of the prediction that matters, rather the proportion of selected units not passing the test.  Keeping such a proportion at sufficiently low levels is a non-traditional criterion in typical prediction problems. 

This paper studies the reliable selection of promising units/individuals whose unknown outcomes 
are sufficiently large. We develop methods that apply in a model-free fashion---without making any modeling assumptions on the data generating process while controlling the expected proportion of false discoveries among the selected.

\subsection{Reliable selection under distribution shift}

Formally, suppose test samples $\{(X_{j},Y_{j})\}_{j\in\cD_\test}$ 
are drawn i.i.d.~from an unknown distribution $Q$, whose outcomes $\{Y_j\}_{j\in\cD_\test}$ 
are unobserved. Our goal is to identify a subset $\cR\subseteq \cD_\test$ 
with outcomes above some known, potentially random, 
thresholds $\{c_{ j}\}_{j\in \cD_\test}$. To draw statistical evidence for large outcomes in $\cD_\test$,  
we assume we hold an independent set of 
i.i.d.~calibration data $\{(X_i,Y_i)\}_{i\in \cD_\calib}$ whose outcomes  
are, this time, observed. 
 
Before we begin our discussion, we note that
a challenge that underlies almost all applications is that samples in $\cD_\calib$ 
are often different from those in $\cD_\test$. 
That is, in many cases, a point $(X_{j},Y_{j})$ in $\cD_\test$ is sampled from $Q$ while a point $(X_{i},Y_{i})$ in $\cD_\calib$ is sampled from a different distribution $P$. 
We elaborate on the distribution shift issue by considering two motivating examples. 

\vspace{-0.3em}

\paragraph{Drug discovery.} 
{Virtual screening, which identifies viable drug candidates 
using predicted affinity scores from supervised machine learning models, 
is increasingly popular in early stages of drug discovery~\cite{huang2007drug,koutsoukas2017deep,vamathevan2019applications,dara2021machine}. 
In this application, $\cD_\calib$ is the set of drugs with known affinities, while $\cD_\test$ 
is the set of new drugs whose binding affinities are of interest.}  
In practice, there usually is a distribution shift between the two batches: for instance, an experimenter may favor a 
specific structure 
when choosing candidates to physically screen their true binding affinities~\cite{polak2015early}. Since physically screened drugs are subsequently added to the calibration dataset $\cD_\calib$, they may not be representative of the new drugs in $\cD_\test$.
Accounting for such distribution shifts when applying virtual screening is crucial for 
the reliability of the whole drug discovery process, and 
has attracted recent attention~\cite{Krstajic2021}.

\vspace{-0.3em}

\paragraph{College admission.}
A college may be interested in prospective students who are likely to 
graduate after four years of undergraduate education (a binary outcome of interest) 
among all applicants in $\cD_\test$ as to maintain a reasonable graduation rate. 
The issue is that $\cD_\calib$
usually consists of students who \emph{were} admitted 
in the past as the outcome is only known for these students. 
Since previous cohorts may have been selected by applying various admission criteria, 
the distribution of the training data is different from that of the new applicants 
even if the distribution of applicants does not much vary over the years. 
Similar concerns apply to job hiring when using prediction models to select suitable applicants~\cite{faliagka2012application,shehu2016adaptive}; the way an institution chooses to record the outcomes of former candidates 
can lead to a 
discrepancy between the documented candidates and those under 
consideration.

\vspace{0.5em}

In this paper, 
we propose to address the distribution shift issue 
by considering a scenario where the shift 
can be entirely attributed to observed covariates; this is 
commonly referred to as   
a covariate shift in the literature~\cite{sugiyama2007covariate,tibshirani2019conformal}. 
Concretely, we assume that 
there is 
some function $w\colon \cX\to \RR^+$ such that 
\#\label{eq:cov_shift}
\frac{\ud Q}{\ud P}(x,y) = w(x),\quad  P\textrm{-almost surely},
\#
where $ \ud Q/\ud P $ denotes the Radon-Nikodym derivative. 
While all kinds of distribution shift could happen in practice, 
covariate shift---widely used to characterize  
distribution shifts  
caused by selection on known features---covers many important cases. 
In drug discovery, if the preference for choosing which samples 
to screen only depends upon the covariates, e.g.~the sequence of amino acids of the compound or the physical properties of the compound, we are dealing with the covariate shift \eqref{eq:cov_shift}.
In college admission or job hiring, 
if a specific type of applicant---characterized by certain features---were favored in 
previous admission/hiring cycles, such a preference would 
yield the shift \eqref{eq:cov_shift} between previous students/employees and current applicants. 

Returning to our decision problem,
we are interested in finding as many units $j\in \cD_\test$ as possible 
with $Y_{j}>c_j$, while controlling the false discovery rate (FDR) 
below some pre-specified level $q\in (0,1)$. 
The FDR is the expected fraction  of false discoveries, where a false discovery is a selected unit for which $Y_j \le c_j$, i.e., the outcome turns out to fall below the threshold. 
Formally, we wish to find $\cR$ obeying 
\#\label{eq:fdr}
\fdr := \EE\Bigg[  \frac{\sum_{j\in\cD_\test} \ind\{j\in \cR,Y_{j}\leq c_{ j}\}}{\max\{1,|\cR|\}}  \Bigg] \leq q, 
\#
where $q$ is the nominal target FDR level. The expectation in \eqref{eq:fdr} is taken   
over both the new test samples and the calibration data/screening procedure. 
The FDR quantifies the tradeoff 
between resources dedicated to the shortlisted candidates 
and the benefits from true positives~\cite{benjamini1995controlling}. 
By controlling the FDR, we ensure that a sufficiently large proportion of 
follow-up resources---such as clinical trials 
in drug discovery, human evaluation of student profiles and interviews---to confirm predictions from machine learning models are 
devoted to interesting candidates. 

The problem of FDR-controlled selection with machine learning predictions 
has been studied in an earlier paper~\cite{jin2022selection} 
under the assumption of no distribution shift. 
This means the new instances 
are exchangeable with the known calibration data, which, as we discussed above, 
is rarely the case in practice. 
{Therefore, there is a pressing need for novel techniques 
to address the distribution shift issue.}

\vspace{-0.3em}

\subsection{Preview of theoretical contributions}
 
{Our strategy consists in (1) constructing a calibrated confidence measure (p-value), which applies to any predictive model, for detecting a large unknown outcome, and (2) in employing multiple testing ideas to build the selection set from test p-values. }  

\vspace{-0.5em}
\paragraph{Weighted conformal p-value: a calibrated confidence measure.}
We introduce random variables $\{p_j\}_{j\in \cD_\test}$, which resemble p-values in the sense that for each $j\in \cD_\test$, they obey 
\#\label{eq:general_pvalue}
\PP(p_j \leq t, ~ Y_{j} \leq c_{j} ) \leq t,\quad \textrm{for all}~t\in [0,1],  
\# 
where the probability is over both the calibration data and the test sample 
$(X_{j},Y_{j},c_{j})$. 
This is similar to the definition of a classical p-value, hence, we call $p_j$ a``p-value''. With this, for any $\alpha\in(0,1)$, selecting $j$ if and only if $p_j\leq \alpha$ 
controls the type-I error below $\alpha$ (the chance that $Y_j$ is below threshold is at most $\alpha$).  
Note however that~\eqref{eq:general_pvalue} is perhaps an unconventional notion of 
type-I error, since it accounts for the randomness in the ``null hypothesis'' 
$H_{j}\colon Y_{j}\leq c_{j}$.

We construct $p_j$ using ideas from conformal prediction~\cite{vovk2005algorithmic,tibshirani2019conformal}. 
Take any monotone 
function $V\colon \cX\times\cY\to \RR$ (see Definition \ref{def:monotone}) 
built upon any prediction model that is independent of $\cD_\calib$ and 
$\cD_\test$. A concrete instance is 
$V(x,y) = y - \hat\mu(x)$, where $\hat\mu\colon \cX\to \cY$ 
is any predictor trained on a held-out dataset. 
With scores $V_i = V(X_i,Y_i)$,  $i\in \cD_\calib$, and 
weight function $w(\cdot)$ given by~\eqref{eq:cov_shift}, we construct $p_j$ to be approximately equal to  $\hat{F}(\hat{V}_j)$, 
where  $\hat{F}(\cdot)$ is the cumulative distribution function (c.d.f.) of the empirical distribution 
\[
\frac{ \sum_{i \in \cD_\calib} \,  w(X_i) \, \delta_{V_i} }{  \sum_{i \in \cD_\calib} w(X_i) }.
\] 
While the precise definition of $p_j$ is in  Section~\ref{sec:pval}, Figure~\ref{fig:pvalue} visualizes the idea. 
 
Intuitively, $p_j$ contrasts the value of $\hat{V}_{j}=V(X_{j},c_{j})$ against the distribution of the unknown `oracle' score $V_j=V(X_j,Y_j)$. The latter is 
approximated by $\hat{F}(\cdot)$, which reweights the training data 
with weights reflecting the covariate shift~\eqref{eq:cov_shift}. As such,  
$\hat{F}(V_j)$ is approximately uniformly distributed in $[0,1]$. 
Thus, a notably small value of $p_j\approx \hat{F}(\hat{V}_j)$ relative to the uniform distribution informs that $\hat{V}_j$ is smaller than what is typical of $V_j$, which further 
suggests evidence for $Y_j\geq c_j$. 
Leveraging conformal prediction techniques, such high-level ideas are made exact to achieve~\eqref{eq:general_pvalue}.

\begin{figure}[htbp]
    \centering 
    \begin{subfigure}[t]{0.45\linewidth}
        \centering
        \includegraphics[width=0.95\linewidth]{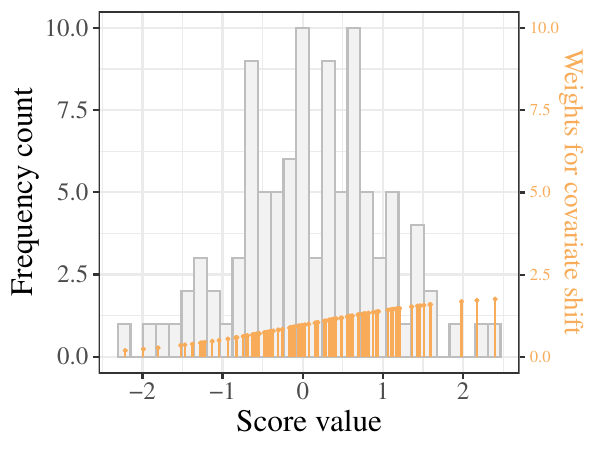} 
    \end{subfigure} \hspace{0.5em}
    \begin{subfigure}[t]{0.45\linewidth}
        \centering
        \includegraphics[width=0.95\linewidth]{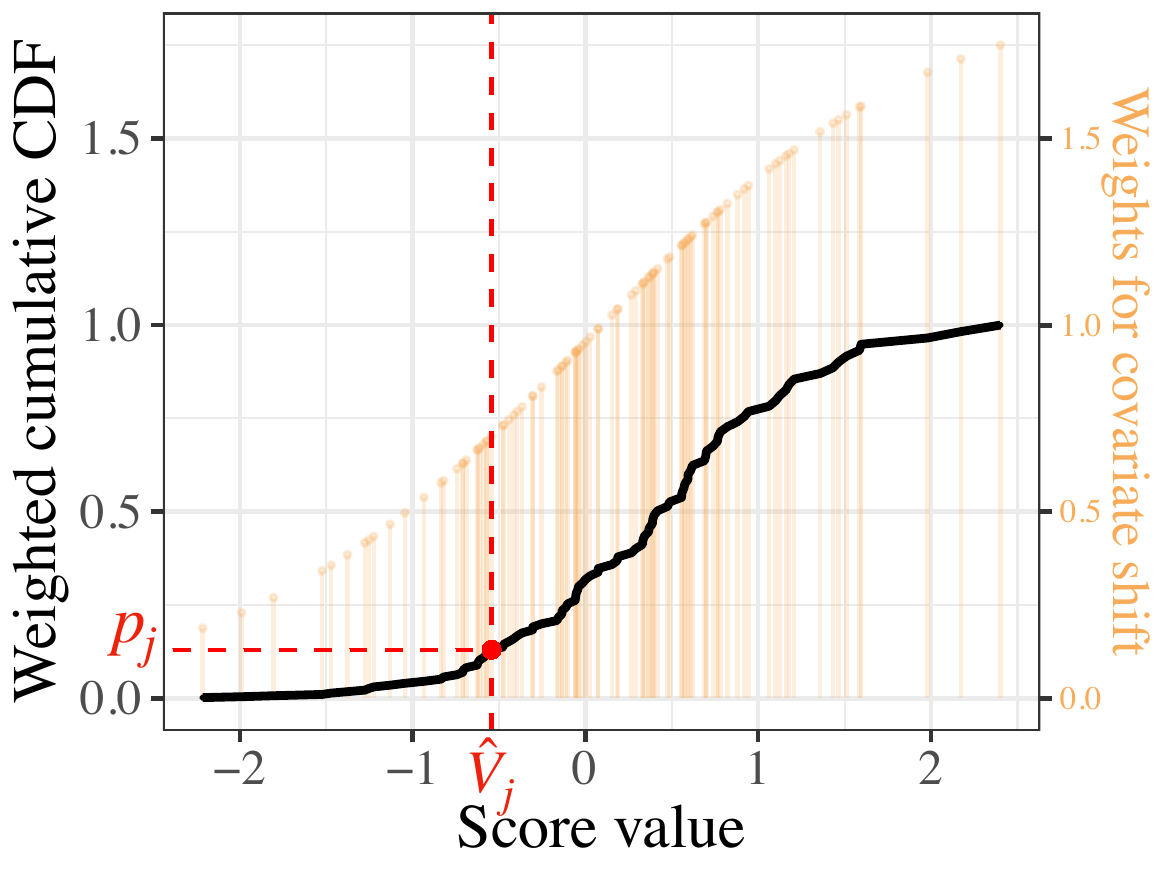} 
    \end{subfigure} 
\caption{Visualization of weighted conformal p-values.
Left: Distribution of scores in the calibration set (gray) and individual values of $w(X_i)$ (orange) for $i\in\cD_\calib$. Right: Weighted empirical c.d.f.~of $\{V_i\}_{i\in \cD_{\calib}}$; the red dashed lines illustrate how $p_j$ is computed from $\hat{V}_j$, in which $j\in \cD_\test$.  
}
\label{fig:pvalue}
\end{figure}

\vspace{-0.5em}
\paragraph{Intricate dependence among weighted conformal p-values.}

Going beyond per-unit type-I error control, 
selecting from multiple test samples 
requires dealing with multiple p-values $\{p_j\}_{j\in \cD_\test}$. 
A natural idea is then to use multiple testing ideas as if 
$\{p_j\}$ were classical p-values. 
However, the situation is complicated here because  $\{p_j\}_{j\in \cD_\test}$ 
all depend on the same set of calibration data, 
and, in sharp contrast to classical multiple testing, 
there is another source of randomness in the unknown outcomes. 

Our second theoretical result shows the mutual dependence among 
multiple p-values is particularly challenging. 
It states that   
the favorable positive dependence structure 
among p-values, which was the key to FDR control 
with other 
p-values derived with conformal prediction ideas~\cite{bates2021testing,jin2022selection}, 
can be violated in the presence of data-dependent weights in \eqref{eq:def_wcpval_rand}. 
As a result, it remains unclear whether applying existing multiple 
testing procedures controls the FDR in finite samples. 

\vspace{-0.5em}
\paragraph{Weighted Conformalized Selection (WCS).} 
Last, we introduce WCS, 
a new procedure 
that achieves finite-sample FDR control under distribution shift. 
The idea is to calibrate the selection threshold for each $j\in \cD_\test$ 
using a set of ``auxiliary'' p-values $\{\bar{p}_\ell^{(j)}\}_{\ell\in \cD_\test}$. 
These p-values  happen to obey 
a special conditional independence property, namely,  
\$
\bar p_j^{(j)}  \indep \{\bar p_\ell^{(j)}\}_{\ell\neq j} \biggiven \cZ_j,
\$
where $\cZ_j$ is the unordered set of the calibration data 
and the $j$-th test sample. 
We show that FDR control is valid regardless 
of the machine learning model used in constructing $V(\cdot,\cdot)$, 
and applies to a wide range of scenarios where the cutoffs $c_{j}$ can themselves be random variables.

\subsection{Broader scope}

Our methods happen to be applicable to additional statistical inference problems. 
 
\vspace{-0.3em}

\paragraph{Multiple counterfactual inference.}
The individual treatment effect (ITE) is a random variable 
that characterizes the difference between an individual’s outcome 
when receiving a treatment versus not, whose inference relies on comparing the observed outcome 
with the counterfactual~\cite{imbens2015causal,lei2021conformal,jin2023sensitivity}. To infer the counterfactual of a unit under treatment, 
the calibration data are units who do not receive the treatment.  
If the treatment assignment depends on the covariates, 
then there may be, and usually will be, a covariate shift between treated and control units. 
We discuss this problem in Section~\ref{sec:ite}.

\vspace{-0.3em}
\paragraph{Outlier detection.} Imagine we have access to i.i.d.~inliers $\{Z_i\}_{i\in \cD_\calib}$
from an unknown distribution $P$. Then outlier detection~\cite{bates2021testing} 
seeks to find outliers among new samples $\{Z_j\}_{j\in \cD_\test}$, i.e., those which 
do not follow the inlier distribution $P$. 
A potential application is fraud detection in financial transactions, 
where controlling the FDR ensures reasonable resource allocation  
in follow-up inquiries.  
However,  financial behavior varies with users and it is possible that normal transactions of interest follow different distributions across different populations. 
A related problem is to detect 
concept drifts; that is, allowing a subset of features $X\subseteq Z$ 
to follow a distinct distribution, and only detecting test samples 
whose $Z\given X$ distribution differs from $P_{Z\given X}$. We discuss this problem in Section~\ref{sec:outlier}.  

As in~\eqref{eq:general_pvalue}, 
the FDR~\eqref{eq:fdr} is marginalized over the randomness in the hypotheses. 
To this end, WCS requires  all 
$\{X_{n+j},Y_{n+j}\}_{j=1}^m$ to be i.i.d.~draws from a super-population. 
To relax this, we  generalize our methods to 
control a notion of FDR that is conditional on the random hypotheses. 
This is useful for handling imbalanced data in classification, and  
the aforementioned outlier detection problem.

\subsection{Data and code}
 
An R package, \texttt{ConfSelect}, that implements Weighted Conformalized Selection (as well as the version without weights) together with code to reproduce all experiments in this paper, can be found in the GitHub repository \href{https://github.com/ying531/conformal-selection}{https://github.com/ying531/conformal-selection}.

\subsection{Related work}

Our methods build upon the conformal inference framework~\cite{vovk2005algorithmic} {to infer unknown outcomes}. There, 
the theoretical guarantee for conformal prediction sets 
is usually marginally valid over the randomness in 
a single test point. 
As discussed in~\cite{jin2022selection},  
one might be interested in multiple test samples
simultaneously; in such situations, standard conformal prediction tools 
are insufficient.   

This work is connected to a recent line of research on 
selective inference issues arising in predictive inference. 
Some works use  
conformal p-values (whose definition differs from ours) 
with `full' observations $\{(X_{j},Y_j)\}_{j\in \cD_\test}$---that is, the response is observed---for outlier detection, i.e., for identifying whether $(X_{j},Y_j)$ 
follows the same distribution as the calibration data 
$\{(X_{i},Y_i)\}_{i\in \cD_\calib}$~\cite{bates2021testing,liang2022integrative,marandon2022machine}. 
Whereas these methods do not consider distribution shift between 
calibration and test inliers, 
Section~\ref{sec:outlier} extends all this. 
Here, the response $Y_{n+j}$---the object of inference---is not observed and, therefore, the problem is completely different. Consequently, methods and techniques to deal with this are new.  
Our focus on FDR control is similar in spirit to other recent papers; for instance, 
rather than marginal coverage of prediction sets,~\cite{fisch2022conformal} 
studies the number of false positives, 
and~\cite{bao2023selective} the 
miscoverage rate on selected units. 

Our methodology draws ideas from the multiple hypothesis testing literature, 
where the FDR is 
a popular type-I error criterion in exploratory analysis for controlling the 
proportion of `false leads' in follow-up studies~\cite{benjamini1995controlling,benjamini2001control}.  
This  paper is distinct from the conventional setting. 
First, testing a random outcome instead of a model parameter 
leads to a random null set and a complicated dependence structure. 
Second, our inference relies on (weighted)
exchangeability of the data while imposing no assumptions on their distributions. In contrast,  
a null hypothesis often 
specifies the distribution of test statistics or p-values. Later on, we shall also draw connections to a recent literature dealing with dependence in multiple testing, 
such as that concerned with e-values~\cite{vovk2021values,wang2020false} 
and conditional calibration of the BH procedure~\cite{fithian2020conditional}.  

The application of our method is related 
to a substantial literature that emphasizes 
the importance of uncertainty quantification in drug discovery~\cite{Norinder2014,svensson2017improving,Ahlberg2017current,svensson2018conformal,cortes2019concepts,wang2022improving}. 
Although existing works aim to employ conformal inference 
to assign confidence to machine prediction, 
this is often done  
in heuristic ways, which 
potentially invalidates the guarantees of conformal prediction due to selection issues. 
Instead, 
we provide here a principled approach to 
prioritizing drug candidates 
with clear and interpretable error control (i.e., guaranteeing a 
sufficiently high `hit rate'~\cite{wang2022improving}), 
and offer a potential solution to  
the widely recognized distribution shift problem~\cite{Krstajic2021,fannjiang2023novelty}.

The application to individual treatment effects in Section~\ref{sec:ite}
is connected to~\cite{caughey2021randomization};
the distinction is that they focus on a summary statistic such as 
a population quantile of the ITEs, 
while our method {detects individuals with positive ITEs.} 
Also related is~\cite{duan2021interactive},  
which tests individuals for positive treatment effects
with FDR control; however, in their setting a positive ITE 
means the whole distribution  of 
the control outcome is dominated by that of the treated outcome, 
whereas we
directly compare two random variables. 
Accordingly, our techniques are very different from these two references.

\section{Weighted conformal p-values}
\label{sec:pval}

\subsection{Construction of weighted conformal p-values}
 
Our weighted conformal p-values build upon 
any \emph{monotone} score function, defined as follows. 

\begin{definition}[Monotonicity]
    \label{def:monotone}
A score function $V(\cdot,\cdot)\colon \cX\times\cY\to \RR$ is monotone if 
$V(x,y)\leq V(x,y')$ holds for any $x\in \cX$ and any $y,y'\in \cY$ obeying $y\leq y'$. 
\end{definition}

Intuitively, the score function (often referred to as nonconformity score 
in conformal prediction~\cite{vovk2005algorithmic}) 
describes how well a hypothesized value $y\in \cY$ 
\emph{conforms} to the machine prediction. 
A popular and monotone nonconformity score 
is $V(x,y) = y-\hat\mu(x)$, where $\hat\mu\colon\cX\to \RR$ is 
a point prediction; other choices include those based on 
quantile regression~\cite{romano2019conformalized} or estimates of the conditional 
c.d.f.~\cite{chernozhukov2021distributional}. 

{From now on, write $\cD_\calib = \{1,\dots,n\}$ and $\cD_\test=\{n+1,\dots,n+m\}$. 
Compute $V_i=V(X_i,Y_i)$ for $i=1,\dots,n$  and 
$\hat{V}_{n+j}=V(X_{n+j},c_{n+j})$ for $j=1,\dots,m$. 
Our weighted conformal p-values are defined as  
\#\label{eq:def_wcpval_rand}
p_j = \frac{\sum_{i=1}^n w(X_i)\ind {\{V_i <\hat{V}_{n+j} \}}+  ( w(X_{n+j}) + \sum_{i=1}^n w(X_i)\ind {\{V_i = \hat{V}_{n+j} \}})\cdot U_j}{\sum_{i=1}^n w(X_i) + w(X_{n+j})},
\#
where $w(\cdot)$ is the covariate shift function in~\eqref{eq:cov_shift}, 
and  $U_j\iid \textrm{Unif}([0,1])$ are tie-breaking random variables. 
When $w(\cdot)\equiv 1$, our p-values reduce to the conformal p-values in~\cite{jin2022selection}, see Appendix~\ref{app:recap_unweighted}. 
We also refer readers to Appendix~\ref{app:connection_conformal} for 
the connection between 
our p-values and conformal prediction intervals.}

{Roughly speaking, with a monotone score, 
the p-value~\eqref{eq:def_wcpval_rand} 
measures how small $c_{n+j}$ is compared with 
typical values of $Y_{n+j}$, and provides calibrated 
evidence for a large outcome as expressed in~\eqref{eq:general_pvalue}. 
This can be derived using conformal inference theory~\cite{tibshirani2019conformal} and the monotocity of $V$. We include a formal result here for completeness, whose proof is 
in Appendix~\ref{app:lem_general_pval}.} 

\begin{lemma}\label{lem:general_pval}
The tail inequality~\eqref{eq:general_pvalue} holds 
if the covariate shift~\eqref{eq:cov_shift} holds and the score function is monotone.
\end{lemma}
Consequently, selecting a unit 
if $p_j \le \alpha$ controls the type-I error at level $\alpha$ for each $j\in \cD_\test$.  
We shall however see that dealing with multiple test samples 
is far more challenging.

\subsection{Weighted conformal p-values are not PRDS}
\label{subsec:no_PRDS}

Applying multiple testing procedures 
to our p-values to obtain a selection set naturally comes to mind. Previous works 
applying the  Benjamini-Hochberg (BH) procedure~\cite{benjamini1995controlling} to p-values obtained via conformal inference ideas indeed control the FDR~\cite{bates2021testing,jin2022selection,marandon2022machine}. This is because the p-values obey a favorable dependence structure called 
\emph{positive dependence on a subset} (PRDS).

\begin{definition}[PRDS]
\label{def:prds}
A random vector $X=(X_1,\dots,X_m)$ is PRDS 
on a subset $\cI$ if for any $i\in \cI$ 
and any increasing set $D$, 
the probability $\PP(X\in D\given X_i=x)$ is increasing in $x$. 
\end{definition}
Above, a set $D\subseteq \RR^m$ is 
\emph{increasing} if $a\in D$ and $b\succeq a$ implies $b\in D$ ($b \succeq a$ means that all the components of $b$ are larger than or equal to those of $a$). It is well-known that the BH procedure 
controls the FDR when applied to PRDS p-values~\cite{benjamini2001control}.  
The PRDS property among (unweighted) conformal p-values was first studied in
\cite{bates2021testing},\footnote{They also show the non-randomized unweighted conformal 
p-values are PRDS when the scores are almost surely distinct. 
We do not elaborate on this distinction as it is not the focus here.} 
and generalized in  
\cite{jin2022selection} 
to plug-in values $c_{n+j}$ in lieu of $Y_{n+j}$, thereby  
forming the basis for FDR control. 
If the PRDS property were to hold for 
the weighted conformal p-values, then applying the BH procedure 
would also control the FDR. This is however not always the case; 
a constructive proof of Proposition~\ref{prop:counter_PRDS} is 
in Appendix~\ref{app:subsec_prds}. 

\begin{prop}\label{prop:counter_PRDS}
There exist a sample size $n$, a weight function $w\colon \cX\to \RR^+$, 
a nonconformity score $V\colon \cX\times\cY\to \RR$, 
and  distributions $P$ and $Q$ obeying~\eqref{eq:cov_shift}, 
such that 
the weighted conformal p-values~\eqref{eq:def_wcpval_rand} are not PRDS for 
$\{X_i,Y_i\}_{i=1}^n\iid P$ and $\{X_{n+j},Y_{n+j}\}_{j=1}^m\iid Q$. 
\end{prop} 

A little more concretely, we find that the PRDS property 
is likely to fail when 
the nonconformity scores are negatively 
correlated with the 
weights.  
Note that the dependence among conformal p-values 
arises from sharing the calibration data. 
Under exchangeability (i.e.~taking equal weightes $w(x)\equiv 1$), a smaller conformal p-value 
is (only) associated with larger calibration scores, 
and hence smaller p-values for other test points. 
However, with negatively associated  weights and scores, 
a smaller p-value may also be due to 
smaller calibration weights---instead of larger calibration scores---which 
implies that other p-values may take on larger values.

\subsection{BH procedure controls FDR asymptotically}

Before introducing our new selection methods, 
we nevertheless show {\em asymptotic} FDR control using the BH procedure.   

\begin{theorem}\label{thm:fdr_asymp}
Suppose $w(\cdot)$ is uniformly bounded by a fixed constant, the covariate shift~\eqref{eq:cov_shift} holds, 
and $\{c_{n+j}\}_{j=1}^m$ are i.i.d.~random variables. 
Fix any $q\in(0,1)$, 
and let $\cR$ be the output of the BH procedure applied to 
$\{p_j\}_{j=1}^m$ in~\eqref{eq:def_wcpval_rand}. 
Then the following  results hold. 
\begin{enumerate}[(i)]
\item For any fixed $m$, 
it holds   that $\limsup_{n\to \infty} \fdr \leq q$.
\item Suppose $m,n\to \infty$. Under a mild technical condition (see Appendix~\ref{app:fdr_asymp}), $\limsup_{m,n\to\infty} \fdr \leq q$. Furthermore, the asymptotic FDR and power of BH  can be exactly characterized in this case.  
\end{enumerate} 
\end{theorem}

We provide the complete technical version of Theorem~\ref{thm:fdr_asymp} 
in Theorem~\ref{thm:fdr_asymp_full}, whose 
proof is in Appendix~\ref{app:thm_fdr_asymp}.
In fact, we prove that as $n\to \infty$, the weighted conformal p-values 
converge to i.i.d.~random variables that are dominated by the uniform distribution, 
which ensures asymptotic FDR control (see Appendix~\ref{app:thm_fdr_asymp} for details). 
The technical condition  we impose resembles that in~\cite{storey2004strong}, 
which ensures the existence of a limit point of the rejection threshold and enables the asymptotic analysis. 

We also  remark that 
the PRDS property is a sufficient, but not necessary, condition for 
the BH procedure to control the FDR.  
In fact, we will see  that 
the BH procedure with weighted conformal p-values  
empirically controls 
the FDR in all of our numerical experiments, hence it 
remains a reasonable option in practice.

\section{Finite-sample FDR control}
\label{sec:method}

We now introduce a new 
multiple testing procedure, WCS, that 
controls the FDR in finite samples with weighted conformal p-values. 
Our method builds on the following p-values:  
\#\label{eq:weighted_pval}
{p}_j = \frac{\sum_{i=1}^n w(X_i)\ind{\{V_i < \hat{V}_{n+j} \}}+  w(X_{n+j}) }{\sum_{i=1}^n w(X_i) + w(X_{n+j})}.
\#
It is a slight modification of~\eqref{eq:def_wcpval_rand} up to 
random tie-breaking. The asymptotic results in 
Theorem~\ref{thm:fdr_asymp} also hold with these  
p-values as long as the distributions of the $V_i$'s and 
$\hat{V}_{n+j}$'s do not have point masses. 

\subsection{Weighted Conformalized Selection}
 
As before, we compute $V_i=V(X_i,Y_i)$ for $i=1,\dots,n$ 
and $\hat{V}_{n+j} = V(X_{n+j},c_{n+j})$ for $j=1,\dots,m$. 
For each $j$, we compute a set of auxiliary p-values 
\#\label{eq:mod_pval}
 {p}_{\ell}^{(j)} := \frac{\sum_{i=1}^n w(X_i) \ind {\{V_i < \hat{V}_{n+\ell}\}} + w(X_{n+j}) \ind {\{ \hat{V}_{n+j}<  \hat{V}_{n+\ell}\}} }{\sum_{i=1}^n w(X_i) + w(X_{n+j})}, \quad \ell\neq j.
\# 
Then, we let $\hat{\cR}_{j\to 0}$ be the rejection set 
of BH applied to 
$\{  p_1^{(j)},\dots,  p_{j-1}^{(j)}, 0,   p_{j+1}^{(j)},\dots,  p_{m}^{(j)}\}$ at the nominal level $q$. Note that $j\in \hat{\cR}_{j\to 0}$ by default, and  
set $s_j = \frac{q|\hat{\cR}_{j\to 0}  | }{m}$.  
We then compute the first-step rejection set $\cR^{(1)}:=\{j \colon   p_j \leq s_j\}$. 
Finally, we prune $\cR^{(1)}$ 
to obtain the final selection set $\cR$ 
using any of the following three methods. 
\begin{enumerate}[(a)]
\item \emph{Heterogeneous pruning}: 
generate i.i.d.~random variables $\xi_j \sim \textrm{Unif}[0,1]$, 
and set  
\#\label{eq:R_hete}
\cR := \cR_{\hete} = \Big\{ j\colon  p_j \leq s_j, ~ \xi_j |\hat\cR_{j\to 0} | \leq  r_{\hete}^*  \Big\},
\#
where 
$
r_\hete^* := \max\big\{ r\colon  \sum_{j=1}^m \ind {\{  p_j\leq s_j, \, \xi_j   |\hat\cR_{j\to 0}  | \leq r \}}   \geq r \big\}.
$
\item \emph{Homogeneous pruning}:  
generate an independent $\xi\sim \textrm{Unif}[0,1]$, 
and set 
\#\label{eq:R_homo}
\cR :=  \cR_{\homo} = \Big\{ j\colon  p_j \leq s_j, ~ \xi  |\hat\cR_{j\to 0} | \leq  r_\homo^*  \Big\},
\#
where  
$
r_{\homo}^* := \max\big\{ r\colon  \sum_{j=1}^m \ind {\{ p_j\leq s_j, \, \xi   |\hat\cR_{j\to 0}  | \leq r \}}   \geq r \big\}.
$
\item \emph{Deterministic pruning}: define the rejection set 
\#\label{eq:R_dete}
\cR := \cR_{\dtm} = \Big\{ j\colon p_j \leq s_j, ~  |\hat\cR_{j\to 0} | \leq  r_\dtm^*  \Big\},
\#
where 
$
r_{\dtm}^* := \max\big\{ r\colon  \sum_{j=1}^m \ind {\{ p_j\leq s_j, \, |\hat\cR_{j\to 0}  | \leq r \}}   \geq r \big\}.
$
In all options, we use $p_j$ defined in~\eqref{eq:weighted_pval}.  
\end{enumerate}

It is straightforward to see that $\cR_{\dtm}\subseteq \cR_{\hete}$ and 
$\cR_{\dtm}\subseteq \cR_{\homo}$, i.e., 
random pruning leads to larger rejection sets. 
The selection procedure is summarized in Algorithm~\ref{alg:bh}. 

\begin{algorithm}[h]
  \caption{Weighted Conformalized Selection}\label{alg:bh}
  \begin{algorithmic}[1]
  \REQUIRE Calibration data $\{(X_i,Y_i)\}_{i=1}^n$, 
  test data  $\{X_{n+j}\}_{j=1}^m$, 
  thresholds $\{c_{n+j}\}_{j=1}^m$, 
  weight $w(\cdot)$,
  FDR target $q\in(0,1)$, monotone nonconformity score $V\colon \cX\times\cY\to \RR$, 
  pruning method $\in\{\texttt{hete}, \texttt{homo}, \texttt{dtm}\}$.
  \vspace{0.05in} 
  \STATE Compute $V_i = V(X_i,Y_i)$ for $i=1,\dots,n$,  
  and $\hat{V}_{n+j}= V(X_{n+j},c_{n+j})$ for $j=1,\dots,m$.
  \STATE Construct weighted conformal p-values $\{ p_j\}_{j=1}^m$ as in~\eqref{eq:weighted_pval}. 

  \vspace{0.3em}
  \noindent \texttt{- First-step selection -}
  \FOR{$j=1,\dots,m$}
  \STATE Compute p-values $\{ {p}_\ell^{(j)}\}$ as in~\eqref{eq:mod_pval}.
  \STATE (BH procedure) Compute $k^*_j = \max\big\{k\colon 1 +\sum_{\ell\neq j} \ind\{{p}_\ell^{(j)}\leq qk/m\}\geq k\big\}$. 
  \STATE Compute $\hat{\cR}_{j\to 0} = \{j\}\cup\{\ell \neq j\colon  {p}_\ell^{(j)}\leq q k^*_j /m\}$.
  \ENDFOR
  \STATE Compute the first-step selection set $\cR^{(1)} = \{j\colon  {p}_j \leq q|\hat\cR_{j\to 0}|/m\}$.
  
  \vspace{0.3em}
  \noindent \texttt{- Second-step pruning -}
  \STATE Compute $\cR = \cR_{\textrm{hete}}$ as in~\eqref{eq:R_hete}
  or $\cR = \cR_{\textrm{homo}}$ as in~\eqref{eq:R_homo}
  or $\cR = \cR_{\textrm{dtm}}$ as in~\eqref{eq:R_dete}. 
  \vspace{0.05in}
  \ENSURE Selection set $\cR$.
  \end{algorithmic}
\end{algorithm}

\subsection{Theoretical guarantee}

The following theorem, whose proof is in Appendix~\ref{app:thm_calib_ite}, shows exact finite-sample FDR control.

\begin{theorem}\label{thm:calib_ite}
  Write $Z_i=(X_i,Y_i)$ for $i=1,\dots,m+n$, 
  and $\tilde{Z}_{n+j}=(X_{n+j},c_{n+j})$ for $j=1,\dots,m$.  Suppose $\{Z_i\}_{i=1}^n\iid P$ 
and $\{Z_{n+j}\}_{j=1}^m\iid Q$, and~\eqref{eq:cov_shift} holds 
for 
the input weight function $w(\cdot)$ of 
Algorithm~\ref{alg:bh}. 
Assume that for each $j=1,\dots,m$, the samples in  $\{Z_1,\dots,Z_n,Z_{n+j}\}\cup\{\tilde{Z}_{n+\ell}\}_{\ell\neq j}$ are mutually independent.  
Then with either 
$\cR \in \{\cR_{\hete}, \cR_{\homo}, \cR_{\dtm}\}$, it holds that  
\$
\EE\Bigg[  \frac{\sum_{j=1}^m \ind {\{j \in \cR, Y_{n+j}\leq c_{n+j}\}} }{1\vee |\cR|} \Bigg]   \leq q 
\$
(the expectation is taken over both calibration and test data).  
\end{theorem}

As we allow for random thresholds $c_{n+j}$, the independence assumption 
rules out the cases where the thresholds $c_{n+j}$ are adversarially chosen; the reason is that they would break certain exchangeability properties between scores $\{\hat{V}_{n+\ell}\}_{\ell=1}^m$ 
and calibration scores $\{V_{i}\}_{i=1}^n$. Independence holds 
if $c_{n+j}$ is pre-determined, or more generally, if $c_{n+j}$ is a random variable associated 
with independent test samples (such as when
$\{X_{n+j},c_{n+j},Y_{n+j}\}_{j=1}^m$ are i.i.d.~tuples that 
are independent of the calibration data).  Examples 
include 
individual treatment effects studied in Section~\ref{sec:ite}, 
and the drug-target interaction prediction task studied in Section~\ref{subsec:dti}.  
Other examples with random thresholds related to healthcare are in~\cite[Section 2.4]{jin2022selection}.
 
While the $p_j$'s are still marginally stochastically larger than uniforms, their complicated dependence requires additional care, and this is why 
we use the auxiliary p-values $\{p_\ell^{(j)}\}$ 
as `calibrators' to determine a potentially different rejection rule 
than naively applying the BH procedure. The proof of Theorem~\ref{thm:calib_ite} 
relies on comparing our p-values~\eqref{eq:weighted_pval}  to 
oracle p-values  
\#\label{eq:orc_w_pval_nr}
\bar p_j  = \frac{\sum_{i=1}^n w(X_i)\ind {\{V_i < {V}_{n+j} \}}+  w(X_{n+j})  }{\sum_{i=1}^n w(X_i) + w(X_{n+j})}.
\#
As we shall see, 
replacing our conformal p-values with 
their oracle counterparts does not change the first-step 
rejection set $\hat\cR_{j\to 0}$. The advantage of working with the oracle p-values is a friendly dependence structure, which ultimately ensures FDR control. 
The main ideas of the argument are: 
 
\begin{itemize} 
  \item {\em Randomization reduction:} for each $j$,
  $\hat\cR_{j\to 0}$ can be expressed as $\hat\cR_{j\to 0} = f(p_j,\bm{p}_{-j}^{(j)} )$,
  where $\bm{p}_{-j}^{(j)}:=\{p_\ell^{(j)}\}_{\ell\neq j}$. For all three pruning options, it holds that $$\fdr\leq \sum_{j=1}^m \EE\bigg[\frac{ \ind\{  p_j \leq q|f( p_j,\bm{p}_{-j}^{(j)})|/m\}}{ |f( p_j, {\bm{p}}_{-j}^{(j)})| }\bigg].$$ 
  \item {\em Leave-one-out:} $\hat\cR_{j\to 0}$ remains the same if we replace all $\{p_\ell\}$ with $\{\bar{p}_\ell\}$ defined above, i.e., $f( p_j,\bm{p}_{-j}^{(j)})= f(\bar p_j,\bar{\bm{p}}_{-j}^{(j)})$. 
  Since $\hat\cR_{j\to 0}$ is obtained by BH, it follows that $f(\bar p_j,\bar{\bm{p}}_{-j}^{(j)})=f(0,\bar{\bm{p}}_{-j}^{(j)})$. 
  \item {\em Conditional independence:} 
under the covariate shift~\eqref{eq:cov_shift}, the oracle p-values possess a nice conditional independence structure, namely, $\bar{p}_j \indep \bar{\bm{p}}_{-j}^{(j)} \given \cZ_j$, where $\cZ_j = [Z_1,\dots,Z_n,Z_{n+j}]$ is an unordered set of $Z_i=(X_i,Y_i)$, $i=1,\dots,n,n+j$. This, together with the first two steps and the fact that 
  $\bar{p}_j$ is stochastically larger than a uniform random variable conditional on $\cZ_j$, gives $\fdr\leq q$. 
\end{itemize}

Along the way, we shall explore interesting connections between Algorithm~\ref{alg:bh} 
and existing ideas in the multiple testing literature. 

\begin{remark}\normalfont
Algorithm~\ref{alg:bh} is related to the 
conditional calibration method~\cite{fithian2020conditional}, 
which utilizes sufficient statistics to calibrate a 
rejection threshold  for each individual hypothesis 
to achieve finite-sample FDR control. 
Indeed, for each $j$, 
$s_j:= {q|\hat{\cR}_{j\to 0}  | }/{m}$ 
can be viewed as the `calibrated threshold' for p-values in their framework; 
the unordered set $\cZ_j$ (after a careful leave-one-out analysis
in our proof)
plays a similar role as a `sufficient statistic'. 
Our heterogeneous pruning is similar to their random pruning step, 
while the other two options generalize their approach.
In addition, the procedure and its theoretical analysis  
are specific to our problem, and they are significantly different.
\end{remark}

\begin{remark}\normalfont
Algorithm~\ref{alg:bh} is also connected to 
the eBH procedure~\cite{wang2020false}, a generalization of 
the BH procedure to {\em e-values}. 
In the conventional setting with 
a set of (deterministic) hypotheses $\{H_j\}_{j=1}^m$, 
e-values are nonnegative random variables $\{e_j\}_{j=1}^m$  
such that $\EE[e_j]\leq 1$ 
if $H_j$ is null. 
For a target level $q\in(0,1)$, 
the eBH procedure outputs 
$\cR_{\ebh}:=\{j\colon e_j \geq m/(q\hat{k})\}$, where 
$\hat{k} =\max\big\{k\colon \sum_{j=1}^m \ind\{e_j\geq m/(qk)\}\geq k\big\}$. 
One can check that $\cR_{\dtm}$ is 
equivalent to $\cR_{\ebh}$ applied to 
\#\label{eq:eval}
e_j := \frac{\ind\{p_j\leq q|\hat\cR_{j\to 0}|/m\}}{q|\hat\cR_{j\to 0}|/m},
\quad j=1,\dots,m.
\#
Similar to~\eqref{eq:general_pvalue}, 
the $e_j$'s above obey a generalized notion of ``null'' e-values, in the sense that 
\$
\EE\big[e_j\ind\{j\in \cH_0\}\big]\leq 1,\quad \textrm{for all}~j=1,\dots,m, 
\$
see Appendix~\ref{app:thm_calib_ite} for details. 
Furthermore, using the generalized e-values~\eqref{eq:eval},  
$\cR_\hete$ and $\cR_\homo$ are equivalent to running eBH 
using $\{e_j/\xi_j\}$ and $\{e_j/\xi\}$, respectively. 
Note that $\{e_j/\xi_j\}$ and $\{e_j/\xi_j\}$ are no longer 
e-values, yet our procedures still control the FDR while achieving higher power. 
\end{remark}

Our next result shows that the 
extra randomness in the second step 
does not incur too much additional variation for $\cR_\hete$ and $\cR_\homo$.  
The proof of Proposition~\ref{prop:asymp_equiv} is in Appendix~\ref{app:subsec_asymp_equiv}.

\begin{prop}
\label{prop:asymp_equiv}
Suppose $\|w\|_\infty \leq M$ for some constant $M>0$. 
Suppose the distributions of $\{V(X_i,Y_i)\}_{i=1}^n$ 
and $\{V(X_{n+j},c_{n+j} )\}_{j=1}^m$ have no point mass.  
Let $\cR_{\bh}$ be the rejection set 
of BH with weighted conformal $p$-values~\eqref{eq:def_wcpval_rand}, 
and $\cR$ be any of the three selections from Algorithm~\ref{alg:bh}. Let $\cR^{(1)}=\{j\colon  p_j\leq q|\hat\cR_{j\to 0}|/m\}$ 
be the first-step selection set. 
Then the following holds: 
\begin{enumerate}[(i)]
\item If $m$ is fixed, then $\lim_{n\to \infty}\PP(\cR_{\bh}= \cR = \cR^{(1)})=1$ for each $\cR \in \{\cR_{\homo},\cR_{\hete}, \cR_{\dtm}\}$. 
\item If $m,n\to \infty$ and the regularity conditions in case (ii) of Theorem~\ref{thm:fdr_asymp}  hold, 
then  
 $\frac{|\cR^{(1)}\Delta \cR_{\bh}|}{|\cR^{(1)}|} ~\asto~0$, 
 $\frac{|\cR^{(1)}\Delta \cR_{\bh}|}{|\cR_\bh|} ~\asto~0$, 
$ \frac{|\cR^{(1)}\Delta \cR_{\homo}|}{|\cR^{(1)}|} ~\pto ~0$, and 
$ \frac{|\cR_\bh\Delta \cR_{\homo}|}{|\cR_\bh|} ~\pto ~0$.
\end{enumerate}
\end{prop}
 
Proposition~\ref{prop:asymp_equiv} 
shows that when the size of the calibration set  
is sufficiently large,  
our procedure is  close to the BH procedure applied to 
the weighted conformal p-values, the latter being a deterministic function of the data. 
In particular, our first-step rejection set is asymptotically equivalent to BH 
under mild conditions, and the same applies to 
$\cR_{\homo}$. 
The asymptotic analysis of $\cR_{\hete}$ and $\cR_{\dtm}$ is 
challenging. 
In our numerical experiments, we find that $\cR_{\hete}$ 
is often very close to $\cR_{\homo}$, 
while $\cR_{\dtm}$ usually makes too few rejections.

\subsection{FDR bounds with estimated weights}

In practice, the covariate shift $w(\cdot)$ may be unknown. 
When the conformal p-values 
are computed with some fitted weights $\hat{w}(X_i)$, 
Theorem~\ref{thm:est_w} establishes upper bounds on the FDR.  
The proof is in Appendix~\ref{app:thm_est_w}. 

\begin{theorem}\label{thm:est_w}
Take $w(\cdot):=\hat{w}(\cdot)$ 
as the input weight function in Algorithm~\ref{alg:bh} and assume $\hat{w}(\cdot)$ 
is estimated in a training process with data independent from 
$\{(X_i,Y_i)\}_{i=1}^{n}\cup\{(X_{n+j},c_{n+j})\}_{j=1}^m$. 
Under the conditions of Theorem~\ref{thm:calib_ite}, for each $\cR \in \{\cR_{\homo},\cR_{\hete}, \cR_{\dtm}\}$, we have 
\$
\fdr \leq  q\cdot \EE\bigg[ \frac{\hat\gamma^2}{1+ q(\hat\gamma^2-1)/m} \bigg],
\$
where $\hat\gamma := \sup_{x\in \cX} \max\{  \hat{w}(x)/w(x),\, w(x)/\hat{w}(x) \}$. 
\end{theorem}

Above, $\hat\gamma\geq 1$ measures the estimation error in $\hat{w}(\cdot)$ 
relative to the true weights. 
When both $w(\cdot)$ and $\hat{w}(\cdot)$ 
are bounded away from zero and infinity, 
$\hat\gamma-1$ is of the same order as $ \sup_x|\hat{w}(x)-w(x)|$,  
and hence the FDR inflation in Theorem~\ref{thm:calib_ite} converges to zero 
if $\hat{w}$ is consistent.

\section{Application to drug discovery datasets}
\label{sec:drug}

{As a direct application of WCS, we consider the goal of  
prioritizing drug candidates. 
We focus on two tasks: 
(i) drug property prediction, i.e., 
selecting molecules that bind to a target protein, and 
(ii) drug-target interaction prediction, i.e., 
selecting drug-target pairs with high affinity scores.}  
We use  the
DeepPurpose library~\cite{huang2020deeppurpose} for data pre-processing and model training. 

\subsection{Drug property prediction}
\label{subsec:drug_pred}

Our first goal is to find molecules that may bind to a target protein for HIV. 
Machine learning models 
are  trained on a subset of screened molecules from 
a drug library,  
and  then  used to predict the remaining ones.

We use the HIV screening dataset in the DeepPurpose library
with a total size of $n_\textrm{tot}=41127$. 
In this dataset, the covariate $X\in \cX$ is 
a string that represents the chemical structure of a molecule 
(encoded by Extended-Connectivity FingerPrints, ECFP), and  
the response $Y\in \{0,1\}$ is binary, indicating whether 
the molecule binds to the target protein. 
Our goal is to select as many new drugs with $Y=1$ as possible while 
controlling the FDR below a specified level. 
This can be viewed as the goal~\eqref{eq:fdr} with $c_{n+j}\equiv 0.5$. 

Oftentimes, experimenters introduce a bias by selecting the first batch of 
molecules to screen (the training data), which results in a covariate shift 
between training (calibration) and test data. 
Here, we mimic an experimentation procedure that builds upon 
a pre-trained prediction model $\hat\mu \colon \cX\to \RR$ 
for binding activity, so that  
those with higher predicted values are more likely 
to be included in the training (calibration) data. 
In our experiment, to reduce computational time, 
we train a single model for both predicting test samples 
and for selecting the calibration fold. 
We take a subset of size $0.4 \times n_{\textrm{tot}}$ as the training set $\cD_\train$, 
on which we train a small neural network $\hat\mu(\cdot)$ 
with three layers trained over three epochs. 
The remaining $0.6 \times n_{\textrm{tot}}$ samples are randomly selected 
as the calibration set $\cD_\calib$ with probability 
$p(x) = \min\{0.8,  \sigma(\hat\mu(x)-\bar\mu)\}$; here, $x\in \cX$, 
$\bar\mu=\frac{1}{|\cD_\train|}\sum_{i\in \cD_{\train}} \hat\mu(X_i)$ 
is the average prediction on the training fold, and $\sigma(t) = e^t/(1+e^t)$ is the sigmoid function. 
The covariate shift~\eqref{eq:cov_shift} is thus of the form $w(x)\propto 1/p(x)$, 
which we assume is known.

We compare the BH procedure with weighted conformal p-values~\eqref{eq:def_wcpval_rand}, 
as well as our Algorithms~\ref{alg:bh} and~\ref{alg:bh_cond}. The last one 
is applicable because we here take $c=0.5$ for binary classification. We consider two scores used in BH and Algorithm~\ref{alg:bh}:
\begin{itemize}
    \item \texttt{res}:  $V(x,y) = y-\hat\mu(x)$. 
    \item \texttt{clip}: the score $V(x,y) = M\cdot \ind\{y>0\} -\hat\mu(x)$, with $M>2\sup_x|\hat\mu(x)|$.   
\end{itemize}
The empirical FDR over $N=200$ independent runs 
for FDR target levels $q\in \{0.1, 0.2,0.5\}$ is summarized 
in Figure~\ref{fig:drug_pred_fdr}. 
All algorithms empirically control the FDR below the nominal levels 
(up to random fluctuation), showing 
the reliability of WCS in 
prioritizing drug discovery under covariate shifts. 
The two scores yield similar power as in Figure~\ref{fig:drug_pred_power}.

\begin{figure}[h]
    \centering
    \includegraphics[width=5.5in]{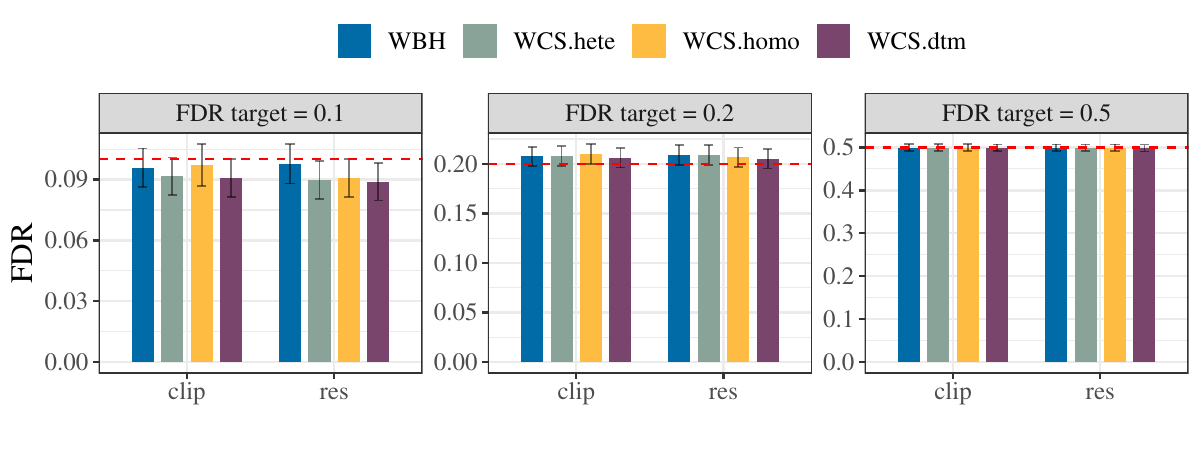}
    \caption{Empirical FDR for drug property prediction.
    The label \texttt{wBH} is a shorthand 
    for BH, and \texttt{wCC.*} for 
    Algorithm~\ref{alg:bh} (for \texttt{clip} and \texttt{res}) or 
    Algorithm~\ref{alg:bh_cond} (for \texttt{sub}) 
    with three pruning options $*\in \{\texttt{hete},\texttt{homo},\texttt{dtm}\}$. 
    The red dashed lines 
    are the nominal FDR levels.}
    \label{fig:drug_pred_fdr}
\end{figure}

\begin{figure}[h]
    \centering
    \includegraphics[width=5.5in]{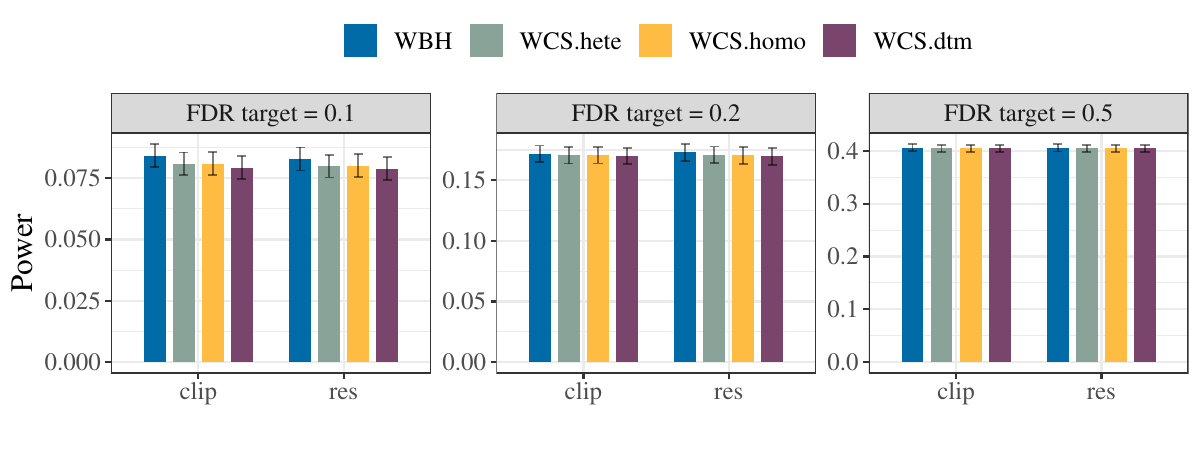}
    \caption{Empirical power for drug property prediction. 
    Everything else is as in Figure~\ref{fig:drug_pred_fdr}.}
    \label{fig:drug_pred_power}
\end{figure}

\subsection{Drug-target interaction prediction}
\label{subsec:dti}

We then consider drug-target interaction (DTI) prediction, 
where the goal is to select drug-target pairs with a high binding score. 
This task is relevant if a therapeutic company 
is interested in prioritizing resources for 
drug candidates that may be effective for any target they are interested in. 
We use the DAVIS dataset~\cite{davis2011comprehensive}, 
which records real-valued binding affinities for $n_{\textrm{tot}} = 30060$ 
drug-target pairs. 
The drugs and the targets 
are encoded into numeric features using
ECFP and Conjoint triad feature (CTF)~\cite{shen2007predicting,shao2009predicting}.

We essentially follow the same procedure as in Section~\ref{subsec:drug_pred} and only rehearse the main components. 
First, we randomly sample a subset of size $n_{\train} = 0.2\times n_{\textrm{tot}}$ 
as $\cD_\train$, on which we train a regression model $\hat\mu(\cdot)\colon \cX\to \RR$ 
using a $3$-layer neural network with $10$ training epochs. 
This  relatively 
simple model is suitable for numerical experiments on CPUs (one 
can surely use more complicated prediction models in practice). 
As before, we use the same model $\hat\mu(\cdot)$ 
for both predicting test samples and selecting calibration data into first-batch screening (in practice, they can of course be different). We sample a drug-target pair for inclusion in the calibration set  
with probability $p(x) = \sigma(2\hat\mu(x)-\bar\mu)$, 
where $\bar\mu$ is the average training prediction as before. 
Finally, among those not selected in the calibration data, 
we randomly sample a subset of size $m=5000$ as test data. 

We choose a more complicated threshold $c_j$ for the continuous response. For a drug-target pair $X_{n+j}$, 
we set $c_j$ to be the $q_{\textrm{pop}}$-th 
quantile of the binding scores of all drug-target pairs in $\cD_\train$ 
with the same target. We evaluate $q_{\textrm{pop}} \in\{0.7,0.8\}$. 
Thus, $c_j:=c(X_{n+j},\cD_\train,q_{\textrm{pop}})$ where 
$c$ is a mapping that takes both   $\cD_\train$ and 
the target information in $X_{n+j}$ as inputs.  
Lastly, we evaluate the BH procedure and Algorithm~\ref{alg:bh} with scores: 
\begin{itemize}
    \item \texttt{res}:  $V(x,y) = y-\hat\mu(x)$.
    \item \texttt{clip}: $V(x,y) = M\cdot \ind\{y>c(x,\cD_\train,q_{\textrm{pop}})\} 
    + c(x,\cD_\train,q_{\textrm{pop}}) \ind\{y\leq c(x,\cD_\train,q_{\textrm{pop}})\}-\hat\mu(x)$ in which $M = 100$.
\end{itemize}
We always use $\cD_\calib$ as the calibration set, 
and set the FDR target as $q\in\{0.1,0.2,0.5\}$.  

Figure~\ref{fig:drug_dti_fdr} shows false discovery proportions (FDPs) 
for $q_{\textrm{pop}}=0.8$ 
in $N=200$ independent. 
Similar results for $q_{\textrm{pop}}=0.7$ 
are presented in Appendix~\ref{app:subsec_dti}. We see that the FDR is controlled at the desired level for all algorithms and nonconformity scores. 
This shows the validity of our algorithms and the 
plausibility of the independence assumptions we make on the drug-target pairs. 
In this task, we do not observe 
much difference between deterministic pruning (\texttt{WCS.dtm}) and 
the other two pruning options. 
Furthermore, we observe that the FDPs across replications 
tightly concentrate especially for \texttt{clip} and $q\in\{0.2,0.5\}$, 
showing that our algorithms are stable with respect to data splitting (i.e., the 
randomness in choosing the initial screening sets and in the training process). 
Comparing the two nonconformity scores, 
we see that \texttt{clip} exploits the error 
budget and obtains a realized FDR, which is very close to the 
nominal level, while \texttt{res} has a much lower FDR in all cases. Not surprisingly,  Figure~\ref{fig:drug_dti_power} shows that  
\texttt{clip} has much higher power. 

\begin{figure}[h]
    \centering
    \includegraphics[width=5.5in]{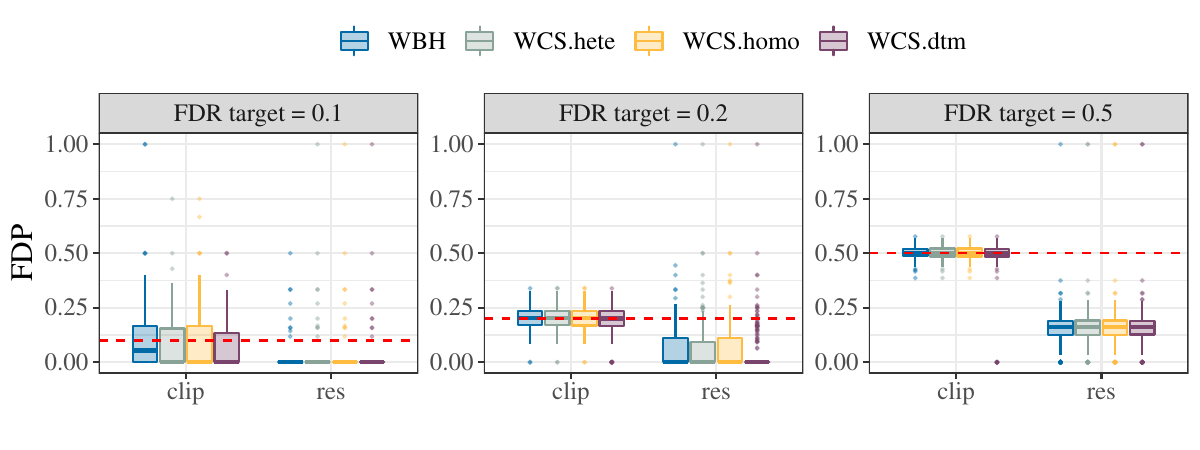}
    \caption{Empirical FDR for drug-target interaction prediction  with $q_{\textrm{pop}}=0.8$. The shorthand \texttt{WBH} 
    stands for BH procedure, and \texttt{WCS.*} for 
    Algorithm~\ref{alg:bh}  
    with three pruning options $*\in \{\texttt{hete},\texttt{homo},\texttt{dtm}\}$. 
    The red dashed lines 
    are the nominal FDR levels. Solid lines are empirical averages (here empirical FDR).}
    \label{fig:drug_dti_fdr}
\end{figure}

\begin{figure}[h]
    \centering
    \includegraphics[width=5.5in]{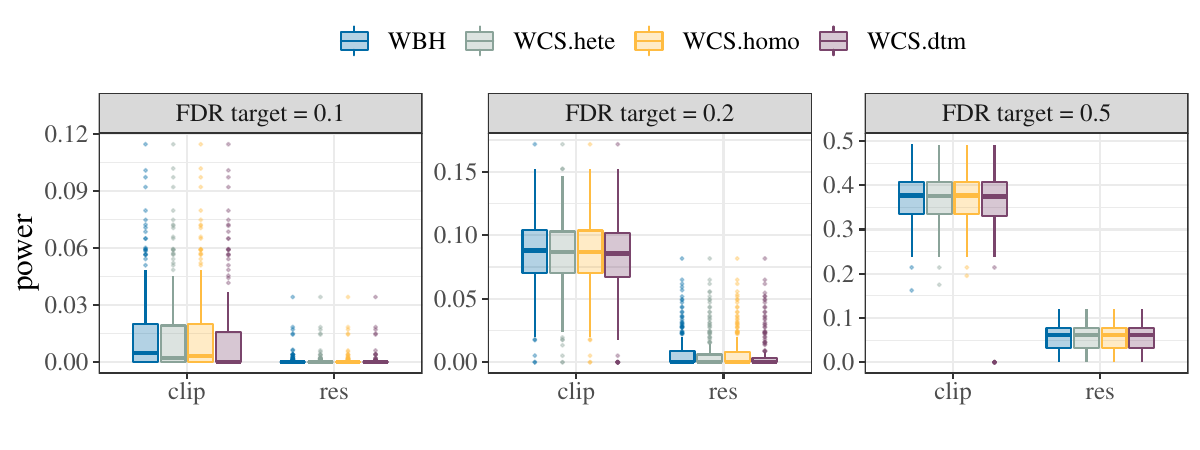}
    \caption{Empirical power for drug-target interaction prediction with $q_{\textrm{pop}}=0.8$. 
    Everything else is as in Figure~\ref{fig:drug_dti_fdr}.}
    \label{fig:drug_dti_power}
\end{figure}

\section{Application to multiple individual treatment effects}
\label{sec:ite}

{We now apply our method to infer individual treatment 
effects (ITEs)~\cite{lei2021conformal,jin2023sensitivity}
under the potential outcomes framework~\cite{imbens2015causal}.
The ITE describes the difference between an individual's outcomes when receiving a treatment versus not; its
variation comes from both 
individual characteristics and intrinsic uncertainty in the outcomes. 
Inference on multiple ITEs is useful in 
assisting reliable personalized decision making.} 

We consider a super-population setting, where  
$(X_i,O_i(1),O_i(0),T_i)$ 
are drawn i.i.d.~from a joint distribution $\pcausal$ (distinct 
from the distributions $P$ and $Q$ we use for the data). 
Here, $X_i\in \cX$ is the  features, 
$T_i\in \{0,1\}$ indicates whether unit $i$ receives treatment, 
and $O_i(1),O_i(0)$ are the potential outcomes under 
treatment and not, respectively. 
Under the standard SUTVA~\cite{imbens2015causal}, 
we observe $(X_i,O_i,T_i)$, where $O_i=O_i(T_i)=T_iO_i(1)+(1-T_i)O_i(0)$. 
We will specify how the treatment $T_i$ is allocated later on. 
We focus on  units \emph{in the study}, i.e., 
those who have received a certain treatment option. The ITE of unit $i$ is  
the random variable 
$\Delta_i = O_i(1)-O_i(0)$.  
As only one potential outcome is observed, 
a crucial part of inferring the 
ITE is to predict the counterfactual, i.e., 
$O_i(1)$ for those units with $T_i=0$, and $O_i(0)$ for $T_i=1$. 
In the following, we consider simultaneous inference for 
the ITEs of a set of units in the 
control group (inference for treated units is similar). 

Formally, the test samples are $\{ X_{n+j} \}_{j=1}^m$ 
are units in the control group (the outcome $O_{n+j}$ is observed). 
Our goal is to screen for positive ITEs with FDR control, i.e., 
finding a subset $\cR\subseteq\{1,\dots,m\}$, such that 
\#\label{eq:fdr_ite}
\fdr := \EE\Bigg[  \frac{\sum_{j=1}^m \ind\{O_{n+j}(1)\leq O_{n+j}(0),~j\in \cR\}}{1\vee |\cR|}   \Bigg]  \leq q.
\#  

To cast the counterfactual problem in our framework,  
set $Y_{n+j}:= O_{n+j}(1)$ to be the unobserved outcome, and the 
thresholds $c_{n+j} := O_{n+j}(0)$ to be the observed outcomes. 
To infer $Y_{n+j}$, the training data is 
$\{(X_i,Y_i )\}_{i=1}^n = \{(X_i,O_i(1))\}_{i=1}^n$ 
for which $T_i=1$.  
Conditional on the treatment status, the training samples are i.i.d.~from 
$P:=\pcausal_{X,O(1)\given T=1}$, 
while the test samples $\{(X_{n+j},Y_{n+j})\}_{j=1}^m$ are i.i.d.~from 
$Q:=\pcausal_{X,O(1)\given T=0}$. 
The relationship between $P$ and $Q$ will depend on the treatment assignment 
mechanism.

\subsection{Warm-up: completely randomized experiments}

As a warmup, consider a completely randomized experiment, where the 
treatment assignments $T_i$ are i.i.d.~draws from $\textrm{Bern}(\pi)$ 
for some $\pi\in(0,1)$ independently from everything else. 
In this case, $P =Q = \pcausal_{X,O(1)}$ and 
it suffices to use conformal p-values without weights. 
The 
testing procedure
for randomized experiments has already been studied in~\cite{jin2022selection}. 
Take any \emph{monotone} nonconformity score $V(\cdot,\cdot)$ obtained 
from an independent training process, 
and compute $V_i = V(X_i,O_i) = V(X_i,O_i(1))$ for $i=1,\dots,n$,  
and $\hat{V}_{n+j} =V(X_{n+j},O_{n+j})= V(X_{n+j},O_{n+j}(0))$ for $j=1,\dots,m$. 
Then construct conformal p-values~\eqref{eq:weighted_pval} with $w(\cdot)\equiv 1$, 
and run the BH procedure with  these p-values at level $q$.  
FDR control is a corollary of \cite[Theorem 2.3]{jin2022selection}, and we omit the proof.

\begin{corollary}\label{cor:ite_randomized}
The procedure above has FDR~\eqref{eq:fdr_ite} at most  $q$.
\end{corollary}

\subsection{Stratified randomization and observational studies}

We now consider a more general setting where the treatment assignment 
may depend on the observed covariates, formalized 
as the following strong ignorability condition~\cite{imbens2015causal}.

\begin{assumption}[Strong ignorability]\label{assump:ignor}
Under the joint distribution $\pcausal$, it holds that  
$(O(1),O(0)) \indep T \given X$. Equivalently, 
the treatment assignments are independently 
generated from $T_i\sim \textrm{Bern}(e(X_i))$, 
where $e\colon \cX\to (0,1)$ is known as the propensity score.
\end{assumption}

Assumption~\ref{assump:ignor} is automatically satisfied in 
stratified randomization experiments, where the treatment is randomized 
in a way that only depends on the covariates, and 
the propensity score $e(\cdot)$ is known.
In observational studies where the treatment assignment mechanism is completely unknown, 
Assumption~\ref{assump:ignor} is standard  
for the identifiability of average treatment effects~\cite{rosenbaum2002observational}. 
Inference on ITEs when this assumption is violated (i.e., when 
there is unmeasured confounding)
is studied in~\cite{jin2023sensitivity}; that said, 
multiple testing under confounding needs additional techniques, 
and is beyond the scope of 
the current work. 

Under Assumption~\ref{assump:ignor}, 
the covariate shift condition~\eqref{eq:cov_shift} holds~\cite{lei2021conformal}, as 
\$
\frac{\ud Q}{\ud P}(x,y) = \frac{\ud \pcausal_{X,O(1)\given T=0}}{\ud \pcausal_{X,O(0)\given T=1}}(x,y) = \frac{\ud \pcausal_{X\given T=0}}{\ud \pcausal_{X \given T=1}}(x) = w(x) := \frac{\pi(1-e(x))}{(1-\pi) e(x)},
\$
where $e(x)=\pcausal(T=1\given X=x)$ is the propensity score, and $\pi=\pcausal(T=1)$ is the 
marginal probability of being treated.  
As such, Algorithm~\ref{alg:bh} is readily applicable when $e(\cdot)$ is known. 

\begin{corollary}\label{cor:ite_weight}
Suppose $e(\cdot)$ is known, and Assumption~\ref{assump:ignor} holds. 
Consider calibration data $\{(X_i,O_i(1)\colon T_i=1\}_{i=1}^n$, 
test data $\{X_{n+j}\colon T_{n+j}=1\}_{j=1}^m$, 
thresholds $\{O_{n+j}(0)\colon T_{n+j}=0\}_{j=1}^m$, 
any monotone score $V$, 
and weight function $w(x) \propto (1-e(x))/{e(x)}$ as the input 
of Algorithm~\ref{alg:bh}. 
Then any selection procedure $\cR \in \{\cR_{\hete}, \cR_{\homo}, \cR_{\dtm}\}$ has FDR at most $q$.
\end{corollary} 

The propensity score function $e(\cdot)$ is unknown for observational 
data. Under Assumption~\ref{assump:ignor}, 
we can estimate the propensity scores (hence the weight function) 
using an independent training fold 
and plug this estimate into the construction of p-values. 
As a corollary of Theorem~\ref{thm:est_w}, we can develop an 
FDR bound for observational data. 

\begin{corollary}
    In the setting of Corollary~\ref{cor:ite_weight}, take  
$w(x) \propto (1-\hat e(x))/{\hat e(x)}$ as the input 
of Algorithm~\ref{alg:bh}, where 
$\hat{e}(\cdot)$ is an estimate of $e(\cdot)$ which is independent of 
the calibration and training data. Then any selection procedure $\cR \in \{\cR_{\hete}, \cR_{\homo}, \cR_{\dtm}\}$ obeys
\$
\fdr \leq  q\cdot \EE\bigg[ \frac{\hat\gamma^2}{1+ q(\hat\gamma^2-1)/m} \bigg],
\qquad 
\hat\gamma := \sup_{x\in\cX}\max\bigg\{ \frac{1-e(x)\hat{e}(x)}{e(x)(1-\hat{e}(x))},
\frac{e(x)(1-\hat{e}(x))}{1-e(x)\hat{e}(x)}  \bigg\}.
\$
\end{corollary}

At a high level, the conformal p-values  
compare the observed outcome $O_{n+j}(0)$ of the control units 
to the empirical distribution of $O_{i}(1)$ in the calibration data;  
the latter---after proper weighting---shows the typical behavior 
of the counterfactual 
$O_{n+j}(1)$, and provides evidence for whether 
it may be larger than $O_{n+j}(0)$, i.e., whether $\Delta_{n+j}$ is positive.  

In predicting $O_{n+j}(1)$, we use the marginal distribution of $(X,O(1))$ from the calibration data, but ignore the information in the observed outcome $O(0)$.
Put differently, our 
inference for ITEs is valid regardless of how 
the potential outcomes are coupled. 
Thus,  
the actual false discovery rate may be lower than $q$.  
Indeed, it is observed in~\cite{jin2023sensitivity} that the 
FDR for identifying positive ITEs---even without adjusting for multiplicity---can sometimes be lower than the nominal level.   
Next, we are to empirically investigate 
the FDR and power of our method under various couplings of potential outcomes. 
A theoretical understanding is left for future work. 

\subsection{Simulation studies}

We design joint distributions of the variables $(X,O(1),O(0),T)$  
with various covariate distributions, regression functions, 
and couplings of potential outcomes.  
The covariates $X_i \in \mathbb{R}^{10}$ ($p = 10$) 
are obtained via $X_{ij}=\Phi(X_{ij}^0)$, $j=1,\dots,10$,
where $X_{i}^0\in \RR^{10}$ are i.i.d.~$\mathcal{N}(0,\Sigma)$, 
and $\Phi(\cdot)$ is the CDF of a $\mathcal{N}(0,1)$ random variable.
We set  $\Sigma=I_{p}$ for the independent case and $\Sigma_{k,j}=0.9^{|k-j|}$ 
for the correlated case. \footnote{Simulations and real data analysis for ITE inference 
can be found at \href{https://github.com/ying531/conformal-selection}{https://github.com/ying531/conformal-selection}.}    
We
consider three cases:
\begin{itemize}
    \item Setting 1: 
    $
    O_i(0) = 0.1\, \epsilon_{0,i}$, $O_i(1) = \max\{ 0, \mu_1(X_i) + \sigma_1(X_i) \, \epsilon_{1,i} \}$, where 
    $\sigma_1 (x) =   0.2 -  \log x_1$ and $ \mu_1(x) = {4}/{(1+e^{-12x_1-0.5})(1+e^{-12x_2-0.5})}$. 
    \item Setting 2 is a mixture of deterministic and stochastic ITEs. 
    With probability $0.1$, set $O_i(1)= 0.1\, \epsilon_{0,i}-0.5$, and $O_i(0)=O_i(1)+0.05$; otherwise, set them as in Setting 1.
    \item Setting 3: $O_i(0) = \mu_0(X_i) + 0.1 \, \epsilon_{0,i}$,   
    $O_i(1)=\max\{0, \mu_1(X_i) + \sigma_1(X_i) \, \epsilon_{1,0}\}$, where $\sigma_1(\cdot)$ is as in Setting 1, and
    $
    \mu_0(x) = {2}/{(1+e^{-3x_1-0.5})(1+e^{-3x_2-0.5})}$, $
    \mu_1(x)=0.1+1.5\, \mu_0(x)$. 
\end{itemize}
Above, the variables ($\epsilon_{0,i},\epsilon_{1,i}$) are i.i.d.~with $\mathcal{N}(0,1)$ as marginals.
We consider three coupling scenarios: 
(i) independent, in which $\epsilon_{0,i}$ and $\epsilon_{1,i}$  are  independent; 
(ii) negative, in which $\epsilon_{1,i}=-\epsilon_{0,i}$; 
(iii) positive, in which  $\epsilon_{1,i} =\epsilon_{0,i}$. 
We generate treatment indicators via 
$T_i\sim \textrm{Bern}(e(X_i))$ independently, 
where $e(x)= (1+\textrm{Beta}_{2,4}(x_1))/{4}$, 
and $\textrm{Beta}_{2,4}(\cdot)$ is the CDF
of the Beta distribution with shape parameters $(2,4)$. 
Finally, we set those $T_i=1$ as the calibration data ($n=250$), 
and those $T_i=0$ as the test sample ($m=100$); 
as training data we select $n_\train=750$ units, 
which are used to construct four  conformity scores:
\begin{itemize}
    \item \texttt{reg}: $V(x,y) = y-\hat\mu(x)$, where $\hat\mu(x)$ is an estimate of  $\mu_1(x)$ using regression forest from the \texttt{grf} R-package~\cite{athey2019generalized}.
    \item \texttt{cdf}: $V(x,y) = \hat{F}(x,y)$ \cite{chernozhukov2021distributional}, where $\hat{F}(x,y)$ is an estimate of  $\pcausal(O(1)\leq y\given X=x)$. Here, we fit $\hat{F}$ via 
    inverting quantile regression forests from  \texttt{grf}.
    \item \texttt{oracle}: $V(x,y) = \pcausal(O(1)\leq y\given X=x)$, which is not computable from data and is shown only for illustration. 
    \item \texttt{cqr}: $V(x,y) =  y -\hat{q}_\beta(x)$~\cite{romano2019conformalized}, where 
    $\hat{q}_\beta(x)$ is an estimate of the $\beta$-th quantile of $\pcausal_{O(1)\given X=x}$, $\beta\in\{0.2,0.5,0.8\}$, from the quantile regression forests in \texttt{grf}. 
\end{itemize}

After generating the three folds of data, we fit  $V(\cdot,\cdot)$ on 
the training fold, and run 
four procedures on the calibration and training folds: 
BH with weighted conformal p-values in~\eqref{eq:def_wcpval_rand} (named \texttt{WBH} 
in the plot), 
and  Algorithm~\ref{alg:bh} with three pruning options 
(named \texttt{WCS.hete}, \texttt{WCS.homo}, \texttt{WCS.dtm}, respectively).   
The experiment is repeated for $N=1000$ independent runs. 

Among the three, setting 2 is the most challenging:  
$10\%$ of the samples are with slightly negative ITEs; 
they are likely to be falsely rejected 
since the selection procedure ignores the coupling 
(it may reject a sample with a small value of $O(0)$ 
even if $O(1)$ is smaller). 
We expect setting 1 to be the least challenging 
as there is a strong signal in $O(1)$. 
Setting 3 is more challenging than setting 1 since $O(1)$ and $O(0)$ are close 
to each other as the regression functions are similar.  

\begin{figure}
\centering
\includegraphics[width=6.5in]{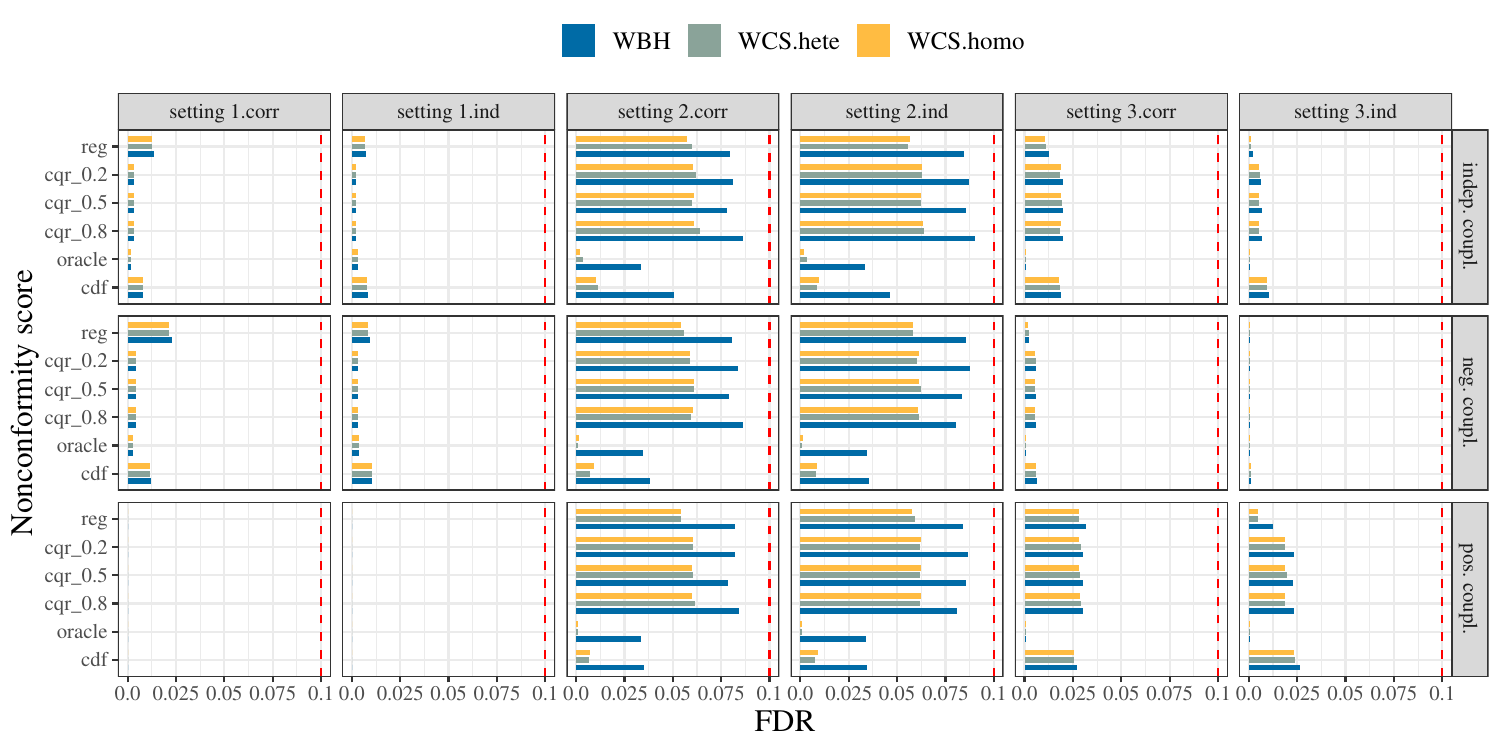} 
\caption{Empirical FDR for finding positive ITEs. 
The columns represent three settings with either independent (\texttt{ind}) 
or correlated (\texttt{corr}) covariates. Each row corresponds to a 
coupling of potential outcomes (positive, negative, and no coupling). 
In each subplot, the $y$-axis is the nonconformity score, and the $x$-axis is 
the empirical FDR. 
The red dashed lines indicate the nominal level $q=0.10$.}
\label{fig:ite_simu_fdr} 
\end{figure}

\paragraph{FDR control.} 
The empirical FDR in all settings is plotted in Figure~\ref{fig:ite_simu_fdr}. 
We observe FDR control for all methods, including the BH procedure. 
(We omit 
the results of $\cR_{\dtm}$ because it does not succeed in making any selection.)
Across settings with the same regression function, 
those with correlated covariates often have higher FDR.
The realized FDR also varies with the scores:  
\texttt{oracle} incurs very low FDR,  
while its empirical counterpart \texttt{cdf} has higher FDR. 
On the other hand, quantile regression based scores are robust to the quantile level $\beta$ in the sense that \texttt{cqr} achieves similar FDR 
for all choices of $\beta$. 

The impact of coupling on FDR is less consistent across settings. The positive coupling from setting 1 yields a low FDR, perhaps because it is difficult to obtain sufficiently strong evidence 
which leads to fewer rejections (see the power analysis below). 
The three couplings yield similar FDRs for setting 2. 
For setting 3, positive coupling leads to the highest FDR.


 
\paragraph{Power.} The empirical 
power is defined as 
$
\textrm{Power} := 
\EE\Big[\frac{\sum_{j=1}^m\ind\{j\in \cR,O_{n+j}(1)>O_{n+j}(0)\}}{\sum_{j=1}^m \ind\{O_{n+j}(1)>O_{n+j}(0)\}}\Big].
$ 
As we see in  Figure~\ref{fig:ite_simu_power}, 
the power is in general higher when the covariates are correlated.  The ITE depends on $(X_1,X_2)$, and in our design, higher correlation between 
the first two entries leads to higher power. 
The power also varies with the nonconformity scores. 
The regression-based \texttt{reg} leads to lower power than 
other methods. This means that regression functions may be underpowered
in capturing the essential information for contrasting 
the distributions of two potential outcomes. 
Quantile-regression based methods (\texttt{cqr}) achieve 
similar power for different choices of $\beta$ in all settings. 
Finally, fitting the conditional distribution function 
(\texttt{cdf}) is the most powerful option in settings 1 and 3, 
but is far less powerful in setting 2. 
Its oracle version (\texttt{oracle}) is also surprisingly less powerful; we observe that the oracle cdf cannot create sufficiently small p-values, while 
the estimated one is able to do so due to fluctuations caused by estimation uncertainty. 

\begin{figure} 
\centering
\includegraphics[width=6.5in]{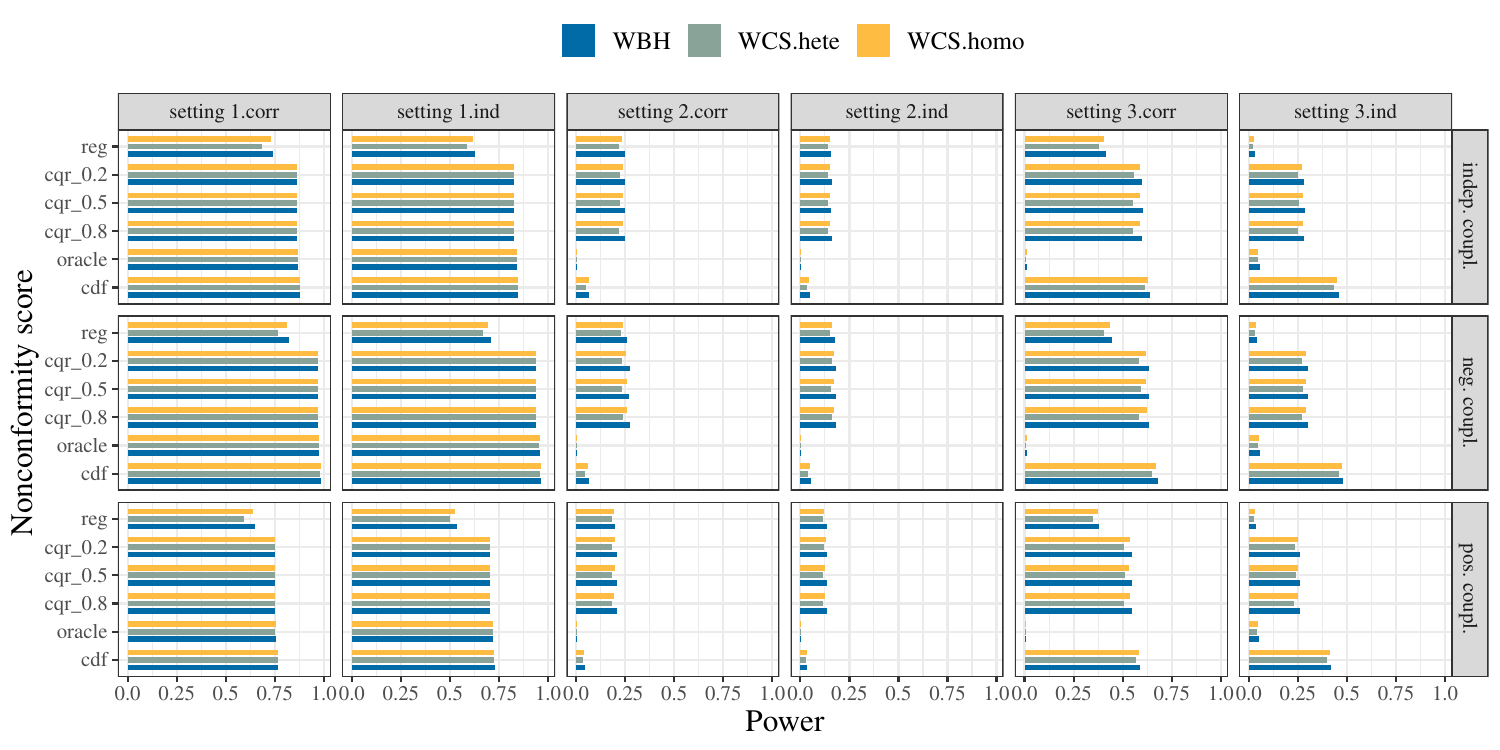} 
\caption{Empirical power for finding positive ITEs. Everything else is  as in Figure~\ref{fig:ite_simu_fdr}.}\label{fig:ite_simu_power}
\end{figure}

The impact of coupling on the power is consistent across settings. 
In general, negative coupling leads to the highest power, while 
positive coupling leads to the lowest. 
This is because the conformal p-values compare the 
observed outcome to the \emph{marginal} distribution of the counterfactuals, 
without accounting for their joint distribution. 
Thus, when $O_{n+j}(0)$ is extremely small, i.e., when we see a small p-value and 
 $j\in \cR$, 
under negative coupling, it is more probable that $O_{n+j}(1)$ 
is relatively large, hence leading to a true discovery. 
On the contrary, under positive coupling, 
a large $O_{n+j}(1)$ often corresponds to a relatively large $O_{n+j}(0)$ 
and hence a large conformal p-value, which may not be selected. 

Finally, $\cR_{\dtm}$ makes no selection. 
This may be due to the threshold effect:  
since the selection requires $p_j\leq q|\hat\cR_{j\to 0}|/m$ 
for test samples $j$ with the smallest $|\hat\cR_{j\to 0}|$, the pruning step may 
exclude too many candidates.  We  recommend using 
$\cR_{\hete}$ and $\cR_{\homo}$ in practice.

\vspace{-1em}
\paragraph{Stability.} As aforementioned in Section~\ref{sec:method}, 
our method introduces 
extra randomness. It is shown in Proposition~\ref{prop:asymp_equiv} 
that this becomes asymptotically negligible, 
so that our selection set is asymptotically the same as that the BH procedure yields. 
We empirically evaluate the discrepancy in the selections, namely, $|\cR\Delta \cR_{\bh}|/|\cR_\bh|$, 
for $\cR$ being either $\cR_{\homo}$ or $\cR_{\hete}$. The results are shown in 
Figure~\ref{fig:ite_simu_dif}, confirming our theory. 
Though we were not able to prove the asymptotic equivalence of $\cR_{\hete}$ and 
$\cR_{\bh}$  
 in the $m,n\to \infty$ 
regime,  
the discrepancy is still small (despite being 
larger than that of $\cR_{\homo}$). All in all, 
the impact of additional randomness seems acceptable in practice. 

\begin{figure}[htbp]
    \centering
    \includegraphics[width=6.5in]{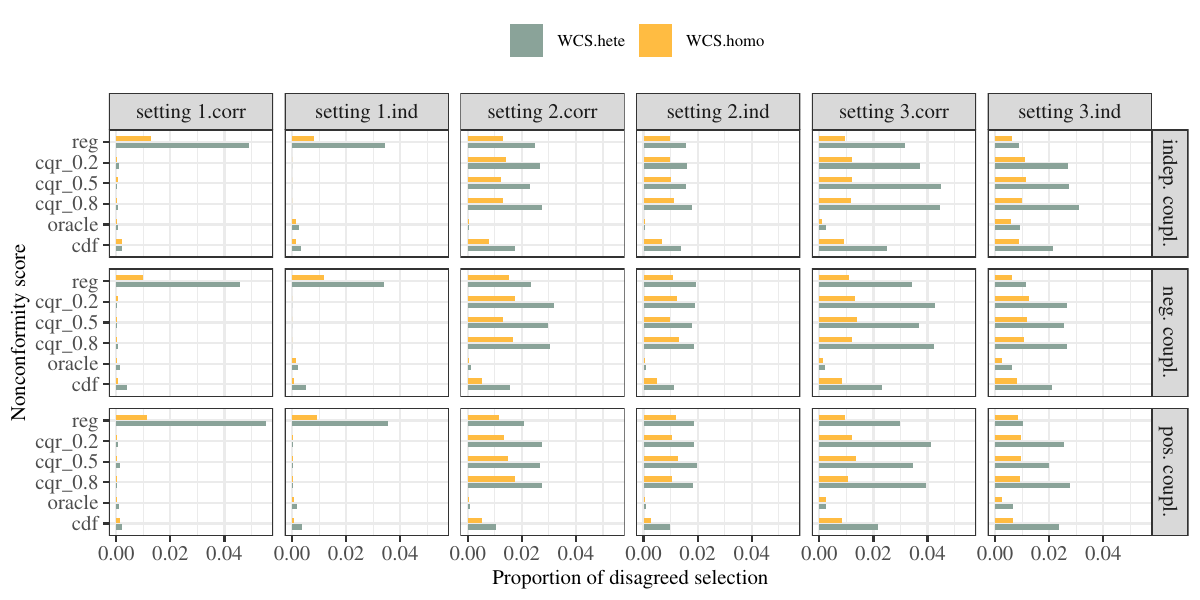} 
    \caption{Discrepancy between selection  rules $|\cR\Delta \cR_{\bh}|/|\cR_\bh|$. Everything else is 
    as in Figure~\ref{fig:ite_simu_fdr}.}\label{fig:ite_simu_dif}
\end{figure}

The discrepancy  for $\cR_{\dtm}$ from $\cR_{\bh}$ is large in general is $\cR_{\dtm}$ is usually an empty set, and hence we did not plot it. 
We conjecture that the stability of $\cR_{\dtm}$ 
claimed in case (i) from Proposition~\ref{prop:asymp_equiv}
may not necessarily apply to the moderately large sample sizes $m=100$ and $n=250$, 
and  $\cR_{\dtm}$ may be 
far from $\cR_{\bh}$ for such sample sizes. 
In such cases, we recommend using $\cR_\homo$ or 
$\cR_{\hete}$ for better stability and power.

\subsection{Real data analysis}

We revisit the NSLM observational dataset from~\cite{carvalho2019assessing} 
and use our method to detect multiple positive individual treatment effects. 
It is a semi-synthetic observational dataset
curated from a real randomized experiment, 
also analyzed in \cite{lei2021conformal,jin2023sensitivity}.

Figure~2 in~\cite{jin2023sensitivity} plots 
the $\Gamma$-value, a measure of robustness of positive ITEs against 
unmeasured confounding, versus individual covariates. 
While units with larger $\Gamma$-values are of natural interest, 
properties of inference hold on average, instead of 
over \emph{selected} units (e.g.~those with high $\Gamma$-values). 
We are to produce a similar plot (Figure~\ref{fig:ite_pval}) 
for detecting positive ITEs, 
with guarantees over units that exhibit the strongest evidence. 
(We do not consider the confounding issue, 
which may need
nontrivial extension of our techniques.)

We randomly subsample three disjoint folds of size $5000$, $1000$, 
and $391$ from the original dataset. 
The first is the training fold, and consists of both 
treated and control units. 
The calibration data consists of all the $n = 997$ treated units 
in the second fold. 
The test data  
consists of all the $m = 256$ control units
in the last fold. 
To deploy our  procedure, we first train a propensity score model $\hat{e}(\cdot)$ using a 
regression forest from the \texttt{grf} R-package, and set  
$\hat{w}(x) = (1-\hat{e}(x))/{\hat{e}(x)}$. 
We then use a quantile regression forest from \texttt{grf} 
on the training data to fit $\hat{q}_{0.5}(x)$, 
the conditional median of $O(1) \mid X$, 
and set  
$V(x,y) = y - \hat{q}_{0.5}(x)$. 
Finally, we apply Algorithm~\ref{alg:bh} as 
in Theorem~\ref{thm:calib_ite} to test for positive ITEs. 

Figure~\ref{fig:ite_cdf} plots the 
empirical CDF of the 
weighted conformal p-values~\eqref{eq:weighted_pval} 
computed with $c_{n+j} := O_{n+j}(0)$ on the test samples. 
These p-values are stochastically smaller than Unif([0,1]), 
showing evidence for positive treatment effects in general. 
Yet, to identify positive ITEs with rigorous error control 
we must apply multiple testing ideas here. 

\begin{figure}[h]
    \centering 
    \includegraphics[width=2.5in]{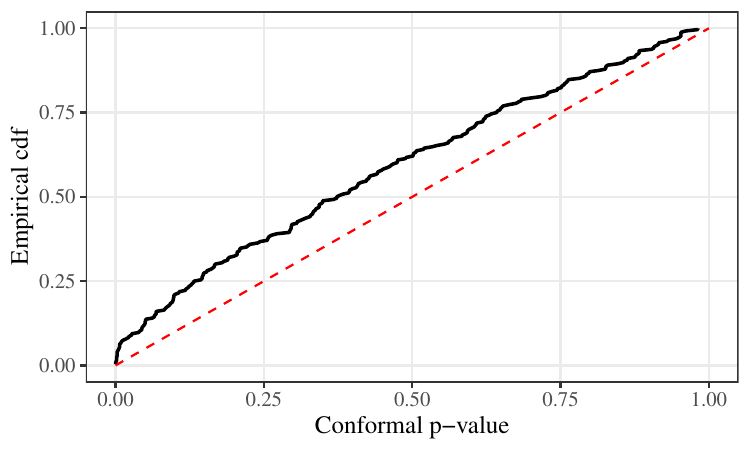}
    \caption{Empirical CDF~of the weighted conformal p-values for 
    detecting positive ITEs.}
\label{fig:ite_cdf}
\end{figure}

We observe $\cR_{\bh} = \cR_{\homo} = \cR_{\hete}$
while $\cR_{\dtm}=\varnothing$ at  FDR levels $q\in\{0.1,0.2,0.5\}$. 
Figure~\ref{fig:ite_pval}
plots the weighted conformal p-values 
versus school achievement levels of test units 
(each dot represents a test unit) with red dots (from light to dark) 
identified as positive ITEs at various 
FDR levels. 
Students in schools with moderate 
achievement levels demonstrate the strongest evidence of benefiting from 
the treatment. 

\begin{figure}[H]
\centering 
\includegraphics[width=6in]{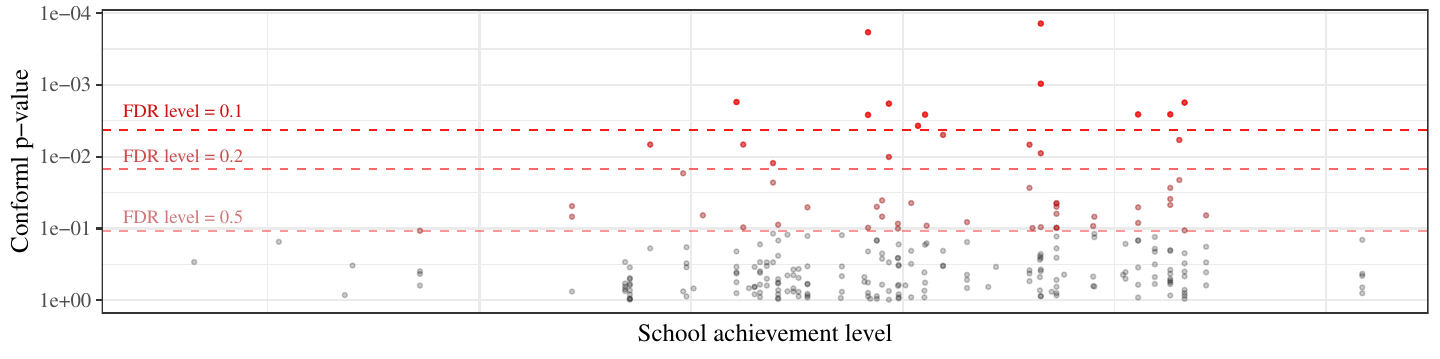}
\caption{P-values and rejection sets for detecting positive ITEs among control units.}
\label{fig:ite_pval}
\end{figure}

\section{Application to outlier detection}
\label{sec:outlier}

Finally, we study the extension of our framework 
for outlier detection, where the calibration inliers may follow 
a distinct distribution compared with test inliers. We discuss a new setup and apply 
a variant of Algorithm~\ref{alg:bh} to bank marketing data.

\subsection{Hypothesis-conditional FDR control}
\label{subsec:hypo_cond_fdr}

Assuming access to i.i.d.~training data $\{  Z_i  \}_{i=1}^n $, 
we consider a set of test data $\{ Z_{n+j}  \}_{j=1}^m $ 
for which the $Z_{n+j}$'s may only be partially observed (e.g., if $Z=(X,Y)$ 
we would observe the features $X$ but not the response $Y$).  
We are interested in some null hypotheses $\{H_j \}_{j=1}^m $ 
associated with the test samples. 
As before, whether $H_j$ is true 
can be a random event depending on $Z_{n+j}$; this includes our previous problem 
with $H_j\colon Y_{n+j}\leq c_{n+j}$.

\subsubsection{Hypothesis-conditional covariate shift}
\label{subsec:outlier}

\begin{assumption}[Hypothesis-conditional covariate shift] 
  \label{assump:label_conditional}
$\{Z_i\}_{i=1}^n\cup\{Z_{n+j}\}_{j=1}^m$ are mutually independent. 
Also, conditional on the subset $\cH_0\subset\{1,\dots,m\}$ 
of all null hypotheses, it holds that $\{Z_i\}_{i=1}^n\iid P$, 
and $\{Z_{n+j}\}_{j\in \cH_0} \iid Q$, where $ \ud Q/\ud P (Z)=w(X)$ 
for some function $w\colon \cX\to \RR^+$ and $X\subseteq Z$. 
\end{assumption}

A special case of 
problem~\eqref{eq:fdr}  
obeys these assumptions. 

\begin{example}[Binary classification]\normalfont 
\label{ex:binary}
Set $Z=(X,Y)$. where $Y\in\{0,1\}$ is a binary response. 
Suppose the goal is to find positive $Y_{n+j}$, 
e.g.~an active drug or a qualified candidate. We only observe the covariates $\{X_{n+j}\}_{j=1}^m$ for the test samples
$\{(X_{n+j},Y_{n+j})\}_{j=1}^m\iid Q$. 
Consider a reference dataset 
that only preserves positive samples among a set of i.i.d.~data 
from a covariate shifted distribution $P$ ; that is,  
$\cD_\calib = \{Z_i\colon Y_{i}=0\}$ where $\{(X_i,Y_i)\}\iid P$. 
Although the (super-population) covariate shift~\eqref{eq:cov_shift} no longer applies, 
conditional on $\cH_0 = \{j\colon Y_{n+j}=0\}$,  
it still holds that $Z_i\iid P_{Z\given Y=0}$ for $Z_i\in \cD_\calib$, and 
$Z_{n+j}\iid Q_{Z\given Y=0}$ for $j\in \cH_0$.   
\end{example}

The above example also applies to 
candidate screening 
with constant thresholds 
$c_{n+j} \equiv c_0$.  
Setting $\tilde{Y} = \ind\{Y > c_0\}\in\{0,1\}$ 
and $\cD_{\calib} = \{Z_i\colon Y_i\leq c_0\}$, all 
arguments apply similarly to $\tilde{Z}=(X,\tilde{Y})$. 
However, this does not necessarily apply when the  thresholds $c_{n+j}$ are random variables (especially when no 'threshold $c_j$' is observed in the calibration data, such as in the 
counterfactual inference problem we will study later).

Another application is outlier detection under covariate shifts. 

\begin{example}[Outlier detection]\normalfont
We revisit outlier detection~\cite{bates2021testing} 
while allowing for identifiable covariate shifts between the calibration inliers 
and the test inliers. 
Given $\{Z_i\}_{i=1}^n$ drawn i.i.d.~from 
an unknown distribution $P$ 
and a set of test data $\{Z_{n+j}\}_{j=1}^m$, 
we assume the inliers in 
the test data are i.i.d.~from a distribution $Q$ 
with $\ud Q/\ud P(Z)=w(Z)$ for a known function $w$, while allowing 
outliers to be from arbitrary distributions. 
The covariate shift may happen, for example, when 
inliers \emph{were} from $Q$ but 
the calibration set is selected with preferences 
relying on $z$: for instance, 
one may include more female users to balance the gender distribution
when curating 
a reference panel of normal transactions (inliers). 
In this case, $\cH_0 = \{j\colon Z_{n+j}\sim Q\}$ 
is a deterministic set, and Assumption~\ref{assump:label_conditional} clearly holds. 
\end{example}

The outlier detection example is closely related to 
identifying concept drifts. 

\begin{example}[Concept shift detection]\normalfont
Letting $Z=(X,Y)$ where $X\in \cX$ is the family of covariates and 
$Y$ is the response, concept shift 
focuses on potential changes in the conditional distribution of $Y$ given $X$. 
Given calibration data $\{Z_i\}_{i=1}^n\iid P$ 
and independent test data $\{Z_{n+j}\}_{j=1}^m$,  
\cite{hu2020distribution} 
assume $\{Z_{n+j}\}_{j=1}^m \iid Q$,  
and test for the global null $H_0\colon \ud Q/\ud P(Z) = w(X)$ 
for some $w\colon \cX\to \RR^+$. They achieve this by combining independent p-values 
after sample splitting.  
Our framework can be used to test 
individual concept drifts with dependent p-values. For instance, 
assume $\{X_{n+j}\}_{j=1}^m\iid Q_X$ for some unknown (but estimable) $Q_X$, 
we can test whether 
$P_{Y_{n+j}\given X_{n+j}} = P_{Y\given X}$. 
The null hypotheses can be formulated as 
$H_j\colon Z_{n+j}\sim Q$, where $\ud Q/\ud P(Z)=w(X)$ 
for some  $w\colon \cX\to \RR^+$ that is either known 
or can be estimated well under proper conditions.  
\end{example}

\subsubsection{Multiple testing procedure}

Our procedure for outlier detection under covariate shift is detailed 
in Algorithm~\ref{alg:bh_cond}. This 
slightly modifies Algorithm~\ref{alg:bh} 
by removing the thresholds (note differences in lines 1, 2, and 4). 
In the classification or constant threshold problem 
(Example~\ref{ex:binary}), it suffices to set $Z=X$ and leave out $Y$. 

\begin{algorithm}[h]
  \caption{Hypothesis-conditional Weighted Conformalized Selection}\label{alg:bh_cond}
  \begin{algorithmic}[1]
  \REQUIRE Calibration data $\{Z_i\}_{i=1}^n$, 
  test data  $\{Z_{n+j}\}_{j=1}^m$,  
  weight function $w(\cdot)$,
  FDR target $q\in(0,1)$, monotone nonconformity score $V\colon \cX\times\cY\to \RR$, 
  pruning method $\in\{\texttt{hete}, \texttt{homo}, \texttt{dtm}\}$.
  \vspace{0.05in} 
  \STATE Compute $V_i = V(Z_i)$ for $i=1,\dots,n$,  
  and $ {V}_{n+j}= V(Z_{n+j})$ for $j=1,\dots,m$.
  \STATE Construct weighted conformal p-values $\{ p_j\}_{j=1}^m$ as in~\eqref{eq:weighted_pval} with $\hat{V}_{n+j}$ replaced by $V_{n+j}$. 

  \vspace{0.3em}
  \noindent \texttt{- First-step selection -}
  \FOR{$j=1,\dots,m$}
  \STATE Compute p-values $\{ {p}_\ell^{(j)}\}$ as in~\eqref{eq:mod_pval} with $\hat{V}_{n+\ell}$ replaced by $V_{n+\ell}$ for all $\ell=1,\dots,m$.
  \STATE (BH procedure) Compute $k^*_j = \max\big\{k\colon 1 +\sum_{\ell\neq j} \ind\{{p}_\ell^{(j)}\leq qk/m\}\geq k\big\}$. 
  \STATE Compute $\hat{\cR}_{j\to 0} = \{j\}\cup\{\ell \neq j\colon  {p}_\ell^{(j)}\leq q k^*_j /m\}$.
  \ENDFOR
  \STATE Compute the first-step selection set $\cR^{(1)} = \{j\colon  {p}_j \leq q|\hat\cR_{j\to 0}|/m\}$.
  
  \vspace{0.3em}
  \noindent \texttt{- Second-step selection -}
  \STATE Compute $\cR = \cR_{\textrm{hete}}$  
  or $\cR = \cR_{\textrm{homo}}$  
  or $\cR = \cR_{\textrm{dtm}}$ as in Algorithm~\ref{alg:bh}. 
  \vspace{0.05in}
  \ENSURE Selection set $\cR$.
  \end{algorithmic}
\end{algorithm}

Algorithm~\ref{alg:bh_cond} returns to 
the conventional perspective, where the null set is 
deterministic, and the null p-values 
are dominated by Unif$([0,1])$. That is, for $p_j$ constructed 
in Line 2 of Algorithm~\ref{alg:bh_cond}, it holds that 
\$
\PP(p_j \leq t \given j\in \cH_0 ) \leq t\quad \textrm{for all }t\in[0,1]. 
\$
After conditioning on $\cH_0$, 
we no longer need to deal with the randomness of 
the hypotheses  and their interaction with the p-values. 
The only issue is the mutual dependence among the p-values, 
which is addressed using a similar idea as in our theoretical analysis 
of Algorithm~\ref{alg:bh}. 

Using calibration data obeying the covariate shift assumption, Algorithm~\ref{alg:bh_cond} achieves 
a slightly stronger  hypotheses-conditional FDR control. 
The proof of Theorem~\ref{thm:fdr_cond} is in Appendix~\ref{app:thm_outlier}. 
\begin{theorem}\label{thm:fdr_cond}
  Under Assumption~\ref{assump:label_conditional}, Algorithm~\ref{alg:bh_cond} yields 
  \$
\EE\bigg[  \frac{ |\cR\cap \cH_0| }{1\vee |\cR|} \bigggiven \cH_0 \bigg]   \leq q\cdot \frac{|\cH_0|}{m}
\$
for any fixed $q\in(0,1)$, and each  $\cR \in \{\cR_{\homo},\cR_{\hete}, \cR_{\dtm}\}$.
\end{theorem}

\subsubsection{Comparison with Algorithm~\ref{alg:bh}}
\label{subsubsec:compare}

In binary classification, or more generally, 
WCS with a constant threshold, 
we have shown in Example~\ref{ex:binary} that 
Algorithm~\ref{alg:bh_cond} yields FDR control.  
In this case, Algorithms~\ref{alg:bh} and \ref{alg:bh_cond} differ 
in terms of (i) power, and (ii) distributional assumptions. 
We elaborate on these distinctions.

First, Algorithm~\ref{alg:bh_cond}
only uses a subset of calibration data to construct p-values, 
which leads to a power loss 
for specific choices of nonconformity scores;  
see~\cite[Appendix A.1]{jin2022selection} for the i.i.d.~case.  
Let us consider the binary setting. 
Suppose we have access to a set of calibration data $\{(X_i,Y_i)\}$ 
consisting of both $Y=1$ and $Y=0$ samples. 
As discussed in Example~\ref{ex:binary}, 
Assumption~\ref{alg:bh_cond} holds if we only use data in 
the subset 
$\cI_0 = \{i\colon Y_i=0\}$ as $\cD_\calib$ in Algorithm~\ref{alg:bh_cond}. In contrast,  
Algorithm~\ref{alg:bh} uses all data points.  
Suppose we set $V(x,y) = My - \hat\mu(x)$ 
and $c_{n+j} \equiv 0$ in Algorithm~\ref{alg:bh}, 
where  $\hat\mu(\cdot)$ is  a fitted point prediction, and 
$M>2\sup_{x\in\cX}|\hat\mu(x)|$ is a sufficiently  large constant. 
Similarly, we set $V(x)=M-\hat\mu(x)$ in Algorithm~\ref{alg:bh_cond}. 
This construction ensures  
\#\label{eq:monotone_V}
\inf_{x\in \cX} V(x,1) = M-\sup_{x\in \cX}\hat\mu(x) 
> \sup_{x\in \cX}|\hat\mu(x)| \geq \sup_{x\in \cX} V(x,0).
\#
Theorems~\ref{thm:calib_ite} 
and~\ref{thm:fdr_cond} state that the FDR is controlled 
for both approaches. 
However, letting $p_j$ denote the p-values 
constructed in Algorithm~\ref{alg:bh} and $p_j'$ denote those in 
Algorithm~\ref{alg:bh_cond}, we note that 
\$
p_j &= 
\frac{\sum_{i\in \cI_0} w(X_i)\ind{\{V(X_i,0) < V(X_{n+j},0) \}} 
+  w(X_{n+j}) }{\sum_{i=1}^n w(X_i) + w(X_{n+j})} 
+ \frac{ \sum_{i\in \cI_1} w(X_i)\ind{\{V(X_i,1) < V(X_{n+j},0) \}}  }{\sum_{i=1}^n w(X_i) + w(X_{n+j})}  \\
&=\frac{\sum_{i\in \cI_0} w(X_i)\ind{\{V(X_i,0) < V(X_{n+j},0) \}}   +  w(X_{n+j}) }{\sum_{i=1}^n w(X_i) + w(X_{n+j})} \\ 
&< \frac{\sum_{i\in \cI_0} w(X_i)\ind{\{V(X_i,0) < V(X_{n+j},0) \}}   +  w(X_{n+j}) }{\sum_{i\in \cI_0} w(X_i) + w(X_{n+j})} = p_j',
\$
where the second lines uses~\eqref{eq:monotone_V}. 
That is, with this   nonconformity score (which is shown 
in~\cite{jin2022selection}
to be powerful), the p-values constructed in Algorithm~\ref{alg:bh}
is strictly  smaller than those in Algorithm~\ref{alg:bh_cond}, 
leading to larger rejection sets and higher power. 
Furthermore, in this case, 
\$
\frac{p_j}{p_j'} = \frac{\sum_{i\in \cI_0} w(X_i) + w(X_{n+j})}{\sum_{i=1}^n w(X_i) + w(X_{n+j})}
\$
is the weighted proportion of negative calibration samples. 
Thus, the power gain of Algorithm~\ref{alg:bh} is more significant 
when there are more positive samples in the test distribution. 

Second, Algorithm~\ref{alg:bh_cond} is suitable for dealing with 
imbalanced data, such as those encountered in drug discovery. 
For instance, after selecting a subset of molecules or compounds  for virtual screening, 
the experimenter may discard a few samples with $Y=0$ as she would deem them as uninteresting. 
This bias would make Algorithm~\ref{alg:bh} inapplicable. 
However, if the decision to discard or not
does not depend on $X$, 
the  shift between $P_{X\given Y=0}$ 
and $Q_{X\given Y=0}$ remains the same as
the covariate shift incurred by selection into screening, 
and Algorithm~\ref{alg:bh_cond} still provides reliable selection.

We finally illustrate the hypothesis-conditional variant 
from Algorithm~\ref{alg:bh_cond} on outlier detection tasks, 
where there is a covariate shift between calibration inliers 
and test inliers. We focus on scenarios where the covariate shift 
is known. For instance, if the demographics in normal financial transactions 
are adjusted by rejection sampling 
to balance between male  and female, urban and rural users, and so on, then 
the covariate shift is given by the sampling weights in the adjustment.\footnote{Results in this section can be reproduced at \href{https://github.com/ying531/conformal-selection}{https://github.com/ying531/conformal-selection}.}

\subsection{Simulation studies}

We adapt the simulation setting in~\cite{bates2021testing}
to a scenario with covariate shift induced by rejection sampling.   
We fix a test sample size $n_{\test} = 1000$ and 
calibration sample size $n_\calib=1000$. 

At the beginning of the experiment, we sample a subset $\cW\subseteq \RR^{50}$ with $|\cW|=50$, where each element in $\cW$ is 
independently drawn from $\textrm{Unif}([-3,3]^{50})$, 
and hold it as fixed afterwards. 
Fix a proportion of outliers at $\rho_{\texttt{pop}}\in \{0.1,0.2,\dots,0.5\}$. 
The number of outliers in the test data 
is $n_\test \cdot \rho_{\texttt{pop}}$, 
where each of them is i.i.d.~generated as 
$X_{n+j}=\sqrt{a}V_{n+j}+W_{n+j}$ 
for signal strength $a$ varying in the range $[1,4]$, 
$V_{n+j}\sim N(0,I_{50})$, and $W_{n+j}\sim \textrm{Unif}(\cW)$.  
Following this, the test inliers are i.i.d.~generated as 
$X_{n+j}= V_{n+j}+W_{n+j}$, whose distribution is denoted as $Q_X$. 
The calibration inliers are i.i.d.~generated from
$P_X$ with $\ud Q_X/\ud P_X(x)=w(x)\propto \sigma(x^\top\theta)$ ($\sigma(\cdot)$ is the sigmoid), where $\theta\in \RR^{50}$ and $\theta_j=0.1\cdot\ind\{j\leq 5\}$. 
We also generate $n_{\train}=1000$ training sample from $P_X$. This setting mimics a stylized scenario 
where the calibration inliers are collected in a way such 
that the preference is prescribed by a logistic function of $X$, 
and is known to the data analyst. 

We train a one-class SVM $\hat\mu(x)$ with \texttt{rbf} kernel using 
the \texttt{scikit-learn} Python library, 
and apply Algorithm~\ref{alg:bh_cond} with $Z_i=X_i$ at 
FDR level $q=0.1$
using all the three pruning methods. 
All procedures are repeated $N=1000$ times, 
and the FDPs and proportions of true discoveries are averaged 
to estimate the FDR and power. We also evaluate the BH 
procedure applied to p-values constructed in Line 4 of 
Algorithm~\ref{alg:bh_cond}. 
Figure~\ref{fig:fdr_outlier} shows the 
empirical FDR across runs. In line with Theorem~\ref{thm:fdr_cond}, we see that the FDR is 
always below $(1-\rho_{\texttt{pop}})q$. Also, 
BH applied to weighted conformal p-values 
empirically controls the FDR and demonstrates comparable performance. 

\begin{figure}[H]
    \centering
    \includegraphics[width=6.5in]{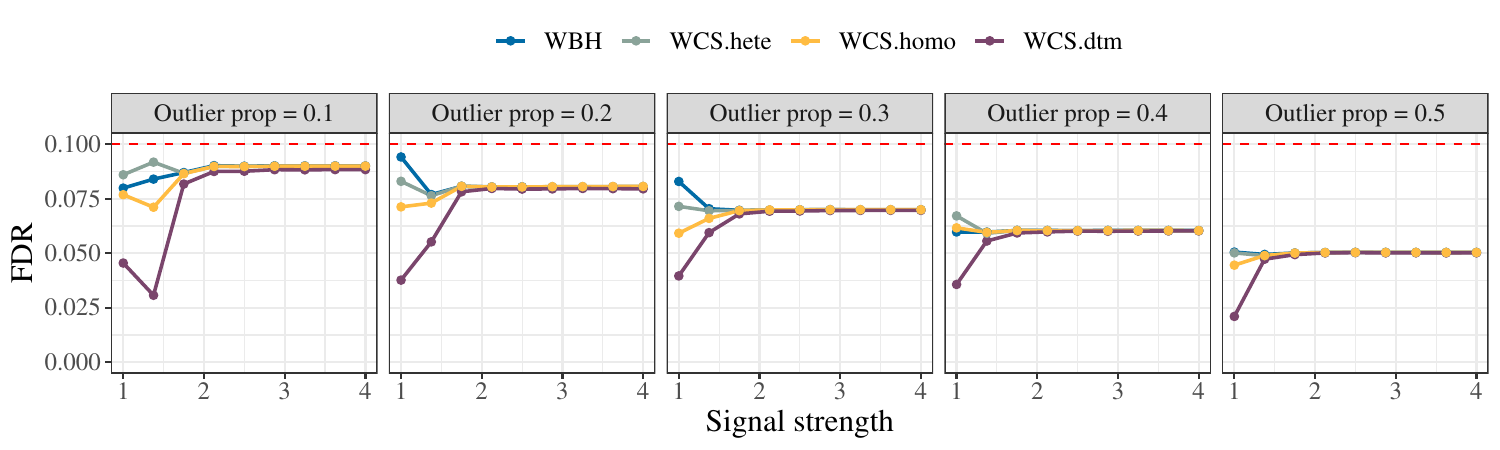}
    \caption{Empirical FDR averaged over $N=1000$ runs under increasing levels $a$ of signal strength. Each subplot corresponds to one value of $\rho_{\texttt{pop}}$. The red dashed lines are at the nominal level $q=0.1$.}
    \label{fig:fdr_outlier}
\end{figure}

Figure~\ref{fig:power_outlier} plots the power of the four methods. 
Interestingly, although they differ in FDR, especially 
when the signal strength is small, they achieve nearly identical power, finding about the same number of true discoveries. 

\begin{figure}[H]
    \centering
    \includegraphics[width=6.5in]{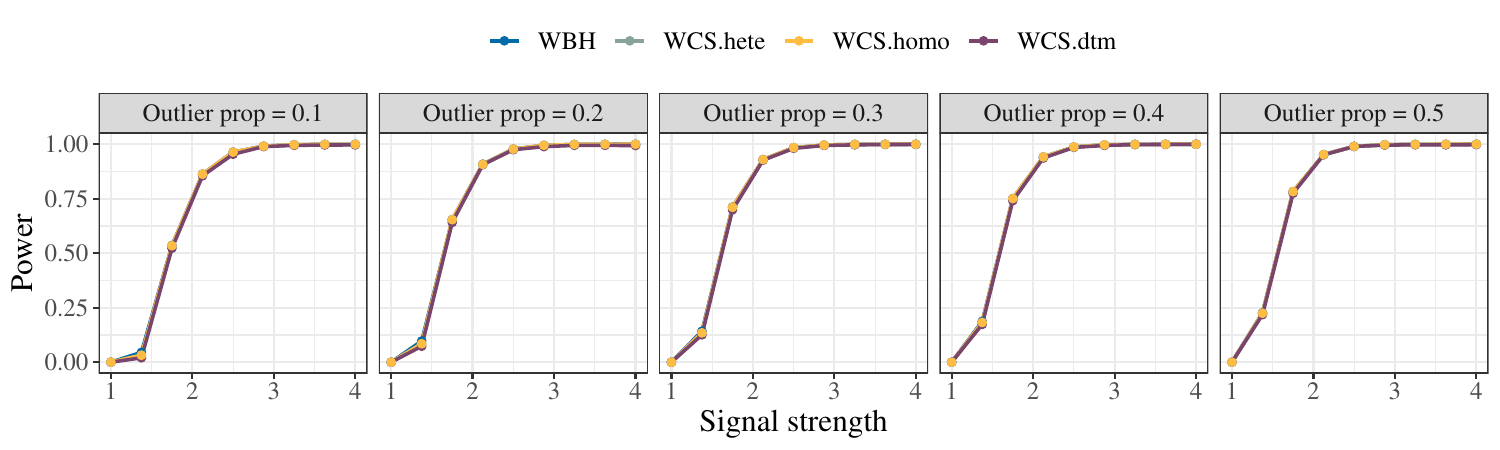}
    \caption{Empirical power averaged over $N=1000$ runs. Everything else is as in Figure~\ref{fig:fdr_outlier}.}
    \label{fig:power_outlier}
\end{figure}

\subsection{Real data application}

We then apply Algorithm~\ref{alg:bh_cond} to a bank marketing dataset~\cite{moro2014data}; this dataset has been used to benchmark  
outlier detection algorithms~\cite{pang2019deep,pang2021deep}. 
In short, each sample in the dataset represents a phone contact to a potential 
client, whose demographic information such as age, gender, and marital status 
is also recorded in the features.
The outcome is a binary indicator $Y$, where 
$Y=1$ represents a successful compaign. Positive samples are relatively rare, 
accounting for about $10\%$ of all records.

In our experiments, we always use negative samples as the calibration data. 
We find this dataset interesting because 
the classification and outlier detection perspectives 
are somewhat blurred here. 
Viewing this as a classification task, Example~\ref{ex:binary} 
becomes relevant, and our theory implies FDR control as long as the 
distribution of negative samples 
(i.e., that of $X$ given $Y=0$) admits the covariate shift. 
In the sense of finding positive responses (so as to reach out 
to these promising clients), 
controlling the FDR ensures efficient use of campaign resources. 
Alternatively, from an outlier detection perspective, 
those $Y=1$ are outliers~\cite{pang2019deep} that 
deserve more attention. 
Taking either of the two perspectives leads to the same calibration set 
and the same guarantees following Section~\ref{subsec:hypo_cond_fdr}.  
Perhaps the only methodological distinction is 
whether positive samples are leveraged in the 
training process. Similar considerations  
also appeared in~\cite{liang2022integrative}, but the scenario therein 
is more sophisticated than ours as they also use 
known outliers (positive samples) in 
constructing p-values.  
As a classification problem, we may use positive samples to train 
a classifier and produce nonconformity scores. 
As an outlier detection problem, however, 
the widely-used  one-class classification 
relies exclusively on inliers (negative samples). 
We will evaluate the performance of both 
(i.e., train the nonconformity score with or without positive samples); 
see procedures \texttt{cond\_class} and \texttt{outlier} below. 

For comparison, we additionally evaluate Algorithm~\ref{alg:bh} 
which operates under a covariate shift assumption 
on the joint distribution of $(X,Y)$ (see procedure \texttt{sup\_class} below); this is in contrast 
to Algorithm~\ref{alg:bh_cond} which puts no assumption on positive samples. 

The total sample size is $N=41188$ with $4640$ positive samples. 
We use rejection sampling to create the covariate shift between 
training/calibration and test samples. In details, we use a subset of features $X^*\in \RR^4$ 
representing age and indicators of marriage, 
basic four-year education, and housing. We first randomly select a subset of the original dataset as the test data.
Each sample enters the test data with probability 
$e(x):=0.125\, \sigma({\theta^\top x^*})$, where $\theta=(1,1/2,1/2,1/2)^\top$.  This creates a covariate shift $w(x)\propto e(x)/(1-e(x))$ between 
null samples in calibration and test folds, so that 
the test set contains more senior, married, well-educated and housed people.  

To test the two perspectives, we consider three procedures with FDR levels $q\in \{0.2, 0.5, 0.8\}$:
\begin{enumerate}[(i)]
    \item \texttt{sup\_class}: We randomly split the data that is not in the test fold into two equally-sized halves as the training and calibration folds, $\cD_\train$ and $\cD_\calib$. 
    We use $\cD_\train$ to train an SVM classifier using the \texttt{scikit-learn} Python library with \texttt{rbf} kernel. 
    Then, we apply Algorithm~\ref{alg:bh} with $V(x,y)=100\cdot y - \hat\mu(x)$, where $\hat\mu(\cdot)$ is the class probability output from the SVM classifier. 
    \item \texttt{cond\_class}: With the same random split as in (i), the first half is used as $\cD_\train$, while we set $\cD_\calib$ as the set of negative samples in the second half. We use $\cD_\train$ to train an SVM classifier in the same way as in (i) to obtain the predicted class probability $\hat\mu(\cdot)$. Lastly, we apply Algorithm~\ref{alg:bh_cond} with $V(x)=-\hat\mu(x)$. 
    \item \texttt{outlier}: With the same split, we set $\cD_\train$ and $\cD_\calib$ as the two folds after discarding all the negative samples. We use $\cD_\train$ to train a one-class SVM classifier using the \texttt{scikit-learn} Python library with \texttt{rbf} kernel and obtain the output $\hat\mu(\cdot)$. We then apply Algorithm~\ref{alg:bh_cond} with $V(x)=-\hat\mu(x)$. 
\end{enumerate}

To ensure consistency, in all procedures, 
the SVM-based classifier uses the parameter \texttt{gamma}$=0.01$.\footnote{We find that this parameter leads to the highest power for outlier detection, and is 
also a relatively powerful choice for classification.} 
FDPs and power (proportions of true discoveries) across  $N=1000$ independent runs are 
shown in Figures~\ref{fig:real_outlier_fdr} and~\ref{fig:real_outlier_power}, 
respectively. 

\begin{figure}[h]
    \centering 
    \includegraphics[width=\linewidth]{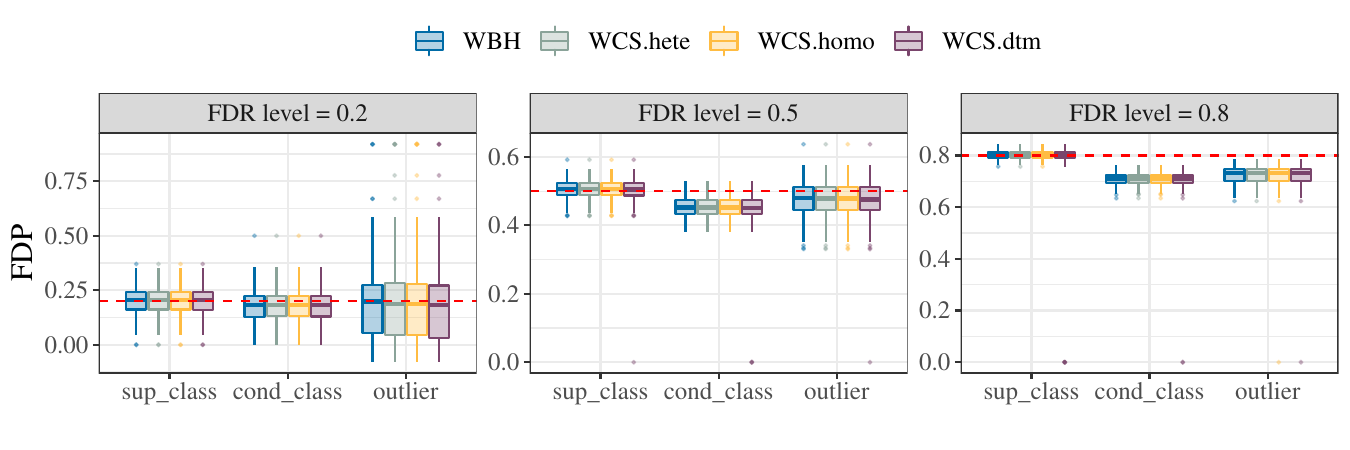}
    \caption{FDPs over $N=1000$ independent data splits. Each subplot corresponds to one nominal FDR level shown by means of red dashed lines.} 
    \label{fig:real_outlier_fdr}
\end{figure}

In Figure~\ref{fig:real_outlier_fdr}, we observe 
FDR control for all the methods in all settings. 
Across the four testing methods, the FDPs are very similar, 
hence all of them are reasonable choices. However, both the values and the variability of the FDP 
vary across settings. The methods 
\texttt{cond\_class} and \texttt{outlier} from Algorithm~\ref{alg:bh_cond} 
do not use all of the error budget---note the $\pi_0$ factor in Theorem~\ref{thm:fdr_cond}. 
In contrast, Algorithm~\ref{alg:bh} leverages super-population structure 
(which is present in this problem) and 
leads to tight FDPs and FDR around the target level. 
In all cases, the outlier detection approach, where 
only negative samples are used in the training process, 
has more variable FDPs. 
With this dataset, 
the variability of FDP around the FDR is visible for $q=0.2$, but it 
decreases as we increase the 
FDR level.

\begin{figure}[h]
    \centering 
    \includegraphics[width=6.5in]{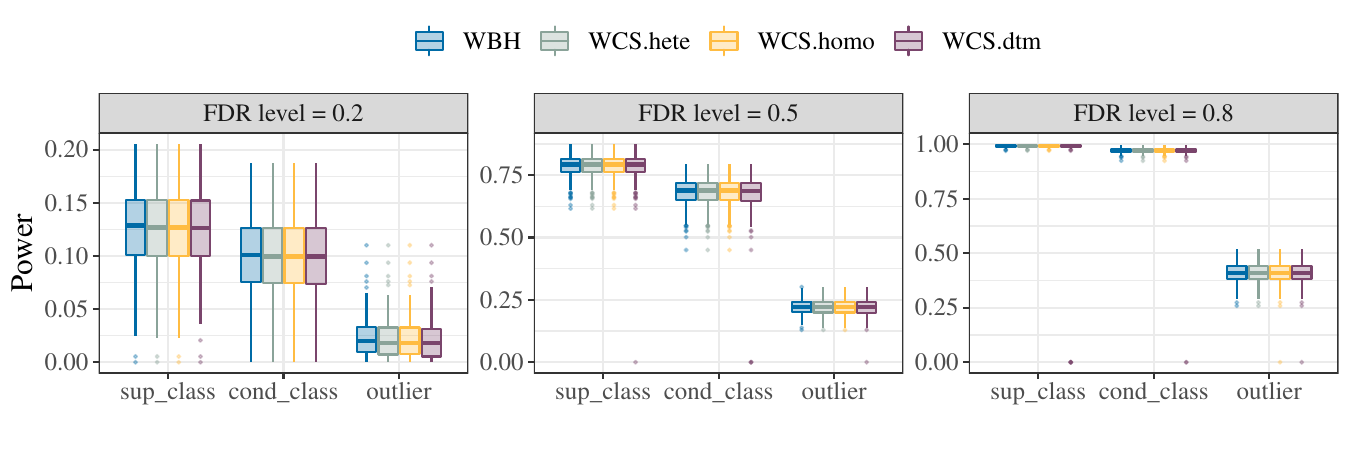}
    \caption{Proportion of true discoveries (empirical power). Everything else is as in Figure~\ref{fig:real_outlier_fdr}.}
    \label{fig:real_outlier_power}
\end{figure}

In Figure~\ref{fig:real_outlier_power}, the power across 
the four multiple testing methods is similar. 
However, we observe  drastic 
differences among the three settings. 
While using the same classifier, 
\texttt{cond\_class} has lower power than \texttt{sup\_class}, 
per our discussion in Section~\ref{subsubsec:compare}. 
We also see that \texttt{outlier} has much lower power, which may be due to the 
fact that the nonconformity score is not sufficiently powerful
in distinguishing inliers from outliers 
as the training process does not use outliers. 
Our results suggest that in outlier detection problems, 
even when we do not want to impose any distributional assumption 
on outliers (so that Algorithm~\ref{alg:bh} becomes less reasonable), 
utilizing known outliers may be helpful in obtaining better nonconformity 
scores. 

Finally, we observe negligible difference between the rejection sets 
returned by the BH procedure and Algorithm~\ref{alg:bh}. There at most five distinct decisions among $\approx 3700$ test samples.

\section*{Acknowledgement}

The authors thank John Cherian, Issac Gibbs, Kevin Guo, Kexin Huang, 
Jayoon Jang, Lihua Lei, Shuangning Li, Zhimei Ren, 
Hui Xu, and Qian Zhao for helpful discussions.  
E.C.~and Y.J.~were supported by the Office of Naval Research grant N00014-20-1-2157, the National Science
Foundation grant DMS-2032014, the Simons Foundation under award 814641, and the ARO grant
2003514594.

\bibliographystyle{alpha}
\bibliography{reference}

\newpage 
\appendix 

\section{Deferred results and details}

\subsection{Recap: Conformalied Selection}
\label{app:recap_unweighted}

In this section, we briefly recall the conformalized selection framework of~\cite{jin2022selection}. 
As before, set $\cD_\calib=\{1,\dots,n\}$ 
and $\cD_\test=\{n+1,\dots,n+m\}$. 
Screening for large outcomes can be viewed as testing the null hypotheses 
\$
H_j\colon Y_{n+j}\leq c_{n+j},\quad  j=1,\dots,m,
\$
where selecting $j\in \cR$ is equivalent to rejecting $H_j$. 
Perhaps unconventionally, whether $H_j$ is true or not 
is here a random event. Yet,  we 
still construct \emph{conformal} p-values built upon the (split) conformal inference 
framework~\cite{vovk2005algorithmic} to test these hypotheses.

Given a monotone score  $V$ and thresholds $\{c_{n+j}\}_{j=1}^m$, 
we compute 
$V_i = V(X_i,Y_i)$ for $i\in \cD_\calib$, 
$\hat{V}_{n+j}=V(X_{n+j},c_{n+j})$, 
and then apply the Benjamini-Hochberg (BH) procedure~\cite{benjamini1995controlling} to 
conformal p-values 
\#\label{eq:conformal_pval}
p_j^\nw = \frac{\sum_{i=1}^n \ind {\{V_i <\hat{V}_{n+j} \}}+ U_j ( 1+ \sum_{i=1}^n \ind {\{V_i = \hat{V}_{n+j} \}})}{n+1},\quad U_j\iid \textrm{Unif}([0,1]).
\#  
Under the assumption that $Q=P$, or $w(x)\equiv 1$ in~\eqref{eq:cov_shift}, 
\revisea{each $ p_j^\nw $ is a valid p-value obeying~\eqref{eq:general_pvalue}.} 
However, as shown earlier, 
under non-trivial covariate shift, the random variables $p_j^\nw$ 
are no longer valid p-values.\footnote{ 
A side result in Appendix~\ref{app:subsec_wfdr} shows that 
running the BH procedure with $\{p_j^\nw\}$  
controls a weighted notion of FDR below the target level $q$;  
this weighted notion downweights the error for $(X,Y)$ values 
that are less frequent in the calibration data.}

\subsection{Connection of our p-values and weighted conformal prediction}
\label{app:connection_conformal}

\begin{remark}\normalfont
  An intuitive interpretation for 
  the p-values in~\eqref{eq:def_wcpval_rand} 
  is that they are roughly the critical confidence level 
  when inverting weighted conformal prediction intervals
  at the thresholds $\{c_{n+j}\}$. 
  Indeed, the split weighted conformal inference
  procedure~\cite{tibshirani2019conformal} (using $-V$ as the nonconformity score for consistency) 
  yields the one-sided prediction interval 
  $\hat{C}(X_{n+j},1-\alpha) = [\eta(X_{n+j},1-\alpha),+\infty)$
  for $Y_{n+j}$, where 
  \$
  -\eta(x,1-\alpha) = \textrm{Quantile} \big(1-\alpha; 
  {\textstyle  \sum_{i=1}^n} w_i(x)\delta_{-V_i} + w_{n+1}(x) \delta_{-\infty}\big),
  \$
  and $w_i(x) = \frac{w(X_i)}{\sum_{i=1}^n w(X_i) + w(x)}$, $i=1,\dots,n$,
  and $w_{n+1}(x)=  \frac{w(x)}{\sum_{i=1}^n w(X_i) + w(x)}$ for any $x\in \cX$. 
  Ignoring the tie-breaking random variables, we see that $p_j$ is the smallest 
  $\alpha$ such that $c_{n+j}<\eta(X_{n+j},1-\alpha)$. 
  \end{remark}

\subsection{Weighted FDR control for unweighted p-values}
\label{app:subsec_wfdr}

The following theorem shows the weighted FDR control 
when running BH($q$) with unweighted conformal p-values 
while there is a covariate shift between calibration and test samples. 

\begin{theorem}\label{thm:bh_weighted_fdr}
  Suppose $V$ is monotone, 
  $\{(X_i,Y_i)\}_{i=1}^n\sim P$ and 
  $\{(X_{n+j},Y_{n+j})\}_{j=1}^m\sim Q$ are mutually independent 
  and obey~\eqref{eq:cov_shift}. 
  Suppose $\{(X_i,Y_i)\}_{i=1}^n \cup \{(X_{n+\ell},c_{\ell})\}_{\ell\neq j}\cup \{(X_{n+j},Y_{n+j})\}$ are mutually independent for all $j$. 
  For any $q\in(0,1)$, let $\cR_{\nw}$ be the rejection set of 
  BH($q$) procedure applied to 
  the unweighted conformal $p$-values $\{p_j^\nw\}$ in~\eqref{eq:conformal_pval}. 
   Then 
  \$
  \EE\Bigg[  \frac{\sum_{j=1}^m  \ind\{j\in \cR_\nw, Y_{n+j}\leq c_{n+j}\}/w(X_{n+j})  }{1\vee |\cR_\nw|}   \Bigg] \leq q\cdot \EE\bigg[\frac{  n+1   }{ \sum_{i=1}^n w(X_i) + w(X_{n+1})}\bigg],
  \$
  where the expectation is over both the test and calibration data. 
\end{theorem}

\subsection{General theorem for asymptotic FDR control with BH procedure}
\label{app:fdr_asymp}

We provide the general results for Theorem~\ref{thm:fdr_asymp} below. 
\begin{theorem}\label{thm:fdr_asymp_full}
  Suppose $w(\cdot)$ is uniformly bounded by a fixed constant, the covariate shift~\eqref{eq:cov_shift} holds, 
  and $\{c_{n+j}\}_{j=1}^m$ are i.i.d.~random variables. 
  Fix any $q\in(0,1)$, 
  and let $\cR$ be the output of the BH($q$) procedure applied to 
  $\{p_j\}_{j=1}^m$ in~\eqref{eq:def_wcpval_rand}. 
  Then the following two results hold. 
  \begin{enumerate}[(i)]
  \item For any fixed $m$, as $n\to \infty$, 
  it holds   that $\limsup_{n\to \infty} \fdr \leq q$.
  \item As $m,n\to\infty$, it holds that $\limsup_{m,n\to \infty} \fdr \leq q$ if 
  there exists some $t'\in(0,1]$ such that 
  $G_\infty(t') <q$, where  
  $
  G_\infty(t) :=    t /\PP ( F( V(X_{n+j},C_{n+j}) , U_j) \leq t  )  
  $
  for any $t\in [0,1]$, with the convention that $0/0=0$ and $t/0=\infty$ for $t>0$. 
  Here, 
  for any $v\in \RR$ and $u\in[0,1]$, 
  \#\label{eq:def_F}
  F(v,u) = Q\big(V(X,Y )<v\big) + u\cdot Q\big( V(X,Y )=v \big),
  \#
  where the probability is taken with respect to 
  the test distribution $(X,Y)\sim Q$ obeying~\eqref{eq:cov_shift}.
    
  Furthermore, define $t^* = \sup\{t\in [0,1]\colon G_\infty(t) \leq q\}$. 
  Additionally suppose for any sufficiently small $\epsilon>0$, 
  there exists some $t\in(t^*-\epsilon,t)$ such that $G_\infty(t) < q$. 
  Then the asymptotic FDR of BH($q$) is 
  $\frac{ Q \{ F(V(X,C),U)\leq t^*, Y\leq C \}}{Q\{F(V(X,C),U)\leq t^* \} }$, 
  and the asymptotic power is 
  $Q \big\{ F(V(X,C),U)\leq t^*\given  Y  > C \big\} $, 
  where all the probabilities are induced by the test distribution 
  $Q_{X,Y ,C }$ and an independent $U\sim $ \rm{Unif([0,1])}.  
  \end{enumerate} 
  \end{theorem}

\subsection{Deferred simulation results for DTI prediction}
\label{app:subsec_dti}

\begin{figure}[H]
  \centering
  \includegraphics[width=5.5in]{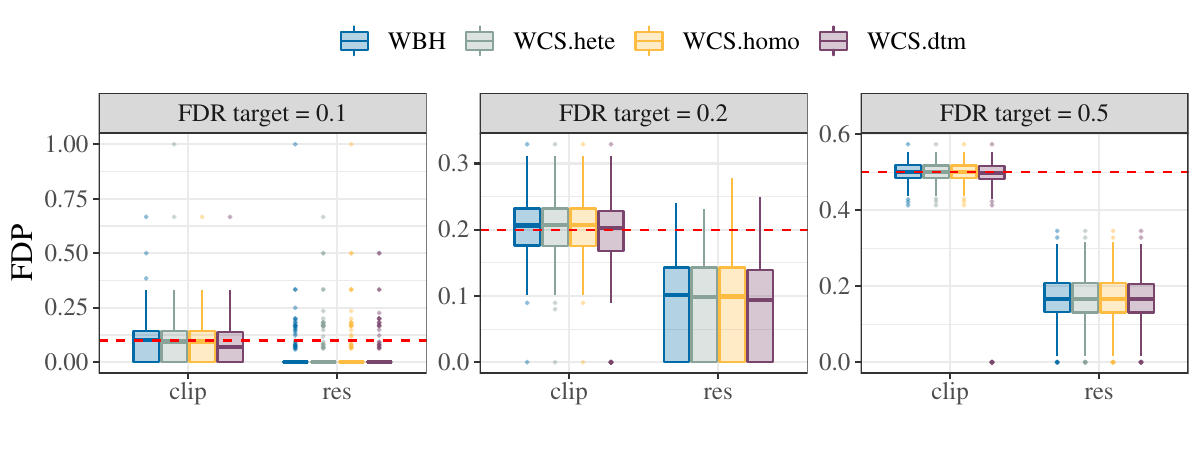}
  \caption{Empirical FDR for DTI prediction  with $q_{\textrm{pop}}=0.7$.
  Details are the same as Figure~\ref{fig:drug_dti_fdr}.}
  \label{fig:drug_dti_fdr_7}
\end{figure}

\begin{figure}[H]
  \centering
  \includegraphics[width=5.5in]{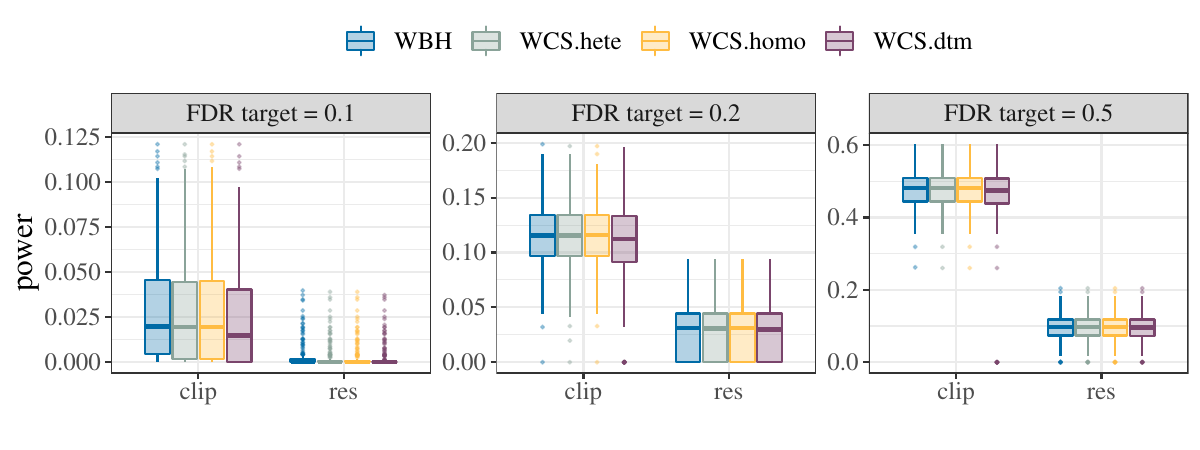}
  \caption{Empirical power for DTI with $q_{\textrm{pop}}=0.7$. 
  Details are the same as Figure~\ref{fig:drug_dti_fdr}.}
  \label{fig:drug_dti_power_7}
\end{figure}

\section{Proofs for weighted p-values}

\subsection{Proof of Lemma~\ref{lem:general_pval}}
\label{app:lem_general_pval}

\begin{proof}[Proof of Lemma~\ref{lem:general_pval}]
For any $j\in\{1,\dots,m\}$, note that $V(X_{n+j},y)$ is 
increasing in $y$, thus on the event  $\{Y_{n+j}\leq c_{n+j}\}$, 
it holds that $V_{n+j} \leq \hat{V}_{n+j}$. 
We now define 
\$
p_j^* =   \frac{\sum_{i=1}^n w(X_i)\ind {\{V_i < {V}_{n+j} \}}+  ( w(X_{n+j}) + \sum_{i=1}^n w(X_i)\ind {\{V_i = {V}_{n+j} \}})\cdot U_j}{\sum_{i=1}^n w(X_i) + w(X_{n+j})},
\$ 
where $U_j$ is the same tie-breaking random variable as in $p_j$. 
That is, $p_j^*$ is the oracle p-value  we would get 
if we replace $\hat{V}_{n+j}$ by $V_{n+j}$ in the construction of $p_j$. 
Note that $p_j^*$ is not computable; it is only for the purpose of proof. 
Thus, it is clear that $p_j^*\leq p_j$ on the event  $\{Y_{n+j}\leq c_{n+j}\}$. 
This implies
\$
\PP(p_j \leq t, ~ Y_{n+j} \leq c_{n+j} )
\leq \PP(p_j \leq t, ~ p_j^*\leq p_j) \leq \PP(p_j^*\leq t). 
\$
The right-handed side above is below $t$ 
using weighted exchangeability introduced in~\cite{tibshirani2019conformal}. 
Formally, define the unordered set $Z = [Z_1,\dots,Z_n,Z_{n+j}]$ 
for $Z_i=(X_i,Y_i)$, $i=1,\dots,n,n+j$, and write $z$ as the realized value of $Z$, 
and $\cE_z =\{Z=z\}$. Under the covariate shift~\eqref{eq:cov_shift}, the scores 
$\{V_1,\dots,V_n, V_{n+j}\}$ are weighted exchangeable, in the sense that  
\$
Z_{n+j} \given \cE_z ~\sim ~ \sum_{i\in \{1,\dots,n,n+j\}} \delta_{z_i} \cdot \frac{w(x_i)}{\sum_{i=1}^n w(x_i) + w(x_{n+j})},
\$
i.e., given that the unordered set $Z$ equals $z$, 
the probability that $Z_{n+j}$ taking on any value  $(x,y)\in \{z_1,\dots,z_n,z_{n+j}\}$ 
is proportional to the corresponding weight $w(x)$. This implies 
\$
\PP(p_j^*\leq t \given \cE_z) = t,
\$
and thus $\PP(p_j^*\leq t)=t$, which concludes the proof of Lemma~\ref{lem:general_pval}.
\end{proof}

\subsection{Proof of Theorem~\ref{thm:bh_weighted_fdr}} 
\label{app:thm_bh_wfdr}

\begin{proof}[Proof of Theorem~\ref{thm:bh_weighted_fdr}]
For ease of illustration, we prove the result 
when there is no tie among all nonconformity scores;  
similar results can be obtained with ties. 
We first use a leave-one-out trick similar to~\cite{jin2022selection}. 
Define the oracle conformal p-values 
\$
\bar p_j^\nw = \frac{\sum_{i=1}^n \ind {\{V_i < {V}_{n+j} \}}+ U_j ( 1+ \sum_{i=1}^n \ind {\{V_i =  {V}_{n+j} \}})}{n+1},
\$
which is obtained by replacing the $\hat{V}_{n+j}$ in~\eqref{eq:conformal_pval}
by $V_{n+j}=V(X_{n+j},Y_{n+j})$. 
By the monotonicity of $V$, one has $p_j^\nw \geq \bar{p}_j^\nw$ 
for any $j$ such that $Y_{n+j}\leq c_{j}$. 
In the BH procedure, for any $j\in \cR_\nw$, 
replacing $p_j$ with a smaller value does not change the rejection set. 
Therefore, letting $\bar{\cR}_{j}$ be the rejection set of 
BH($q$) applied to $\{p_1^\nw, \dots, p_{j-1}^\nw, \bar{p}_j^\nw, \dots, p_m^\nw\}$, 
we know that 
\$
j\in \cR_{\nw}  ~ \textrm{and}~ Y_{n+j}\leq c_{n+j} \quad \Rightarrow \quad j \in \bar\cR_j ~\textrm{and}~ \bar\cR_j = \cR_\nw.
\$
We then have  
\#\label{eq:bh_weight1}
\EE\Bigg[  \frac{\sum_{j=1}^m  \ind\{j\in \cR_\nw, Y_{n+j}\leq c_{n+j}\}/w(X_{n+j})  }{1\vee  |\cR_\nw| }   \Bigg] 
\leq \sum_{j=1}^m  \EE\Bigg[  \frac{ \ind\{j\in \bar\cR_j \}/w(X_{n+j})  }{1\vee |\bar\cR_j| }   \Bigg] .
\#
Also note that $j\in \bar{\cR}_j$ if and only if $\bar{p}_j^\nw \leq q|\bar\cR_j|/m$. 
Thus, for any $j=1,\dots,m$, 
\#\label{eq:bh_weight2}
&\EE\Bigg[  \frac{ \ind\{j\in \bar\cR_j \}/w(X_{n+j})  }{1\vee |\bar\cR_j| }   \Bigg]
= \sum_{k=1}^m \frac{1}{k}\EE\big [  \ind {\{|\bar\cR_{j}|=k\}}\ind {\{\bar p_j^\nw\leq q k/m\}} /w(X_{n+j})  \big] 
\notag \\
&= \sum_{k=1}^m \frac{1}{k} \EE\big [ ( \ind {\{|\bar\cR_{j}|\leq k\}} - \ind {\{|\bar\cR_{j}|\leq k-1\}} )\cdot \ind {\{\bar p_j^\nw\leq q k/m\}} /w(X_{n+j})  \big] \notag  \\ 
&\leq \frac{1}{m} \EE\bigg [ \frac{ \ind {\{\bar p_j^\nw\leq q  \}}}{w(X_{n+j})}  \bigg] 
+ \sum_{k=1}^{m-1} \EE\bigg[  \frac{\ind {\{|\bar\cR_{j}|\leq k\}} \ind {\{\bar p_j^\nw\leq \frac{q k}{m}\}}}{k \cdot w(X_{n+j})} -  \frac{\ind {\{|\bar\cR_{j}|\leq k\}} \ind {\{\bar p_j^\nw\leq \frac{q (k+1)}{m}\}}}{(k+1) \cdot w(X_{n+j})} \bigg] .
\#

We define the event $\cE_z = \{Z=z\}$, 
where $Z=[Z_1,\dots,Z_n, Z_{n+j}]$ is the unordered set 
for $Z_i = (X_i,Y_i)$ and $Z_{n+j}=(X_{n+j},Y_{n+j})$, 
and $z=[z_1,\dots,z_n,z_{n+j}]$ is the unordered set of 
their realized values. By the covariate shift~\eqref{eq:cov_shift}, 
$\{(X_{i},Y_i)\}_{i=1}^n$ and $(X_{n+j},Y_{n+j})$ are 
weighted exchangeable~\cite{tibshirani2019conformal}, hence 
\#\label{eq:weighted_exch}
Z_{n+j} \given \cE_z ~\sim ~ \sum_{i\in \{1,\dots,n,n+j\}} \delta_{z_i} \cdot \frac{w(x_i)}{\sum_{i=1}^n w(x_i) + w(x_{n+j})},
\#
where $\delta_{z_i}$ is the point mass on $z_i$. 
Therefore, 
by the tower property, 
\$
&\EE\Bigg[ \frac{ \ind {\{|\bar\cR_{j}|\leq k\}} \ind {\{\bar p_j^\nw\leq q k/m\}}}{k \cdot w(X_{n+j})} -  \frac{\ind {\{|\bar\cR_{j}|\leq k\}} \ind {\{\bar p_j^\nw\leq q (k+1)/m\}}}{(k+1) \cdot w(X_{n+j})} \bigggiven \cE_z \Bigg] \\
&\leq \EE\Bigg[ \frac{\ind {\{\bar p_j^\nw\leq q k/m\}}}{w(X_{n+j})} \cdot \frac{ 1 }{k} \cdot \PP\big(|\bar\cR_{j}|\leq k \biggiven \bar p_j^\nw  , \cE_z\big) \bigggiven \cE_z \Bigg]  \\
&\quad - \EE\Bigg[ \frac{\ind {\{\bar p_j^\nw\leq q (k+1) /m\}}}{w(X_{n+j})} \cdot \frac{ 1 }{k+1} \cdot \PP\big(|\bar\cR_{j}|\leq k \biggiven \bar p_j^\nw  , \cE_z\big) \bigggiven \cE_z \Bigg]  
\$
Note that the randomness in $U_j$ is independent of everything else, hence 
conditioning on $(\bar{p}_j^\nw,\cE_z)$ is the same as conditioning on $(Z_{n+j},\cE_z)$. 
We also denote $[1],[2],\dots,[n+1]$ as a permutation of 
$(1,2,\dots,n,n+j)$ such that $V(X_{[1]},Y_{[1]}) \leq V(X_{[2]},Y_{[2]}) \leq \cdots \leq V(X_{[n+1]},Y_{[n+1]})$. 
Conditional on $\cE_z$,
on the event $\{\bar{p}_j^\nw\leq q k/m\}$,  
$Z_{n+j}$ can only take values at $z_i$ for which 
$\frac{\sum_{\ell\neq i} \ind {\{V(z_\ell)<V(z_k)\}} }{n+1} \leq q k/m$; 
that is, the possible values of $Z_{n+j}$ 
are in  the ordered set $\{z_{[1]},\dots,z_{[i_k]}\}$ for some $i_k$. 
Similarly, $Z_{n+j}$ can only take values in 
$\{z_{[1]},\dots,z_{[i_{k+1}]}\}$ for $i_{k+1}\geq i_k$ on 
the event $\{\bar{p}_j^\nw\leq q (k+1)/m\}$. 
Therefore, we have 
\$
&\EE\Bigg[ \frac{ \ind {\{|\bar\cR_{j}|\leq k\}} \ind {\{\bar{p}_j^\nw\leq q k/m\}}}{k \cdot w(X_{n+j})} -  \frac{\ind {\{|\bar\cR_{j}|\leq k\}} \ind {\{\bar{p}_j^\nw\leq q (k+1)/m\}}}{(k+1) \cdot w(X_{n+j})} \bigggiven \cE_z \Bigg] \\
&\leq 
\EE\Bigg[ \sum_{\ell=1}^{i_k} \ind {\{Z_{n+j}=z_{[\ell]}\}} \frac{\ind {\{\bar{p}_j^\nw\leq q k/m\}}}{w(x_{[\ell]})} \cdot \frac{ 1 }{k} \cdot \PP\big(|\bar\cR_{j}|\leq k \biggiven Z_{n+j}=z_{[\ell]}  , \cE_z\big) \bigggiven \cE_z \Bigg]  \\
&\quad - \EE\Bigg[ \sum_{\ell=1}^{i_{k+1}} \ind {\{Z_{n+j}=z_{[\ell]}\}} \frac{\ind {\{\bar{p}_j^\nw\leq q (k+1) /m\}}}{w(x_{[\ell]})}  \cdot \frac{ 1 }{k+1}\cdot \PP\big(|\bar\cR_{j}|\leq k \biggiven Z_{n+j}=z_{[\ell]} , \cE_z\big) \bigggiven \cE_z \Bigg] \\
& = \sum_{\ell=1}^{i_k} \frac{\PP\big(|\bar\cR_{j}|\leq k \biggiven Z_{n+j}=z_{[\ell]}  , \cE_z\big)}{k\cdot w(x_{[\ell]})}
\EE\big[  \ind {\{Z_{n+j}=z_{[\ell]}\}} \ind {\{\bar{p}_j^\nw\leq q k/m\}}  \biggiven \cE_z \big]  \\
&\quad - 
\sum_{\ell=1}^{i_{k+1}} \frac{\PP\big(|\bar\cR_{j}|\leq k \biggiven Z_{n+j}=z_{[\ell]}  , \cE_z\big)}{(k+1)\cdot w(x_{[\ell]})}
\EE\big[  \ind {\{Z_{n+j}=z_{[\ell]}\}} \ind {\{\bar{p}_j^\nw\leq q (k+1)/m\}}  \biggiven \cE_z \big] .
\$
By weighted exchangeability, 
for $i<i_k$ it holds deterministically that $\bar{p}_j^\nw\leq q k/m$ 
and $\PP(Z_{n+j}=z_{[\ell]}\given \cE_z) = \frac{w(x_{[\ell]})}{\sum_{i=1}^n w(x_i) + w(x_{n+j})}$. 
The randomness of $U_j\sim \textrm{Unif}(0,1)$ implies 
\$\EE\big[  \ind {\{Z_{n+j}=z_{[i_{k}]}\}} \ind {\{\bar{p}_j^\nw\leq q k/m\}}  \biggiven \cE_z \big]
= (n+1)\cdot\bigg(\frac{q k}{m} - \frac{i_k-1}{n+1}\bigg)\cdot \frac{x_{[i_k]}}{\sum_{i=1}^n w(x_i) + w(x_{n+j})}.
\$
As a result, we have 
\$ 
&\sum_{\ell=1}^{i_k} \frac{\PP\big(|\bar\cR_{j}|\leq k \biggiven Z_{n+j}=z_{[\ell]}  , \cE_z\big)}{k\cdot w(x_{[\ell]})}
\EE\big[  \ind {\{Z_{n+j}=z_{[\ell]}\}} \ind {\{\bar{p}_j^\nw\leq q k/m\}}  \biggiven \cE_z \big]  \\
&=\sum_{\ell=1}^{i_k-1}  \frac{\PP\big(|\bar\cR_{j}|\leq k \biggiven Z_{n+j}=z_{[\ell]}  , \cE_z\big)}{k\cdot (\sum_{i=1}^n w(x_i) + w(x_{n+j}))} + \frac{(n+1)\frac{q k }{m} - i_k+1}{k\cdot (\sum_{i=1}^n w(x_i) + w(x_{n+j}))}\cdot \PP\big(|\bar\cR_{j}|\leq k \biggiven Z_{n+j}=z_{[i_k]}  , \cE_z\big).
\$
Similar computation for $k+1$  leads to  
\$ 
&\EE\Bigg[ \frac{ \ind {\{|\bar\cR_{j}|\leq k\}} \ind {\{\bar{p}_j^\nw\leq q k/m\}}}{k \cdot w(X_{n+j})} -  \frac{\ind {\{|\bar\cR_{j}|\leq k\}} \ind {\{\bar{p}_j^\nw\leq q (k+1)/m\}}}{(k+1) \cdot w(X_{n+j})} \bigggiven \cE_z \Bigg] \\
&\leq  
\sum_{\ell=1}^{i_k-1}  \frac{\PP\big(|\bar\cR_{j}|\leq k \biggiven Z_{n+j}=z_{[\ell]}  , \cE_z\big)}{k\cdot (\sum_{i=1}^n w(x_i) + w(x_{n+j}))}  
- \sum_{\ell=1}^{i_{k+1}-1}  \frac{\PP\big(|\bar\cR_{j}|\leq k \biggiven Z_{n+j}=z_{[\ell]}  , \cE_z\big)}{(k+1)\cdot (\sum_{i=1}^n w(x_i) + w(x_{n+j}))} \\
&\quad  - 
\frac{(n+1)\frac{q k }{m} - i_k+1}{k\cdot (\sum_{i=1}^n w(x_i) + w(x_{n+j}))}\cdot \PP\big(|\bar\cR_{j}|\leq k \biggiven Z_{n+j}=z_{[i_k]}  , \cE_z\big) \\ 
&\quad 
- \frac{(n+1)\frac{q (k+1) }{m} - i_{k+1}+1}{(k+1)\cdot (\sum_{i=1}^n w(x_i) + w(x_{n+j}))}\cdot \PP\big(|\bar\cR_{j}|\leq k \biggiven Z_{n+j}=z_{[i_{k+1}]}  , \cE_z\big).
\$
Conditional on $\cE_z$,  
$\PP\big(|\bar\cR_{j}|\leq k \biggiven Z_{n+j}=z_{[\ell]}  , \cE_z\big)$
is increasing in $\ell$; 
to see this, simply note that the p-values $(p_1^\nw,\dots,p_{j-1}^\nw, \bar{p}_{j}^\nw,
p_{j+1}^\nw,\dots,p_m^\nw)$ are coordinate-wisely increasing in $\ell$, 
hence the size of the rejection set $\bar\cR_j$ is decreasing in $\ell$.  
Thus, replacing the index $\ell$ in the two summations 
by $i_k$, we have 
\$ 
&\EE\Bigg[ \frac{ \ind {\{|\bar\cR_{j}|\leq k\}} \ind_{\{p_j^*\leq q k/m\}}}{k \cdot w(X_{n+j})} -  \frac{\ind_{\{|\bar\cR_{j}|\leq k\}} \ind_{\{p_j^*\leq q (k+1)/m\}}}{(k+1) \cdot w(X_{n+j})} \bigggiven \cE_z \Bigg] \\
&\leq  
\sum_{\ell=1}^{i_k-1}  \frac{\PP\big(|\bar\cR_{j}|\leq k \biggiven Z_{n+j}=z_{[i_k]}  , \cE_z\big)}{k\cdot (\sum_{i=1}^n w(x_i) + w(x_{n+j}))}  - \sum_{\ell=1}^{i_{k+1}-1}  \frac{\PP\big(|\bar\cR_{j}|\leq k \biggiven Z_{n+j}=z_{[i_k]}  , \cE_z\big)}{(k+1)\cdot (\sum_{i=1}^n w(x_i) + w(x_{n+j}))} \\
&\quad  + \frac{(n+1)\frac{q k }{m} - i_k+1}{k\cdot (\sum_{i=1}^n w(x_i) + w(x_{n+j}))}\PP\big(|\bar\cR_{j}|\leq k \biggiven Z_{n+j}=z_{[i_k]}  , \cE_z\big) \\ 
&\quad + \frac{(n+1)\frac{q (k+1) }{m} - i_{k+1}+1}{(k+1)\cdot (\sum_{i=1}^n w(x_i) + w(x_{n+j}))}\PP\big(|\bar\cR_{j}|\leq k \biggiven Z_{n+j}=z_{[i_{k}]}  , \cE_z\big) \\ 
&= \frac{\PP\big(|\bar\cR_{j}|\leq k \biggiven Z_{n+j}=z_{[i_k]}  , \cE_z\big)}{ \sum_{i=1}^n w(x_i) + w(x_{n+j}) }
\Big(  \frac{(n+1)q k/m }{k} - \frac{(n+1) q (k+1)/m }{k+1} \Big) = 0. 
\$
For any $j$, 
the first term in~\eqref{eq:bh_weight2} can be calculated as 
\$
\EE\big [ \ind {\{\bar{p}_j^\nw\leq q  \}} /w(X_{n+j}) \biggiven \cE_z \big]
&=\sum_{\ell=1}^{i_m } \frac{\PP(  \bar{p}_j^\nw\leq q \given Z_{n+j}=z_{[\ell]}, \cE_z) \cdot \PP(Z_{n+j}=z_{[\ell]}\given \cE_z)}{w(x_{[\ell]})} \\
&= \frac{i_m-1 + (n+1)q - i_m +1 }{\sum_{i=1}^n w(x_i) + w(x_{n+j})} = \frac{(n+1)q }{\sum_{i=1}^n w(x_i) + w(x_{n+j})},
\$
where the second equality uses~\eqref{eq:weighted_exch}, 
and the definition of $\bar{p}_j^\nw$. 
Returning to~\eqref{eq:bh_weight1} and~\eqref{eq:bh_weight2}, we obtain 
\$
&\EE\Bigg[  \frac{\sum_{j=1}^m  \ind\{j\in \cR_{\nw},Y_{n+j}\leq c_{n+j}\}/w(X_{n+j})  }{1\vee |\cR^\nw| }   \Bigg] \\ 
&\leq \frac{1}{m}\sum_{j=1}^m  \EE\big [ \ind {\{\bar{p}_j^\nw\leq q  \}} /w(X_{n+j})  \big]  
= \EE\bigg[ \frac{(n+1)q }{\sum_{i=1}^n w(X_i) + w(X_{n+1})}\bigg],
\$
where the last equality uses the fact that 
the covariates of test data are i.i.d. This concludes the proof of Theorem~\ref{thm:bh_weighted_fdr}. 
\end{proof}

\subsection{Proof of Non-PRDS}
\label{app:subsec_prds}

\begin{proof}[Proof of Proposition~\ref{prop:counter_PRDS}]
  Without loss of generality, we consider the oracle p-values
  \#\label{eq:orc_w_pval}
p_j^* = \frac{\sum_{i=1}^n w(X_i)\ind {\{V_i < {V}_{n+j} \}}+ ( w(X_{n+j}) + \sum_{i=1}^n w(X_i)\ind {\{V_i =  {V}_{n+j} \}}) \cdot U_j }{\sum_{i=1}^n w(X_i) + w(X_{n+j})},
\#
which 
coincide with~\eqref{eq:def_wcpval_rand} when $Y_{n+j}=c_{n+j}$.

We prove by constructing a counterexample and computing 
the distribution functions of p-values to show the violation of 
the PRDS property. 
We set $w(x)=x$ for $P=\textrm{Unif}([1/2,3/2])$, let $Q$ 
have density function $q(x)=x\ind\{1/2\leq x\leq 3/2\}$, 
and $V(x,y) = -w(x)$. 
Take $n=2$ and any arbitrary $m\geq 2$. 
If the weighted conformal p-values are PRDS, we would have 
\$
F(s,t) := \PP(p_1\leq s\given p_2 \geq t) 
\$
is decreasing in $t\in (0,1)$ for any fixed $s\in (0,1)$. 
However, we will show $F(\frac{1}{10},\frac{3}{5}) < F(\frac{1}{10},\frac{9}{10})$. 
We will compute the two values exactly; 
they can also be approximated to arbitrary accuracy 
via Monte Carlo.

For notational simplicity, we write $W_i=w(X_i)$ for $i=1,2$, 
and $\tilde{W}_j = w(X_{2+j})$ for $j=1,2$. 
By construction, $W_1,W_2,\tilde{W}_1,\tilde{W}_2$ are mutually independent with 
$W_1,W_2\sim P$, $\tilde{W}_1,\tilde{W}_2\sim Q$. 
As $P$ and $Q$ do not have point mass, 
there are no ties among the four scores almost surely. 
By the symmetry of $W_1,W_2$, for any $s,t\in (0,1)$, we have 
\$
F(s,t) =(1-t)^{-1} \cdot  \PP(p_1\leq s, p_2\geq t)
= 2(1-t)^{-1}\cdot \PP(p_1\leq s, p_2\geq t, W_1<W_2).
\$
Since $v(x,y) = -w(x)$, the two weighted conformal p-values are equivalently 
\$
p_i = \frac{W_1\ind\{W_1>\tilde{W}_i\} + W_2\ind\{W_2>\tilde{W}_i\} + U_i \tilde{W}_i }{W_1+W_2+\tilde{W}_i},\quad i=1,2.
\$
When $W_1<W_2$, the form of $p_i$ can be divided into three cases:
\begin{enumerate}[(i)]
    \item When $\tilde{W}_i<W_1<W_2$, one has $p_i = \frac{W_1+W_2+U_i\tilde{W}_i}{W_1+W_2+\tilde{W}_i}$, and  $ p_i \in[2/3, 1]$ holds deterministically. 
    \item When $W_1<\tilde{W}_i<W_2$, one has $p_i = \frac{ W_2+U_i\tilde{W}_i}{W_1+W_2+\tilde{W}_i}$. As $W_i,\tilde{W}_i\in [1/2,3/2]$, it holds that $p_i\in[1/3,6/7]$.
    \item When $W_1<W_2<\tilde{W}_i$, one has $p_i = \frac{ U_i\tilde{W}_i}{W_1+W_2+\tilde{W}_i}$, and $p_i\in [0,3/5]$.
\end{enumerate}

Intuitively, the reason for violating PRDS in this case is that, 
when $p_1\leq 1/10$ is small, it must always take the form 
$
p_1 = \frac{ U_1\cdot \tilde{W}_1 }{W_1+W_2+\tilde{W}_1} 
$
as  in case (iii). 
Meanwhile, when $p_2\geq 3/5$ is relatively large, it may 
be as cases (i) or (ii). In that two cases, roughly, 
a larger $p_2$ means relatively larger $W_1,W_2$, hence smaller $p_1$. 

We now conduct the exact computation. For $t=9/10$, 
when $W_1<W_2$, we note that 
it must be $p_2 = \frac{W_1+W_2+U_2\tilde{W}_2}{W_1+W_2+\tilde{W}_2}$, 
i.e., $\tilde{W}_2<W_1<W_2<\tilde{W}_1$. Therefore, 
\#\label{eq:prob1}
&\PP(p_1\leq 1/10, p_2\geq 9/10, W_1<W_2)\notag \\ 
&= \PP\big(p_1\leq 1/10, p_2\geq 9/10, \tilde{W}_2<W_1<W_2<\tilde{W}_1\big)\notag  \\ 
&=\PP\Bigg(  \frac{ U_1\cdot \tilde{W}_1 }{W_1+W_2+\tilde{W}_1} \leq \frac{1}{10}, ~\frac{W_1+W_2+U_2\tilde{W}_2}{W_1+W_2+\tilde{W}_2} \geq \frac{9}{10} , ~ \tilde{W}_2<W_1<W_2<\tilde{W}_1\Bigg)\notag \\ 
&= \PP\Bigg(  \frac{ U_1\cdot \tilde{W}_1 }{W_1+W_2+\tilde{W}_1} \leq \frac{1}{10}, ~\frac{W_1+W_2+U_2\tilde{W}_2}{W_1+W_2+\tilde{W}_2} \geq \frac{9}{10} , ~ \tilde{W}_2<W_1<W_2<\tilde{W}_1\Bigg) \notag\\ 
&= \PP\Bigg(  U_1\leq \frac{W_1+W_2+\tilde{W}_1}{  10 \tilde{W}_1 }  , ~
1-U_2\geq \frac{W_1+W_2+\tilde{W}_2}{ 10\tilde{W}_2} , ~ \tilde{W}_2<W_1<W_2<\tilde{W}_1\Bigg).
\#
We note that in this case $\frac{W_1+W_2+\tilde{W}_1}{  10 \tilde{W}_1 },\frac{W_1+W_2+\tilde{W}_2}{ 10\tilde{W}_2} \in [0,1]$ deterministically. 
Since $U_1,U_2$ are mutually independent, and also independent of 
everything else, by the tower property, 
\$
\eqref{eq:prob1} =  \EE\Bigg[    \frac{W_1+W_2+\tilde{W}_1}{  10 \tilde{W}_1 } \cdot \frac{W_1+W_2+\tilde{W}_2}{ 10\tilde{W}_2} \cdot \ind\{ \tilde{W}_2<W_1<W_2<\tilde{W}_1\} \Bigg].
\$
By the symmetry of $W_1,W_2$, we know 
\#\label{eq:prob12}
&\PP(p_1\leq 1/10, p_2\geq 9/10 )  = 2\times \eqref{eq:prob1} \notag  \\ 
&= \EE\Bigg[    \frac{W_1+W_2+\tilde{W}_1}{  10 \tilde{W}_1 } \cdot \frac{W_1+W_2+\tilde{W}_2}{ 10\tilde{W}_2} \cdot \ind\{ W_1,W_2 \in[\tilde{W}_2,\tilde{W}_1], \tilde{W}_2 < \tilde{W}_1]\} \Bigg].
\#
Note that conditional on $\tilde{W}_2,\tilde{W}_1$ 
and the event $\{  W_1,W_2 \in[\tilde{W}_2,\tilde{W}_1]\}$, 
one has $W_1,W_2\sim \textrm{Unif}([\tilde{W}_1,\tilde{W}_2])$ and are mutually independent. 
Applying tower property again, we have 
\$
\eqref{eq:prob12} &= \EE\Bigg[    \frac{13\tilde{W}_1^2/6 + 13\tilde{W}_2^2/6 + 14\tilde{W}_1\tilde{W}_2/3 }{  10 0\tilde{W}_1 \cdot \tilde{W}_2 } \cdot  \PP\big(   W_1,W_2 \in[\tilde{W}_2,\tilde{W}_1]\biggiven \tilde{W}_1,\tilde{W}_2\big) \cdot \ind\{\tilde{W}_2<\tilde{W}_1\} \Bigg] \\ 
&=\EE\Bigg[    \frac{13\tilde{W}_1^2/6 + 13\tilde{W}_2^2/6 + 14\tilde{W}_1\tilde{W}_2/3 }{  10 0\tilde{W}_1 \cdot \tilde{W}_2 } \cdot  (\tilde{W}_1-\tilde{W}_1)^2 \cdot \ind\{\tilde{W}_2<\tilde{W}_1\} \Bigg] = \frac{547}{72000},
\$
where the exact number can be computed using the distribution of $\tilde{W}_1$
and $\tilde{W}_2$. We thus have 
\$
F(1/10, 9/10) = \frac{547}{7200}.
\$

For $t=3/5$, when $p_2\geq t$  it could either be case (i) or (ii). Thus
\$
\PP(p_1\leq 1/10, p_2\geq 3/5, W_1<W_2)
&= \PP(p_1\leq 1/10, p_2\geq 3/5, \tilde{W}_2<W_1<W_2<\tilde{W}_1) \\
&\qquad +\PP(p_1\leq 1/10, p_2\geq 3/5, W_1<\tilde{W}_2<W_2<\tilde{W}_1).
\$
For the first case, 
\#\label{eq:prob2}
& \PP(p_1\leq 1/10, p_2\geq 3/5, \tilde{W}_2<W_1<W_2<\tilde{W}_1) \notag \\
&=  \PP\Bigg(  U_1\leq \frac{W_1+W_2+\tilde{W}_1}{  10 \tilde{W}_1 }  , ~
1-U_2\geq \frac{2W_1+2W_2+2\tilde{W}_2}{ 5\tilde{W}_2} , ~ \tilde{W}_2<W_1<W_2<\tilde{W}_1\Bigg). 
\#
One can check that $\frac{2W_1+2W_2+2\tilde{W}_2}{ 5\tilde{W}_2} \geq \frac{6}{5}$ 
when $\tilde{W}_2<W_1<W_2$, and that 
$\frac{W_1+W_2+\tilde{W}_1}{  10 \tilde{W}_1 } \in[0,1]$, hence by the tower property 
and the symmetry of $W_1,W_2$, we know 
\$
\eqref{eq:prob2} &= 
\EE\Bigg[   \frac{W_1+W_2+\tilde{W}_1}{  10 \tilde{W}_1 } \cdot \ind\{ \tilde{W}_2<W_1<W_2<\tilde{W}_1 \}\Bigg] \\ 
&= \frac{1}{2}\cdot \EE\Bigg[  
    \EE\bigg[  \frac{W_1+W_2+\tilde{W}_1}{  10 \tilde{W}_1 } \bigggiven  W_1,W_2\in[ \tilde{W}_2,\tilde{W}_1 ] ,\tilde{W}_1,\tilde{W}_2 \bigg] \cdot  \PP\big(  W_1,W_2\in[ \tilde{W}_2,\tilde{W}_1 ] \biggiven \tilde{W}_1,\tilde{W}_2\big) \cdot \ind\{\tilde{W}_2<\tilde{W}_1 \}\Bigg] \\ 
&=\frac{1}{2}\cdot \EE\Bigg[  
    \EE\bigg[  \frac{W_1+W_2+\tilde{W}_1}{  10 \tilde{W}_1 } \bigggiven  W_1,W_2\in[ \tilde{W}_2,\tilde{W}_1 ] ,\tilde{W}_1,\tilde{W}_2 \bigg] \cdot (\tilde{W}_1-\tilde{W}_2)^2  \cdot \ind\{\tilde{W}_2<\tilde{W}_1 \}\Bigg] 
\$
Conditional on $\tilde{W}_1,\tilde{W}_2$ and the event $\{  W_1,W_2\in[ \tilde{W}_2,\tilde{W}_1 ] \}$, we know $W_1,W_2\sim \textrm{Unif}([\tilde{W}_2,\tilde{W}_1])$, and  
are mutually independent. Applying the tower property again, we have 
\$
\eqref{eq:prob2} &= 
\frac{1}{2}\cdot \EE\Bigg[  
        \frac{2 \tilde{W}_1 + \tilde{W}_2 }{  10 \tilde{W}_1 } \cdot (\tilde{W}_1-\tilde{W}_2)^2\cdot \ind\{ \tilde{W}_2<\tilde{W}_1 \}\Bigg] = \frac{47}{4800}.
\$
For the second case, 
\#\label{eq:prob3}
& \PP(p_1\leq 1/10, p_2\geq 3/5, W_1<\tilde{W}_2<W_2<\tilde{W}_1) \notag \\
&=  \PP\Bigg(  U_1\leq \frac{W_1+W_2+\tilde{W}_1}{  10 \tilde{W}_1 }  , ~
1-U_2\geq \frac{-3W_1+2W_2+2\tilde{W}_2}{ 5\tilde{W}_2} , ~ W_1<\tilde{W}_2<W_2<\tilde{W}_1\Bigg). 
\#
One can check that $\frac{W_1+W_2+\tilde{W}_1}{  10 \tilde{W}_1 },\frac{-3W_1+2W_2+2\tilde{W}_2}{ 5\tilde{W}_2}\in[0,1]$. The tower property thus implies 
\$
\eqref{eq:prob3} &= \EE\Bigg[   \frac{W_1+W_2+\tilde{W}_1}{  10 \tilde{W}_1 }\cdot \frac{-3W_1+2W_2+2\tilde{W}_2}{ 5\tilde{W}_2}\cdot \ind\{ W_1<\tilde{W}_2<W_2<\tilde{W}_1\}\Bigg]\\
&= \EE\Bigg[  \EE\bigg[  \frac{W_1+W_2+\tilde{W}_1}{  10 \tilde{W}_1 }\cdot \frac{-3W_1+2W_2+2\tilde{W}_2}{ 5\tilde{W}_2} \bigggiven W_1<\tilde{W}_2<W_2<\tilde{W}_1, \tilde{W}_1,\tilde{W}_2\bigg] \\ 
&\qquad \qquad \cdot \PP\big( W_1<\tilde{W}_2<W_2<\tilde{W}_1 \biggiven \tilde{W}_1,\tilde{W}_2\big)  \cdot \ind\{\tilde{W}_2<\tilde{W}_1\} \Bigg] = 
\$
Note that conditional on $\tilde{W}_1,\tilde{W}_2$ 
and the event $\{W_1<\tilde{W}_2<W_2<\tilde{W}_1\}$, 
$W_1$ and $W_2$ are mutually independent with 
$W_1\sim \textrm{Unif}([1/2,\tilde{W}_2])$ and $W_2\sim \textrm{Unif}([\tilde{W}_2,\tilde{W}_1])$. We thus have 
\$
\eqref{eq:prob3} &= \EE\Bigg[  \frac{5\tilde{W}_1^2+35\tilde{W}_1\tilde{W}_2/12+17\tilde{W}_2^2/12- 7\tilde{W}_1/8-\tilde{W}_2/8-1/4}{50\tilde{W}_1\cdot \tilde{W}_2} \cdot (\tilde{W}_2-1/2)\cdot(\tilde{W}_1-\tilde{W}_2) \cdot \ind\{\tilde{W}_2<\tilde{W}_1\}  \Bigg].
\$
Using the distribution of $\tilde{W}_1,\tilde{W}_2$, we have 
$\eqref{eq:prob3} = \frac{361}{72000}$. Combining the two cases, we have 
\$
\PP(p_1\leq 1/10, p_2\geq 3/5 ) = 
2\cdot \PP(p_1\leq 1/10, p_2\geq 3/5, W_1<W_2)
= \frac{2\cdot 47}{4800} + \frac{2\cdot 361}{72000} = \frac{533}{18000}.
\$
We thus have 
\$
F(1/10, 3/5) = \frac{533}{7200} < \frac{547}{7200} = F(1/10,9/10).
\$
Therefore, $p_1$ and $p_2$ are not PRDS, concluding the proof of Proposition~\ref{prop:counter_PRDS}.
\end{proof}

\subsection{Proof of Theorem~\ref{thm:fdr_asymp}}
\label{app:thm_fdr_asymp}

\begin{proof}[Proof of Theorem~\ref{thm:fdr_asymp}]
We use a representation of the BH($q$) procedure  
communicated in~\cite{storey2004strong}: 
the rejection set of BH($q$) procedure 
is $\cR = \{j\colon p_j \leq \hat\tau\}$, where 
\#\label{eq:BH_hat_tau}
\hat\tau = \sup \bigg\{ t\in [0,1]\colon \frac{mt}{\sum_{j=1}^m \ind\{p_j \leq t\}} \leq q  \bigg\}.
\#
To clarify the dependence on the calibration data, we denote the $p$-values as 
$
p_j = \hat{F}_n( V_{j}^{0}, U_j, W_j),
$ 
where for simplicity we denote $V_j^0 = \hat{V}_{n+j}=V(X_{n+j},C_{n+j} )$ and 
$W_j = w(X_{n+j})$. As we suppose $\{C_{n+j}\}_{j=1}^n$ are i.i.d.~random varibles, 
here $\{V_j^0\}_{j=1}^m$ are also i.i.d.~random variables. 
Finally, we define 
\$
\hat{F}_n(v,u,w) = \frac{\sum_{i=1}^n w(X_i)\ind {\{V_i < v \}}+  ( w + \sum_{i=1}^n w(X_i)\ind {\{V_i = v \}})\cdot u }{\sum_{i=1}^n w(X_i) + w}
\$
for any $(v,u,w)\in \RR\times[0,1]\times \RR^+$. 
Note that conditional on the training process, 
 $\{(V_j^0,U_j,W_j)\colon 1\leq j\leq m\}$
are i.i.d., and independent of $\cD_\calib$. 
We first show the uniform concentration of $\hat{F}_n$ to $F$ defined in~\eqref{eq:def_F}. 

\begin{lemma}\label{lem:pval_conv}
Let $\hat{F}_n$ and $F$ defined as above. Then for any fixed constant $M>0$, 
as $n\to \infty$, we have 
\$
\sup_{v\in \RR, u\in [0,1], w\in [0,M]} \big|  \hat{F}_n(v,u,w) - F(v,u) \big| ~ \asto ~ 0.
\$
\end{lemma}

\begin{proof}[Proof of Lemma~\ref{lem:pval_conv}]
The strong law of large numbers ensures $\frac{1}{n}\sum_{i=1}^n w(X_i) ~\asto~ 1$. Thus by definition, 
\$
&\bigg|  \hat{F}_n(v,u,w) - \frac{\sum_{i=1}^n w(X_i)\ind {\{V_i < v \}}+  (  \sum_{i=1}^n w(X_i)\ind {\{V_i = v \}})\cdot u }{\sum_{i=1}^n w(X_i) } \bigg| \\
&\leq \frac{w}{\sum_{i=1}^n w(X_i) + w} + w\cdot \frac{\sum_{i=1}^n w(X_i)\ind {\{V_i < v \}}+  ( w + \sum_{i=1}^n w(X_i)\ind {\{V_i = v \}})\cdot u }{(\sum_{i=1}^n w(X_i)) \cdot (\sum_{i=1}^n w(X_i) + w)} \\ 
&\leq \frac{w}{\sum_{i=1}^n w(X_i) } +  \frac{w }{\sum_{i=1}^n w(X_i)},\quad \forall v\in \RR, u\in [0,1], w\in [0,M].
\$
This implies 
\$
\sup_{v\in \RR, u\in [0,1], w\in [0,M]} \bigg|  \hat{F}_n(v,u,w) - \frac{\sum_{i=1}^n w(X_i)\ind {\{V_i < v \}}+  (  \sum_{i=1}^n w(X_i)\ind {\{V_i = v \}})\cdot u }{\sum_{i=1}^n w(X_i) } \bigg| ~\asto ~ 0.
\$
On the other hand, the uniform law of large numbers in Lemma~\ref{lem:ulln} implies 
\$
\sup_{v\in \RR, u\in [0,1], w\in [0,M]} \bigg|  \frac{\sum_{i=1}^n w(X_i)\ind {\{V_i < v \}}+  (  \sum_{i=1}^n w(X_i)\ind {\{V_i = v \}})\cdot u }{\sum_{i=1}^n w(X_i) } - F(v,u) \bigg| ~\asto ~ 0.
\$
Putting this two almost sure convergence together completes the proof of Lemma~\ref{lem:pval_conv}. 
\end{proof}

Lemma~\ref{lem:pval_conv} forms the basis of proving asymptotic results. 
In the following, we study the two asymptotic regimes (i) $n\to \infty$ with $m$ fixed and (ii) $m,n\to \infty$, separately.

\paragraph{Case (i): $m$ fixed, $n\to \infty$.}
By Lemma~\ref{lem:pval_conv}, we know that  
$\{p_j\}$ converges almost surely (and uniformly over all $j=1,\dots,m$) to 
$\{p_j^\infty\}$, where we define 
\$
p_j^\infty := F(V_j^0, U_j), 
\$
with $F$ defined in~\eqref{eq:def_F} of Theorem~\ref{thm:fdr_asymp}. 
Note that $\{p_j^\infty\}$ are i.i.d., and all of them have no point mass 
since $\{U_j\} \sim $ Unif([0,1]) are 
independent of the observations. 
In the following, 
we denote as $\cR$  the rejection set $\cR$ 
of BH($q$) with $\{p_j\}$. We are to show that $\cR$ 
is asymptotically the same as 
$\cR^\infty$, the rejection set 
obtained from BH($q$) with $\{p_j^\infty\}$, in the sense that 
$\limsup_{n\to \infty}\PP(\cR\neq \cR^\infty)=0$. 
Note that 
\$
\PP(\cR\neq \cR^\infty) &\leq \sum_{j=1}^m \PP(j \in \cR, j\notin \cR^\infty) + \sum_{j=1}^m \PP(j \notin \cR, j\in \cR^\infty) \\ 
&\leq \sum_{j=1}^m \sum_{\ell=0}^m \PP(p_j\leq \ell k/m,~ p_j^\infty > \ell k/m) 
+ \sum_{j=1}^m \sum_{\ell=0}^m \PP(p_j> \ell k/m,~ p_j^\infty \leq \ell k/m).
\$
For any fixed $\delta>0$, we have 
\$
&\limsup_{n\to \infty} \PP(p_j\leq \ell k/m,~ p_j^\infty > \ell k/m) \\
&
\leq \limsup_{n\to \infty} \PP(p_j\leq \ell k/m  ,~ p_j^\infty > \ell k/m + \delta ) + 
\limsup_{n\to \infty} \PP( \ell k/m  < p_j^\infty\leq  \ell k/m + \delta) \\ 
&= \limsup_{n\to \infty} \PP( \ell k/m  < p_j^\infty\leq  \ell k/m + \delta).
\$ 
Taking $\delta \to 0$ and noting that $p_j^*$ does not have point mass at $\{\ell k/m\}$, 
we have $\limsup_{n\to \infty} \PP(p_j\leq \ell k/m,~ p_j^\infty > \ell k/m) = 0$. 
Similar arguments also yield 
$\limsup_{n\to \infty} \PP(p_j> \ell k/m,~ p_j^\infty \leq \ell k/m) = 0$. 
Summing over finitely many terms $1\leq j\leq m$ and $0\leq \ell \leq m$ 
leads to 
\# \label{eq:equiv_RR}
\PP(\cR\neq \cR^\infty) \to 0.
\# Let 
$\fdr(\cR)$, $\fdr(\cR^\infty)$ 
denote the FDR for $\cR$ and $\cR^\infty$, respectively.  
As a result of~\eqref{eq:equiv_RR}, we have 
\$
\limsup_{n\to \infty} \fdr(\cR) 
&\leq \limsup_{n\to \infty} \EE\big[  \fdp(\cR) \cdot \ind\{\cR = \cR^\infty \} \big] 
+ \limsup_{n\to \infty} \EE\big[ \fdp(\cR) \cdot \ind\{\cR \neq \cR^\infty \} \big] \\ 
&\leq  \EE\big[  \fdp(\cR^\infty) \big] 
+ \limsup_{n\to \infty} \EE\big[ \ind\{\cR \neq \cR^\infty \} \big] 
\leq \fdr(\cR^\infty),
\$
where the second inequality uses the fact that $\fdp(\cR)\leq 1$, 
and the last inequality uses $\EE\big[  \fdp(\cR^\infty) \big] = \fdr(\cR^\infty)$. 
For each $j$, we let $\cR_{j\to 0}^\infty$ denote 
the rejection set of BH($q$) applied to the set of p-values
$\{p_1^\infty,\dots,p_{j-1}^\infty, 0, p_{j+1}^\infty,\dots, p_m^\infty\}$, i.e.,
when we set $p_j^\infty$ to $0$ while keeping others $\{p_\ell^\infty\}_{\ell\neq j}$ fixed. 
Then we know that $\cR^\infty = \cR_{j\to 0}^\infty$ 
on the event  $\{j\in \cR^\infty\}$. Therefore, 
the FDR for $\cR^\infty$ can be bounded as 
\$
\fdr(\cR^\infty) &= \EE\bigg[  \frac{\sum_{j=1}^m \ind\{j\in \cR^\infty, j\in \cH_0\}}{1\vee |\cR^\infty|}   \bigg]   
= \sum_{j=1}^m \EE\bigg[  \frac{ \ind\{j\in \cR^\infty , j\in \cH_0\}}{1\vee |\cR^\infty_{j\to 0}|}   \bigg]. 
\$
By the definition of BH($q$) procedure, 
we know $j\in \cR^\infty$ implies $p_j^\infty \leq q|\cR^\infty|/m$. 
In the right-handed side of 
the above equation, each term can be decomposed as 
\$
\EE\bigg[  \frac{ \ind\{j\in \cR^\infty , j\in \cH_0\}}{1\vee |\cR^\infty_{j\to 0}|}   \bigg]
&= \sum_{k=1}^m \EE\bigg[ \ind\{|\cR^\infty | = k\}  \frac{ \ind\{j\in \cR^\infty, j\in \cH_0\}}{ k }   \bigg] \\
&\leq \sum_{k=1}^m \EE\bigg[ \ind\{|\cR_{j\to 0}^\infty | = k\}  \frac{ \ind\{ p_j^\infty \leq \frac{qk}{m}, Y_{n+j} \leq C_{n+j} \}}{ k }   \bigg] \\ 
&= \sum_{k=1}^m  \PP\big( |\cR_{j\to 0}^\infty | = k \big) \cdot \frac{1}{k}\PP\Big( p_j^\infty \leq \frac{qk}{m}, Y_{n+j} \leq C_{n+j}\Big),
\$
where the first equality uses the fact that 
$(p_j^\infty,Y_{n+j} ,C_{n+j} )$
are independent of all other samples and hence 
also independent of $\cR_{j\to 0}^\infty$; 
the second line uses the fact that $\{|\cR^\infty|=k,~j\in \cR^\infty\}$ implies $\{|\cR_{j\to 0}^\infty|=k,~p_j^\infty \leq qk/m\}$. 
Due to the monotonicity of $V(x,y)$ in $y$ 
and the monotonicity of $F(v,u)$ in $v$, 
\$
\PP\Big( p_j^\infty \leq \frac{qk}{m}, Y_{n+j} \leq C_{n+j} \Big)
&=
\PP\Big( F\big(V(X_{n+j},C_{n+j} ),U_j\big) \leq \frac{qk}{m}, Y_{n+j} \leq C_{n+j} \Big) \\ 
&\leq \PP\Big( F\big(V(X_{n+j},Y_{n+j} ),U_j\big) \leq \frac{qk}{m}, Y_{n+j} \leq C_{n+j} \Big) \\ 
& \leq \PP\Big( F\big(V(X_{n+j},Y_{n+j} ),U_j\big) \leq \frac{qk}{m} \Big) = \frac{qk}{m},
\$
where the last equality uses the fact that $F(V(X_{n+j},Y_{n+j} ),U)\sim$ Unif[0,1]
from the definition of $F$, as all the probabilities are taken with respect to 
the test distribution 
$\QQ_{X,Y,C}$ and an independent random variable $U\sim$ Unif([0,1]). 
Therefore, we obtain 
\$\fdr(\cR^\infty) \leq \sum_{j=1}^m \sum_{k=1}^m \PP\big( |\cR_{j\to 0}^\infty | = k \big) \cdot \frac{q}{m} \leq q,
\$ 
which completes the proof of this case.

\paragraph{Case (ii): $m,n\to \infty$.} 
In this case, we are to 
repeatedly employ the (uniform) strong law of large numbers to show the asymptotic behavior 
of the testing procedure. 
Based on Lemma~\ref{lem:pval_conv}, we show the uniform convergence of the criterion 
in~\eqref{eq:BH_hat_tau}. 

\begin{lemma}\label{lem:tau_conv}
Under the above setup, suppose $\sup_{x\in \cX}w(x) \leq M$ for some constant $M>0$. 
Then  
\$
\sup_{t\in[0,1]} \bigg| \frac{1}{m} \sum_{j=1}^m \ind\{\hat{F}_n(V_j^0,U_j,W_j)\leq t\}  
- \PP\big( F(V_j^0, U_j) \leq t\big)  \bigg| ~\asto ~ 0,
\$
as $m,n\to \infty$, 
where $\PP$ is taken with respect to $V_j^0 = V(X_{n+j},C_{n+j})$ and an independent $U_j\sim$ Unif$([0,1])$. 
\end{lemma}

\begin{proof}[Proof of Lemma~\ref{lem:tau_conv}]
Let $0=t_0<t_1<\cdots<t_K =1$ be a partition of $[0,1]$. 
Then for each $t\in [0,1]$, 
there exists some $k$ such that $t_k \leq t < t_{k+1}$, hence 
\#
& \frac{1}{m} \sum_{j=1}^m \ind\{\hat{F}_n(V_j^0,U_j,W_j)\leq t_{k}\}  
- \frac{1}{m} \sum_{j=1}^m \ind\{F(V_j^0,U_j)\leq t_{k+1}\} \label{eq:asymp_l}\\
&\leq \frac{1}{m} \sum_{j=1}^m \ind\{\hat{F}_n(V_j^0,U_j,W_j)\leq t\}  
- \frac{1}{m} \sum_{j=1}^m \ind\{F(V_j^0,U_j)\leq t\} \notag  \\
&\leq  \frac{1}{m} \sum_{j=1}^m \ind\{\hat{F}_n(V_j^0,U_j,W_j)\leq t_{k+1}\}  
- \frac{1}{m} \sum_{j=1}^m \ind\{F(V_j^0,U_j)\leq t_k\}. \label{eq:asymp_u}
\#
The upper bound~\eqref{eq:asymp_u} above can be further bounded as 
\#\label{eq_asymp_upper_1}
& \bigg|  \frac{1}{m} \sum_{j=1}^m \ind\{\hat{F}_n(V_j^0,U_j,W_j)\leq t_{k+1}\}  
- \frac{1}{m} \sum_{j=1}^m \ind\{F(V_j^0,U_j)\leq t_k\}   \bigg| \notag \\
&\leq \bigg|  \frac{1}{m} \sum_{j=1}^m \ind\{\hat{F}_n(V_j^0,U_j,W_j)\leq t_{k+1}\}  
- \frac{1}{m} \sum_{j=1}^m \ind\{F(V_j^0,U_j)\leq t_{k+1}\}   \bigg| \notag \\
&\quad + \bigg|  \frac{1}{m} \sum_{j=1}^m \ind\{F(V_j^0,U_j)\leq t_{k+1}\}  
- \frac{1}{m} \sum_{j=1}^m \ind\{F(V_j^0,U_j)\leq t_{k}\}   \bigg|.
\#
For any fixed $\delta>0$, we have 
\$
& \bigg|  \frac{1}{m} \sum_{j=1}^m \ind\{\hat{F}_n(V_j^0,U_j,W_j)\leq t_{k+1}\}  
- \frac{1}{m} \sum_{j=1}^m \ind\{F(V_j^0,U_j)\leq t_{k+1}\}   \bigg| \\ 
& \leq \frac{1}{m} \sum_{j=1}^m \Big( \ind\{\hat{F}_n(V_j^0,U_j,W_j)\leq t_{k+1}, F(V_j^0,U_j)> t_{k+1} \}  
 +  \ind\{\hat{F}_n(V_j^0,U_j,W_j) > t_{k+1}, F(V_j^0,U_j) \leq t_{k+1} \}  \Big) \\ 
& \leq \ind\big\{\sup_j \big|\hat{F}_n(V_j^0,U_j,W_j) -F(V_j^0,U_j)\big|\geq \delta \big\} 
+ \frac{1}{m} \sum_{j=1}^m   \ind\{\hat{F}_n(V_j^0,U_j,W_j)\leq t_{k+1}, t_{k+1}+\delta \geq F(V_j^0,U_j)> t_{k+1}  \}  \\
&\quad 
 + \frac{1}{m} \sum_{j=1}^m  \ind\{\hat{F}_n(V_j^0,U_j,W_j) > t_{k+1}, t_{k+1}- \delta < F(V_j^0,U_j) \leq t_{k+1} \} \\ 
&\leq  \ind\big\{\sup_j \big|\hat{F}_n(V_j^0,U_j,W_j) -F(V_j^0,U_j)\big|\geq \delta \big\} 
+ \frac{1}{m} \sum_{j=1}^m   \ind\{ t_{k+1} - \delta \leq F(V_j^0,U_j) \leq t_{k+1} +\delta \} .
\$
Lemma~\ref{lem:pval_conv} implies $\limsup_{n\to \infty} \ind\big\{\sup_j \big|\hat{F}_n(V_j^0,U_j,W_j) -F(V_j^0,U_j)\big|\geq \delta \big\} = 0$ almost surely. 
Combining 
this observation with the decomposition~\eqref{eq_asymp_upper_1}, we have 
\$
& \limsup_{n\to \infty} \bigg|  \frac{1}{m} \sum_{j=1}^m \ind\{\hat{F}_n(V_j^0,U_j,W_j)\leq t_{k+1}\}  
- \frac{1}{m} \sum_{j=1}^m \ind\{F(V_j^0,U_j)\leq t_k\}   \bigg| \\
& \leq  \frac{1}{m} \sum_{j=1}^m   \ind\{ t_{k+1} - \delta \leq F(V_j^0,U_j) \leq t_{k+1} +\delta \} 
 +  \frac{1}{m} \sum_{j=1}^m \ind\{t_k < F(V_j^0,U_j)\leq t_{k+1}\}   
\$
almost surely. 
Again invoking the (uniform) law of large numbers (c.f.~Lemma~\ref{lem:ulln} applied without weights to i.i.d.~random variables $F(V_j^0,U_j)$),  
\$
\limsup_{m\to \infty} \sup_k &\bigg|\frac{1}{m} \sum_{j=1}^m   \ind\{ t_{k+1} - \delta \leq F(V_j^0,U_j) \leq t_{k+1} +\delta \} 
 - \PP\big(  t_{k+1} - \delta \leq F(V_j^0,U_j) \leq t_{k+1} +\delta  \big) \\ 
 &\quad +  \frac{1}{m} \sum_{j=1}^m \ind\{t_k < F(V_j^0,U_j)\leq t_{k+1}\} - 
 \PP\big( t_k < F(V_j^0,U_j)\leq t_{k+1} \big) \bigg|  = 0
\$
with probability one. We thus have 
\#\label{eq:asymp_upp}
& \limsup_{m,n\to \infty} \sup_k \bigg|  \frac{1}{m} \sum_{j=1}^m \ind\{\hat{F}_n(V_j^0,U_j,W_j)\leq t_{k+1}\}  
- \frac{1}{m} \sum_{j=1}^m \ind\{F(V_j^0,U_j)\leq t_k\}   \bigg| \notag \\
& \leq \sup_{k} \Big|\PP\big(  t_{k+1} - \delta \leq F(V_j^0,U_j) \leq t_{k+1} +\delta  \big) + \PP\big( t_k < F(V_j^0,U_j)\leq t_{k+1} \big) \Big|
\#
for any partition $\{t_k\}$ and $\delta>0$. 
On the other hand, we note that since $U_j \sim $ Unif[0,1] is 
independent of the observations, $F(V_j^0,U_j)$ are i.i.d.~with continuous distributions. 
Letting $\delta \to 0$ 
so that $\{t_k\}$ be fine enough sends the supremum in~\eqref{eq:asymp_upp} 
to zero. With similar arguments, we can show for the lower bound~\eqref{eq:asymp_l} that 
\$
\limsup_{m,n\to \infty} \sup_k \bigg| \frac{1}{m} \sum_{j=1}^m \ind\{\hat{F}_n(V_j^0,U_j,W_j)\leq t_{k}\}  
- \frac{1}{m} \sum_{j=1}^m \ind\{F(V_j^0,U_j)\leq t_{k+1}\} \bigg| = 0
\$ 
almost surely. 
Combining the above two results, we have 
\#\label{eq_asymp_1}
\sup_{t\in[0,1]} \bigg|  \frac{1}{m} \sum_{j=1}^m \ind\{\hat{F}_n(V_j^0,U_j,W_j)\leq t\}  
- \frac{1}{m} \sum_{j=1}^m \ind\{F(V_j^0,U_j)\leq t\}  \bigg| ~\asto ~ 0. 
\#
Invoking the uniform strong law of large numbers to i.i.d.~random 
variables $\{F(V_j^0,U_j)\}_{j=1}^m$, we have 
\$
\sup_{t\in[0,1]} \bigg|   \frac{1}{m} \sum_{j=1}^m \ind\{F(V_j^0,U_j)\leq t\} - \PP\big( F(V_j^0, U_j) \leq t\big)  \bigg| ~\asto ~ 0,
\$
hence by the triangular inequality we complete the proof of Lemma~\ref{lem:tau_conv}. 
\end{proof}

Recall that in Theorem~\ref{thm:fdr_asymp}, 
we assumed that there exists some $t' \in (0,1]$ such that
\$
\frac{\PP\big( F(V_j^0, U_j) \leq t'\big)}{t'} > \frac{1}{q}.
\$
We then define 
\$
t^* = \sup\bigg\{ t\in [0,1]\colon  \frac{\PP\big( F(V_j^0, U_j) \leq t \big)}{t }   \geq \frac{1}{q}    \bigg\}
 = \sup\big\{  t\in[0,1]\colon G_\infty(t) \leq q \big\},
\$
where $G_\infty(\cdot)$ is defined in Theorem~\ref{thm:fdr_asymp}. 
It is well-defined and $t^*\geq t'$. 
Fix any $\delta\in (0, t')$. By Lemma~\ref{lem:tau_conv}, 
\#\label{eq:unif_conv_tau}
\sup_{t\in[\delta,1]} \bigg|   \frac{\sum_{j=1}^m \ind\{F(V_j^0,U_j)\leq t\}}{mt}  - \frac{\PP\big( F(V_j^0, U_j) \leq t\big)}{t}  \bigg| ~\asto ~ 0.
\#
In particular, $\frac{\sum_{j=1}^m \ind\{F(V_j^0,U_j)\leq t'\}}{mt'} ~\asto~ \frac{\PP\big( F(V_j^0, U_j) \leq t'\big)}{t'} > \frac{1}{q}$, 
hence $\hat\tau \geq t' \geq \delta$ eventually 
as $m,n\to \infty$. 
As a result, by the definition of $\hat\tau$ and Fatou's lemma, we have 
\$
\limsup_{m,n\to\infty} \fdr &
= \limsup_{m,n\to \infty} \EE\Bigg[\frac{\sum_{j=1}^m  \ind\{\hat{F}_n(V_j^0, U_j,W_j) \leq \hat\tau, j\in \cH_0\}}{1\vee \sum_{j=1}^m \ind\{\hat{F}_n(V_j^0, U_j,W_j) \leq \hat\tau\}} \Bigg]  \\ 
&\leq  q \cdot \EE\Bigg[\limsup_{m,n\to \infty} \frac{1 }{m \hat \tau}  \sum_{j=1}^m  \ind\{\hat{F}_n(V_j^0, U_j,W_j) \leq \hat\tau, j\in \cH_0\} \Bigg].
\$
The uniform concentration result~\eqref{eq_asymp_1} in the proof of Lemma~\ref{lem:tau_conv} implies 
\$
\limsup_{m,n\to \infty} \bigg| \frac{1 }{m \hat\tau }  \sum_{j=1}^m  \ind\{\hat{F}_n(V_j^0, U_j,W_j) \leq \hat\tau, j\in \cH_0\} - \frac{1 }{m \hat\tau } \sum_{j=1}^m  \ind\{F(V_j^0, U_j,W_j) \leq \hat\tau, j\in \cH_0\} \bigg|
= 0
\$
almost surely. By the monotonicity of the nonconformity score $V$, we know 
\$
&\frac{1 }{m}   \sum_{j=1}^m  \ind\{F(V_j^0, U_j,W_j) \leq \hat\tau, j\in \cH_0\} \\
&= \frac{1 }{m}  \sum_{j=1}^m  \ind \big\{F( V(X_{n+j},C_{n+j}  ), U_j) \leq \hat\tau,~ Y_{n+j} \leq C_{n+j} \big\}  \\ 
&\leq \frac{1 }{m}  \sum_{j=1}^m  \ind \big\{F( V(X_{n+j},Y_{n+j} ), U_j) \leq \hat\tau \big\}.
\$
The uniform strong law of large numbers applied to $F(V(X_{n+j},Y_{n+j} ),U_j)$ implies 
\$
\limsup_{m,n\to \infty} \bigg|  \frac{1 }{m}  \sum_{j=1}^m  \ind \big\{F( V(X_{n+j},Y_{n+j} ), U_j) \leq \hat\tau \big\}
- Q\big( F( V(X , Y ), U ) \leq \hat\tau  \big) \bigg| = 0
\$
almost surely, 
where $Q$ is taken with respect to 
a new pair of $(X,Y)\sim Q$ and 
an independent $U\sim$ Unif([0,1]) while viewing $\hat\tau$ as fixed. 
Consequently, we have 
\$
\limsup_{m,n\to\infty} \fdr 
&\leq q \cdot \EE\Bigg[\limsup_{m,n\to \infty} \frac{ Q\big( F( V(X , Y ), U ) \leq \hat\tau  \big) }{ \hat \tau}   \Bigg].
\$
Note that by definition of $F(u,v)$, for any $t\in [0,1]$, we have 
\$
 Q\big( F( V(X , Y ), U ) \leq t  \big) = t. 
\$
This concludes that $\limsup_{m,n\to \infty} \fdr \leq q$.

Furthermore, since $F(V_j^0,U_j)$ admits a continuous distribution, 
the function $t\mapsto \PP\big( F(V_j^0, U_j) \leq t \big)$ is continuous in $t\in [0,1]$. 
Under the assumption that for any $\epsilon>0$, there exists 
some $|t-t^*|\leq \epsilon$ such that $\frac{\PP ( F(V_j^0, U_j) \leq t)}{t}>1/q$, 
the uniform convergence in~\eqref{eq:unif_conv_tau} implies $\hat\tau ~\asto~ t^*$. 
Let $\delta \in (0, t^*)$ be any fixed value 
such that $\PP\big( F(V_j^0, U_j) \leq \delta \big)>0$. 
Lemma~\ref{lem:tau_conv}, equation~\eqref{eq_asymp_1}, 
the uniform strong law of large numbers, together with the lower boundedness 
of $\PP\big( F(V_j^0, U_j) \leq t  \big) $ for $t\geq \delta$ thus ensures 
\$
\sup_{t\in [\delta,1]} \Bigg|\frac{\sum_{j=1}^m  \ind\{\hat{F}_n(V_j^0, U_j,W_j) \leq t, j\in \cH_0\}}{1\vee \sum_{j=1}^m \ind\{\hat{F}_n(V_j^0, U_j,W_j) \leq t\}} - \frac{ Q \{ F(V(X,C),U)\leq t, Y \leq C \}}{ Q\{F(V(X,C),U)\leq t \} }\Bigg| ~\asto~ 0. 
\$
Since $\hat\tau ~\asto ~  t^*$ and the distribution functions are continuous, the asymptotic FDR is  
\$
\limsup_{m,n\to\infty} \fdr &
= \limsup_{m,n\to \infty} \EE\Bigg[\frac{\sum_{j=1}^m  \ind\{\hat{F}_n(V_j^0, U_j,W_j) \leq \hat\tau, j\in \cH_0\}}{1\vee \sum_{j=1}^m \ind\{\hat{F}_n(V_j^0, U_j,W_j) \leq \hat\tau\}} \Bigg] \\ 
&=  \EE\Bigg[\limsup_{m,n\to \infty}\frac{\sum_{j=1}^m  \ind\{\hat{F}_n(V_j^0, U_j,W_j) \leq \hat\tau, j\in \cH_0\}}{1\vee \sum_{j=1}^m \ind\{\hat{F}_n(V_j^0, U_j,W_j) \leq \hat\tau\}} \Bigg] \\
&= \frac{  Q \{ F(V(X,C),U)\leq t^*, Y \leq C \}}{ Q\{F(V(X,C),U)\leq t^* \} },
\$
where the second line follows from the Dominated Convergence Theorem. 
With similar arguments, we can show that the asymptotic power of the procedure is 
\$
\limsup_{m,n\to\infty}  \text{Power} &= 
\limsup_{m,n\to \infty} \EE\Bigg[\frac{\sum_{j=1}^m  \ind\{\hat{F}_n(V_j^0, U_j,W_j) \leq \hat\tau, j\notin \cH_0\}}{1\vee \sum_{j=1}^m \ind\{Y_{n+j} >C_{n+j} \}} \Bigg] \\
&= \frac{  Q \{ F(V(X,C),U)\leq t^*, Y  > C \}}{ Q\{Y >C \} }.
\$
This concludes the proof for the case $m,n\to \infty$.  
Therefore, we complete the proof of Theorem~\ref{thm:fdr_asymp}. 
\end{proof}

\section{Proofs for finite-sample FDR control}

\subsection{Proof of Theorem~\ref{thm:calib_ite}}
\label{app:thm_calib_ite}

\begin{proof}[Proof of Theorem~\ref{thm:calib_ite}]
We first reduce the FDR for the three choices of $\cR$ 
using the following lemma. 
Its proof is deferred to Appendix~\ref{app:fdr_ite_reduce}. 

\begin{lemma}
\label{lem:fdr_ite_reduce}
Let $\cR = \cR_{\hete}$ or $\cR = \cR_{\homo}$ or $\cR = \cR_{\dtm}$ in Algorithm~\ref{alg:bh}. 
Then $\cR$
satisfies
\# \label{eq:cc_eq1}
{\textnormal \fdr}  =  \EE\Bigg[  \frac{\sum_{j=1}^m \ind{\{j \in \cR, Y_{n+j} \leq c_{n+j}\}} }{1\vee |\cR|} \Bigg]  
\leq \sum_{j=1}^m  \EE\Bigg[ \frac{  \ind {\{  p_j \leq s_j,Y_{n+j} \leq c_{n+j}\}} }{|\hat\cR_{j\to 0} |}   \Bigg],
\#
where $s_j = q|\hat\cR_{j\to 0}|/m$, and $\hat\cR_{j\to 0}$ is the 
rejection set we compute in Line 6 of Algorithm~\ref{alg:bh}. 
\end{lemma}

We now show that each term in the summation~\eqref{eq:cc_eq1} is upper bounded by $q/m$. 
For convenience, we will call BH($q$) as the BH procedure at nominal level $q$. 
Recall that 
$\hat\cR_{j\to 0}$ is the rejection set of BH($q$) applied to
$\{  p_1^{(j)},\dots,   p_{j-1}^{(j)}, 0 ,   p_{j+1}^{(j)}, \dots,   p_{m}^{(j)}\}$. 
For clarity, we define $\hat\cR_j$ as 
the rejection set of BH($q$) applied to 
$\{  p_1^{(j)},\dots,  p_{j-1}^{(j)},   p_j ,   p_{j+1}^{(j)}, \dots,  p_{m}^{(j)}\}$, 
which differs from $\hat\cR_{j\to 0}$ in that we 
use $p_j$ instead of $0$ in the $j$-th p-value. 
By the definition of BH procedure, we note that 
\#\label{eq:equiv_1}
j\in \hat\cR_j \quad  \textrm{if and only if} \quad  
 p_j \leq s_j = \frac{q |\hat\cR_{j\to 0} |}{m}.
\#  
We now define the oracle counterparts of these $p$-values as
\#\label{eq:oracle_p_j}
\bar{p}_j = \frac{\sum_{i=1}^n w(X_i)\ind {\{V_i < {V}_{n+j} \}}+   w(X_{n+j})  }{\sum_{i=1}^n w(X_i) + w(X_{n+j})},
\#
and 
\#\label{eq:oracle_p_ell^j}
\bar{p}_{\ell}^{(j)} := \frac{\sum_{i=1}^n w(X_i) \ind {\{V_i < \hat{V}_{n+\ell}\}} + w(X_{n+j}) \ind {\{ {V}_{n+j}<  \hat{V}_{n+\ell}\}} }{\sum_{i=1}^n w(X_i) + w(X_{n+j})}.
\# 
These p-values are not computable and are only for 
the purpose of analysis. 
The difference between $\{p_j\}$ and $\{\bar{p}_j\}$ 
(similarly $\{p_\ell^{(j)}\}$ versus $\{\bar{p}_\ell^{(j)}\}$) is 
only in whether we use 
the observed score $\hat{V}_{n+j} = V(X_{n+j},c_{n+j} )$ or 
the oracle score $V_{n+j} = V(X_{n+j},Y_{n+j})$; 
on the other hand, $\hat{V}_{n+\ell}$ for $\ell\neq j$ are used in both constructions. 

Now let 
$\bar\cR_j$ 
be the rejection set of  BH($q$) procedure applied to  
$\{\bar{p}_1^{(j)},\dots,\bar{p}_{j-1}^{(j)},\bar{p}_j,\bar{p}_{j+1}^{(j)},\dots,\bar{p}_m^{(j)}\}$. 
In the following, we are to show that if $j\in \hat\cR_j$ (i.e., $ p_j\leq s_j = q|\hat\cR_{j\to 0}|/m$) 
and $j\in \cH_0$ (recall that $\cH_0$ is the random set of null hypotheses), then 
$\hat\cR_j = \bar\cR_j$, i.e., 
changing $ p_j$ and $\{  p_\ell^{(j)}\}$ 
to their oracle counterparts 
$\bar{p}_j$ and $\{ \bar{p}_\ell^{(j)}\}$ does not change the rejection set
of BH($q$). 

To see this, the monotinicity of nonconformity scores 
implies that $\bar{p}_j \leq   p_j$ on the event $\{j\in \cH_0\} =\{Y_{n+j} \leq c_{n+j}\}$. 
We then consider two cases for $\ell\neq j$ 
on the event $\{j\in \cH_0, j\in \hat\cR_j\}$. 
\begin{enumerate}[(a)]
\item if $\hat V_{n+\ell} > \hat{V}_{n+j}$, then $\bar{p} _\ell^{(j)} =   p_\ell^{(j)}$; 
this is because in this case $\ind\{ \hat{V}_{n+j} < \hat{V}_{n+\ell} \} = \ind\{  {V}_{n+j} < \hat{V}_{n+\ell} \} = 1$.
\item if $\hat{V}_{n+\ell} \leq \hat{V}_{n+j}$, from the step-up nature 
of BH($q$), we know $ {p}_{\ell}^{(j)} \leq  {p}_j$, 
hence $\ell \in \hat\cR_j$ as well.  
Therefore, it always holds that 
\$
\bar{p}_{\ell}^{(j)} \leq \frac{\sum_{i=1}^n w(X_i) \ind {\{V_i < \hat{V}_{n+\ell}\}} + w(X_{n+j})  }{\sum_{i=1}^n w(X_i) + w(X_{n+j})} \leq  
\frac{\sum_{i=1}^n w(X_i) \ind {\{V_i < \hat{V}_{n+j}\}} + w(X_{n+j})  }{\sum_{i=1}^n w(X_i) + w(X_{n+j})}=
  p_j.
\$
Thus, after modifying  $\hat{V}_{n+j}$ to $V_{n+j}$, 
the oracle $\bar{p} _\ell^{(j)}$ still does not exceed the observable $ p_j$. 
\end{enumerate}

Summarizing the above two cases, 
for $\ell\neq j$ such that  $p_\ell^{(j)}>  p_j$, we have $\bar{p}_\ell^{(j)} = p_\ell^j$, 
while  
for $ p_\ell^{(j)} \leq  p_j$, we have $\bar{p}_\ell^{(j)} \leq  p_j$. 
As a result,  
the step-up nature of BH($q$) leads to the fact that 
\#\label{eq:subset}
j\in\cH_0~\text{and}~ j\in \hat\cR_j \quad \Rightarrow \quad j\in \bar\cR_j  ~\text{and}~ \hat\cR_j = \bar\cR_j.
\# 
Finally, we let $\bar \cR_{j\to 0} $ be the rejection set 
of BH($q$) applied to p-values 
$\{\bar{p}_1^{(j)},\dots,\bar{p}_{j-1}^{(j)},0, \bar{p}_{j+1}^j,\dots,\bar{p}_{m}^{(j)}\}$, so that  
$j\in \bar{\cR}_{j\to 0}$ determinisitically by default.  
Then we have 
\#\label{eq:subset_2}
\big\{j\in\cH_0,   p_j \leq s_j = q|\hat\cR_{j\to 0}|/m \big\}
&= \big\{j\in\cH_0, j\in \hat\cR_j \big\} \notag  \\ 
&\subseteq \big\{j\in\cH_0, j\in  \bar{\cR}_j \big\} \\ 
&= \big\{j\in\cH_0, j\in  \bar\cR_j, p_j \leq q| \bar\cR_{j\to 0} |/m \big\}, \notag 
\#
where 
the first line is due to~\eqref{eq:equiv_1}, 
the second line is due to~\eqref{eq:subset}, 
and the last equality follows from the 
oracle counterpart to~\eqref{eq:equiv_1}. 
The above inclusion then implies 
\$
\frac{\ind\{j\in\cH_0,  p_j \leq s_j \} }{  |\hat\cR_{j\to 0} | }
&\leq \frac{\ind\{j\in\cH_0, \bar{p}_j \leq q|\bar\cR_{j\to 0}|/m \} }{  | \bar\cR_{j\to 0} |} .
\$
Consequently, 
\#\label{eq:orc_fdr}
\EE\Bigg[ \frac{  \ind {\{  p_j \leq s_j, j\in \cH_0\}} }{|\hat\cR_{j\to 0} |}   \Bigg] 
\leq 
\EE\Bigg[ \frac{  \ind{\{   \bar{p}_j \leq \frac{q |\bar\cR_{j\to 0} |}{m} , j\in \cH_0\}} }{|\bar\cR_{j\to 0}|}   \Bigg] 
\leq \EE\Bigg[ \frac{  \ind{\{   \bar{p}_j \leq \frac{q |\bar\cR_{j\to 0} |}{m} \}} }{|\bar\cR_{j\to 0}|}   \Bigg]. 
\#
On the right-handed side, 
$\bar{\cR}_{j\to 0}$ only depends on the oracle p-values $\{\bar{p}_\ell^{(j)}\colon \ell\neq j\}$ 
but not $\bar{p}_j$.

We define the event 
$\cE_{z,j} = \big\{ [Z_1,\dots,Z_n,Z_{n+j}]=[z_1,\dots,z_n,z_{n+j}]\big\}$, 
where $Z:=[Z_1,\dots,Z_n,Z_{n+j}]$ is an unordered set of 
all calibration data 
$Z_i = (X_i,Y_i)$, $i=1,\dots,n$ 
and the unobserved test sample $n+j$, 
and $z:=[z_1,\dots,z_n,z_{n+j}]$ is the unordered set of their realized values. 

Conditional on $\cE_{z,j}$, 
the only randomness is only with regard to which 
sample takes on which value among $[z_1,\dots,z_n,z_{n+j}]$. 
Now conditional on $\cE_{z,j}$, 
let $I_j$ be the index in $\{1,\dots,n,n+j\}$ such that 
$Z_{n+j}=z_{I_j}$. Then $\bar{p}_j$ is measurable with respect to $I_j$. 
On the other hand, due to the independence among 
$\{Z_1,\dots,Z_n,Z_{n+j}\}\cup \{\tilde{Z}_{n+\ell}\}_{\ell\neq j}$ 
for $\tilde{Z}_{n+\ell} = (X_{n+\ell},c_{n+\ell})$ in Theorem~\ref{thm:calib_ite}, 
we know that $\{\tilde{Z}_{n+\ell}\}_{\ell\neq j}$ is independent of 
$I_j$ given $\cE_{z,j}$. 
Meanwhile, 
the values of $\{\bar p_\ell^{(j)}\colon \ell\neq j\}$ are fully decided 
by the unordered set $Z$ and $\{\tilde{Z}_{n+\ell}\}_{\ell\neq j}$, 
which means 
$\{\bar p_\ell^{(j)}\colon \ell\neq j\}$ is measurable with respect to 
$\cE_{z,j}$ and $\{\tilde{Z}_{n+\ell}\}_{\ell\neq j}$. 
As $\bar\cR_{j\to 0}$ is determined by $\{\bar p_\ell^{(j)}\colon \ell\neq j\}$, 
we know $|\bar\cR_{j\to 0}|$ is  measurable with respect to 
$\cE_{z,j}$ and $\{\tilde{Z}_{n+\ell}\}_{\ell\neq j}$. 
Putting these arguments together, 
we have the following conditional independence structure: 
\$ 
\bar p_j  \indep  \big| \bar\cR_{j\to 0}  \big| ~\Biggiven ~\cE_{z,j},\quad \forall j.  
\$
This implies 
\$
\EE \Bigg[ \frac{  \ind{\{   \bar p_j \leq \frac{q |\bar\cR_{j\to 0} | }{m} \}} }{ |\bar\cR_{j\to 0}| } \Biggiven \cE_{z,j}  \Bigg] =  \frac{ F_{z,j} (\frac{q |\bar\cR_{j\to 0} |}{m} ) }{|\bar\cR_{j\to 0} |}   ,
\$
where $F_{z,j}(t) = \PP(p_j \leq t \given \cE_{z,j})$ for $t\in \RR$. 
We consider 
the randomized weighted conformal p-value
\$
{p}_j^* = \frac{\sum_{i=1}^n w(X_i)\ind {\{V_i < {V}_{n+j} \}}+  U_j \cdot (  w(X_{n+j}) + \sum_{i=1}^n w(X_i) \ind\{V_i = V_{n+j}\}) }{\sum_{i=1}^n w(X_i) + w(X_{n+j})},
\$
where  $p_j^*\given \cE_z\sim \textrm{Unif}[0,1]$ by the weighted 
exchangeability~\cite{tibshirani2019conformal,hu2020distribution}. 
Since 
$p_j^*\leq \bar p_j$ determinisitically, we 
have $F_{z,j}(t) \leq \PP(p_j^* \leq t\given \cE_{z,j}) =t$ for 
all $t\in \RR$. By the tower property, 
\$
\EE\Bigg[ \frac{  \ind {\{  p_j \leq c_{n+j}, j\in \cH_0\}} }{|\hat\cR_{j\to 0} |}   \Bigg] 
\leq 
\EE\Bigg[\frac{ F_{z,j} (\frac{q |\bar\cR_{j\to 0} |}{m} ) }{|\bar\cR_{j\to 0} |} \Bigg]
\leq \EE\Bigg[\frac{q |\bar\cR_{j\to 0} |}{m|\bar\cR_{j\to 0} | }  \Bigg] = \frac{q}{m}.
\$
Summing over all $j\in\{1,\dots,m\}$ concludes the proof of Theorem~\ref{thm:calib_ite}.
\end{proof}

\subsection{Proof of Lemma~\ref{lem:fdr_ite_reduce}}
\label{app:fdr_ite_reduce}

\begin{proof}[Proof of Lemma~\ref{lem:fdr_ite_reduce}]  
We prove the lemma for the three cases separately. 
\paragraph{Case 1: $\cR = \cR_{\hete}$.} 
The arguments for this case resemble the idea of~\cite{fithian2020conditional}.
We include the proof for completeness since 
here we deal with random hypotheses. 
By definition, $r_{\hete}^* = |\cR|$.  
Then 
\$
\fdr  =  \EE\Bigg[  \frac{\sum_{j=1}^m \ind{\{j \in \cR, Y_{n+j} \leq c_{n+j}\}} }{1\vee |\cR|} \Bigg]  
&= \EE\Bigg[ \frac{\sum_{j\in \cH_0} \ind {\{  p_j \leq s_j,Y_{n+j} \leq c_{n+j}\}} \ind {\{\xi_j \leq \frac{|\cR|}{ |\hat\cR_{j\to 0}| }\}} }{1\vee |\cR|} \Bigg].
\$
Fix any $j\in \{1,\dots,m\}$. 
Once $j\in \cR$, 
in our pruning method, 
sending $\xi_j\to 0$ while keeping other $\xi$'s 
unchanged does not change the rejection set. Therefore, 
letting $\cR_{\xi_j\to 0}$ denote 
the rejection set obtained by replacing $\xi_j$ by $0$ 
in the pruning step, we have 
\$
\fdr &= \sum_{j=1}^m \sum_{k=1}^m \EE\bigg[ \frac{\ind {\{|\cR|=k\}}}{k} \ind {\{  p_j \leq s_j,Y_{n+j} \leq c_{n+j}\}} \ind {\{\xi_j \leq {\textstyle \frac{k}{ |\hat\cR_{j\to 0} | } }   \}}    \bigg] \notag \\ 
&= \sum_{j=1}^m  \sum_{k=1}^m \EE\bigg[ \frac{\ind {\{|\cR_{\xi_j\to 0}|=k\}}}{k}  \ind {\{  p_j \leq s_j,Y_{n+j} \leq c_{n+j}\}} \ind {\{\xi_j \leq {\textstyle \frac{k}{ |\hat\cR_{j\to 0} | }  } \}}   \bigg] \notag \\ 
&= \sum_{j=1}^m \sum_{k=1}^m \EE\bigg[ \frac{\ind {\{|\cR_{\xi_j\to 0}|=k\}}}{|\hat\cR_{j\to 0} |}  \ind {\{p_j \leq s_j\}}   \bigg] 
= \sum_{j=1}^m  \EE\bigg[ \frac{  \ind {\{  p_j \leq s_j,Y_{n+j} \leq c_{n+j}\}} }{|\hat\cR_{j\to 0} |}   \bigg], 
\$
where the third line follows from 
the fact that $\xi_j$ are independent of $\cR_{\xi_j\to 0}$, $  p_j$ and $\hat\cR_{j\to 0}$. 
This concludes the proof of the heterogeneous case. 

\paragraph{Case 2: $\cR = \cR_{\homo}$.} This case 
uses a conditional PRDS property of $\{\xi |\hat\cR_{j\to 0}|\}_{j=1}^m$ 
which we will specify later.   
We let $\EE_\cD$ be the conditional expectation given all the data, 
so that only the randomness in $\xi$ remains. 
Also, one still has $\xi\sim \textrm{Unif}([0,1])$ after conditioning. 
The tower property implies $\fdr = \EE[\fdr(\cD)]$, where we define the conditional FDR 
\$
\fdr(\cD) := \EE_{\cD}\Bigg[  \frac{\sum_{j=1}^m \ind{\{j \in \cR, Y_{n+j} \leq c_{n+j}\}} }{1\vee |\cR|} \Bigg],
\$
and $\EE_{\cD}$ is with respect to the randomness in $\xi$ 
conditional on $\cD=\{(X_i,Y_i)\}_{i=1}^n\cup\{(X_{n+j},c_{n+j})\}_{j=1}^m$. 
Note that $\ind\{  p_j \leq s_j\}$ are deterministic conditional on the data. 
We let $\cR^{(1)} = \{j\colon   p_j \leq s_j\}$ be the first-step rejection set, 
and $\cH_0 = \{j\colon Y_{n+j} \leq c_{n+j}\}$ be the set of null hypotheses, 
both being deterministic given the data. 
Then we have the decomposition 
\$
\fdr(\cD) &= \sum_{k=1}^m \EE_{\cD}\bigg[ \ind\{ |\cR| = k\} \frac{\sum_{j\in \cR_0 \cap \cH_0} \ind{\{ \xi |\hat\cR_{j\to 0}| \leq k \}} }{k} \bigg] \\
&= \sum_{j\in \cR^{(1)} \cap \cH_0} \sum_{k=1}^m \EE_{\cD}\bigg[ \frac{\ind{\{ \xi |\hat\cR_{j\to 0}| \leq k \}}}{k} \left(  \ind\{ |\cR| \leq k\} -  \ind\{ |\cR| \leq k-1\} \right) \bigg].
\$
Fix any $j\in \cR^{(1)}\cap \cH_0$, and let $m_1 = |\cR^{(1)}|$. Then 
we know 
$|\cR|\leq m_1$ deterministically, and thus 
\$
\ind{\{ \xi |\hat\cR_{j\to 0}| \leq k \}} \left(  \ind\{ |\cR| \leq k\} -  \ind\{ |\cR| \leq k-1\} \right)
 = 0
\$
for any $k \geq m_1 + 1$. We then have 
\#\label{eq:homo_bd}
\fdr(\cD,j)&:=\sum_{k=1}^{m} \frac{1}{k}\EE_{\cD}\big[  \ind{\{ \xi |\hat\cR_{j\to 0}| \leq k \}} \left(  \ind\{ |\cR| \leq k\} -  \ind\{ |\cR| \leq k-1\} \right) \big] \notag \\ 
&= \sum_{k=1}^{m_1} \frac{1}{k}\EE_{\cD}\big[  \ind{\{ \xi |\hat\cR_{j\to 0}| \leq k \}} \left(  \ind\{ |\cR| \leq k\} -  \ind\{ |\cR| \leq k-1\} \right) \big] \notag \\
&\leq \frac{1}{m_1} + \sum_{k=1}^{m_1 -1}\frac{1}{k}\EE_{\cD}\big[  \ind{\{ \xi |\hat\cR_{j\to 0}| \leq k \}} \ind\{ |\cR| \leq k\} \big]  \notag  \\ 
&\qquad \qquad - \sum_{k=1}^{m_1-1} \frac{1}{k+1}\EE_{\cD}\big[  \ind{\{ \xi |\hat\cR_{j\to 0}| \leq k+1 \}} \ind\{ |\cR| \leq k\} \big] \notag \\ 
&\leq \frac{1}{m_1} + \sum_{k=1}^{m_1-1} \bigg\{\frac{\PP_\cD(\xi |\hat\cR_{j\to 0}| \leq k)}{k}
\PP_{\cD}\Big( |\cR| \leq k \Biggiven \xi |\hat\cR_{j\to 0}| \leq k \Big) \notag  \\
&\qquad \qquad \qquad \qquad - 
 \frac{\PP_\cD(\xi |\hat\cR_{j\to 0}| \leq k+1)}{k+1}
\PP_{\cD}\Big( |\cR| \leq k \Biggiven \xi |\hat\cR_{j\to 0}| \leq k+1 \Big)\bigg\}.
\#
Here $\PP_{\cD}$ is with respect to the conditional distribution of $\xi$ 
given $\cD$. 
We would rely on the PRDS property of $\xi |\hat\cR_{j\to 0}|$ 
as follows to control $\fdr(\cD,j)$. 
\begin{lemma}\label{lem:prds}
Let $a_1,\dots,a_k \in \RR$ be nonnegative fixed constants, and let $\xi \sim \textnormal{Unif}[0,1]$. 
Then the random variables 
$\{a_1\xi,\dots,a_{j-1}\xi,a_{j+1}\xi,\dots,a_k\xi\}$ are PRDS in $a_j \xi$ for any $j\in\{1,\dots,k\}$. 
\end{lemma}

\begin{proof}[Proof of Lemma~\ref{lem:prds}] 
We denote $x_j = a_j \xi$, and $a_j\geq 0$ for all $j=1,\dots,k$. 
Fix any $j\in\{1,\dots,k\}$. 
If $a_j=0$, then $x_j\equiv 0$ hence 
the  distribution of $X=(x_1,\dots,x_k)$ conditional on $x_j$ 
is invatiant to the value of $\xi$, and the PRDS property naturally follows.
If $a_j>0$, then  conditional on $ x_j $, other $p$-values are deterministic 
with $x_\ell = x_j \cdot a_\ell/ a_j $.  
Now let $D$ be any increasing set, then for any $x\in [0,a_j]$, we have 
\$
\PP(X\in D\given x_j =x) = \ind \big\{  (xa_1/a_j,\dots, xa_k/a_j)\in D  \big\},
\$
which is increasing in the whole range $x\in [0,a_j]$. We thus complete the proof of Lemma~\ref{lem:prds}. 
\end{proof}

By Lemma~\ref{lem:prds} and the construction of $\cR=\cR_{\homo}$, we see that 
\$
\PP_{\cD}\Big( |\cR| \leq k \Biggiven \xi |\hat\cR_{j\to 0}| \leq k \Big) 
\leq \PP_{\cD}\Big( |\cR| \leq k \Biggiven \xi |\hat\cR_{j\to 0}| \leq k+1 \Big) ,
\$
hence since $\xi$ is independent of everything else, 
\#\label{eq:prds_bd}
\fdr(\cD,j) &\leq \frac{1}{m_1} + \sum_{k=1}^{m_1-1} \bigg\{\frac{\PP_\cD(\xi |\hat\cR_{j\to 0}| \leq k)}{k} -  \frac{\PP_\cD(\xi |\hat\cR_{j\to 0}| \leq k+1)}{k+1}\bigg\}
\PP_{\cD}\Big( |\cR| \leq k \Biggiven \xi |\hat\cR_{j\to 0}| \leq k \Big) \notag \\ 
&= \frac{1}{m_1} + \sum_{k=1}^{m_1-1} \bigg\{\frac{ \min\{1, k/|\hat\cR_{j\to 0}|\} }{k} -  \frac{\min\{1, (k+1)/|\hat\cR_{j\to 0}|\} }{k+1}\bigg\}
\PP_{\cD}\Big( |\cR| \leq k \Biggiven \xi |\hat\cR_{j\to 0}| \leq k \Big) . 
\#
If $|\hat\cR_{j\to 0}|\geq m_1$, then the above inequality leads to $\fdr(\cD,j)\leq \frac{1}{m_1} \leq \frac{1}{|\hat\cR_{j\to 0}|}$. Otherwise 
\$
\fdr(\cD,j) &\leq \frac{1}{m_1} + \sum_{k=|\hat\cR_{j\to 0}|}^{m_1-1} 
\bigg\{\frac{ 1 }{k} -  \frac{1}{k+1}\bigg\}  
\PP_{\cD}\Big( |\cR| \leq k \Biggiven \xi |\hat\cR_{j\to 0}| \leq k \Big) \\ 
&\leq \frac{1}{m_1} + \sum_{k=|\hat\cR_{j\to 0}|}^{m_1-1} 
\bigg\{\frac{ 1 }{k} -  \frac{1}{k+1}\bigg\} \leq \frac{1}{|\hat\cR_{j\to 0}|}.
\$
Putting the above bounds together, we have 
\$
\fdr(\cD) \leq \sum_{j\in \cR^{(1)}\cap \cH_0} \frac{1}{|\hat\cR_{j\to 0}|}
= \sum_{j=1}^m \frac{  \ind {\{  p_j \leq s_j,Y_{n+j} \leq c_{n+j}\}} }{|\hat\cR_{j\to 0} |}.
\$
Marginalizing over the randomness in $\cD$ concludes the proof of the homogeneous case. 

\paragraph{Case 3: $\cR=\cR_\dtm$.} 
Finally we show the deterministic case. 
Let $m_1 = |\cR_0|$ for $\cR_0=\{j\colon p_j\leq c_{n+j}\}$ and let 
$[1],[2],\dots,[m_1]$ be a reorder of $\cR_0$ such that 
$|\hat\cR_{[1]\to 0}|\leq |\hat\cR_{[2]\to 0}| \leq \cdots \leq |\hat\cR_{[m_1]\to 0}|$. 
Then by definition, we know $m_2 := |\cR|$ is the largest index $k$ 
such that $|\hat\cR_{[k]\to 0}|\leq k$ and $\cR=\{[1],[2],\dots,[m_2]\}$. 
Put another way, we have $|\hat\cR_{j\to 0}|\leq |\cR|$ for all $j\in \cR$. 
Thus, 
\$
\EE\Bigg[  \frac{\sum_{j=1}^m \ind{\{j \in \cR, Y_{n+j} \leq c_{n+j}\}} }{1\vee |\cR|} \Bigg]  
&\leq \sum_{j=1}^m  \EE\Bigg[  \frac{ \ind {\{ j\in \cR,Y_{n+j} \leq c_{n+j}\}}   }{1\vee |\hat\cR_{j\to 0}|} \Bigg]
\leq \sum_{j=1}^m  \EE\Bigg[  \frac{ \ind {\{ p_j\leq s_j,Y_{n+j} \leq c_{n+j}\}}   }{1\vee |\hat\cR_{j\to 0}|} \Bigg].
\$
This concludes the proof of the third case, and therefore 
also concludes the proof of Lemma~\ref{lem:fdr_ite_reduce}. 
\end{proof}

\subsection{Proof of Proposition~\ref{prop:asymp_equiv}}
\label{app:subsec_asymp_equiv}

\begin{proof}[Proof of Proposition~\ref{prop:asymp_equiv}]
We start by describing the BH procedure 
and our procedure using 
a similar characterization as in~\cite{storey2004strong}. 
First, our first-step rejection set 
computed in Line 8 of Algorithm~\ref{alg:bh} can be written as 
$\hat\cR_{j\to 0} = \{j\}\cup\{\ell\colon   p_\ell^{(j)} \leq \hat\tau_j\}$, where 
\#\label{eq:def_hattau_j}
\hat\tau_j = \sup \bigg\{t \in[0,1]\colon \frac{mt}{1+\sum_{\ell\neq j} \ind\{ p_\ell^{(j)} \leq t \}} \leq q \bigg\}.
\# 
The rejection set of BH($q$) is $\cR_{\bh} = \{j\colon p_j \leq \hat\tau\}$, where $\{p_j\}$ are 
defined in~\eqref{eq:def_wcpval_rand}, and 
\#\label{eq:def_hattau}
\hat\tau  = \sup \bigg\{t \in[0,1]\colon \frac{mt}{\sum_{j=1}^m \ind\{ p_\ell \leq t \}} \leq q \bigg\}.
\#
Also, under the conditions of Proposition~\ref{prop:asymp_equiv}, 
there are no ties within 
$\{V(X_i,Y_i)\}_{i=1}^n\cup \{V(X_{n+j},Y_{j} )\}_{j=1}^m$ 
or $\{V(X_i,Y_i)\}_{i=1}^n\cup \{V(X_{n+j},c_{n+j})\}_{j=1}^m$ 
almost surely.  
Thus, the p-values~\eqref{eq:def_wcpval_rand} 
can be written as
\#\label{eq:wcp_notie}
p_j = \frac{\sum_{i=1}^n w(X_i)\ind {\{V_i <\hat{V}_{n+j} \}}+  w(X_{n+j})  \cdot U_j}{\sum_{i=1}^n w(X_i) + w(X_{n+j})}.
\#
To proceed, with a slight abuse of notation, 
we define the cumulative distribution function of $V(X,Y)$ under the 
target distribution: 
\$
F(v) = Q(V(X,Y )\leq v) = Q(V(X,Y )<v),\quad v\in \RR,
\$
with the probability over $(X,Y)\sim Q$. 
Because the distribution of $V(X,Y )$ under $Q$ has no point mass, 
$F(\cdot)$ is continuous and increasing in $v$. 
For convenience, in the following we denote $p_j^{(j)} :=   p_j$, and define 
\$
V_j^0 := V(X_{n+j},c_{n+j})=\hat{V}_{n+j},\quad p_j^\infty := F(V_j^0). 
\$ 
The following lemma can be derived similarly as Lemma~\ref{lem:pval_conv}, 
whose proof is omitted for brevity. 
\begin{lemma}
\label{lem:pval_conv2}
For any fixed constant $M>0$, as $n\to \infty$ and possibly $m\to \infty$, we have 
\$ 
\limsup_{n\to \infty} \sup_{j\geq 1} \big| p_j - p_j^\infty \big| ~\asto~0, \quad 
\limsup_{n\to \infty} \sup_{j\geq 1,\ell\geq 1} \big|  p_\ell^{(j)} - p_\ell^\infty \big| ~\asto~0. 
\$
\end{lemma}

In particular, the almost sure convergence in Lemma~\ref{lem:pval_conv2} is uniform 
over all observed p-values, which can be countably many when we consider $m\to \infty$. 
We then prove the two cases separately. 

\paragraph{Case (i): $m$ fixed, $n\to \infty$.} 
We let $\cR^\infty$ be the rejection set 
of BH($q$) with $\{p_j^\infty\}$. Then by the uniform convergence of 
$\{p_j\}$ to $\{p_j^\infty\}$ 
and the no-point-mass condition on 
the distribution of $V_j^0$ (and hence $p_j^\infty$), 
the same arguments 
as in the proof of case (i) in Theorem~\ref{thm:fdr_asymp} yield 
\$
\lim_{n\to \infty} \PP(\cR_{\bh} \neq \cR^\infty)  = 1. 
\$
We let $\hat\cR_j$ be the rejection set of BH($q$) from 
$\{p_\ell^{(j)}\}_{1\leq \ell \leq m}$, with the convention 
that $p_j^{(j)}=p_j$. 
We also denote $\hat\cR_{j\to 0}$ as 
the rejection set of BH($q$) applied to $\{p_1^{(j)},\dots,p_{j-1}^{(j)},0,p_{j+1}^{(j)},\dots,p_m^{(j)}$. 
Then by the property of the BH procedure, we know  
\#\label{eq:equiv_rj_cj}
 p_j  \leq s_j = q|\hat\cR_{j\to 0}|/m \quad \Leftrightarrow \quad j\in \hat\cR_j. 
\#
For each $j\in\{1,\dots,m\}$, due to the uniform convergence 
of $\{p_\ell^{(j)}\}$ to $\{p_j^\infty\}$, we similarly have 
$
\lim_{n\to \infty} \PP(\hat\cR_{j} \neq \cR^\infty)  = 1. 
$
Then taking a union bound over $j\in\{1,\dots,m\}$ yields 
\#\label{eq:equiv_Rj_all}
\lim_{n\to \infty}\PP(\hat\cR_j = \cR^\infty~\text{for all }j)=1.
\#
Recall that $ \cR^{(1)} = \{j\colon   p_j  \leq s_j\}$ is 
first-step selection set. 
Thus 
\$
\PP\big(\cR^{(1)} = \cR^\infty\big) = \PP\big(\{j\colon p_j  \leq s_j\} = \{j \in \hat \cR_j\} = \cR^\infty \big)
\to 1
\$
as $n\to \infty$, 
where the first equality uses~\eqref{eq:equiv_rj_cj}, 
and the last convergence uses~\eqref{eq:equiv_Rj_all}. We then have 
\$
\lim_{n\to \infty} \PP(E_n) = 1,\quad \textrm{where} \quad 
E_n := \big\{  \cR^{(1)} = \cR^\infty = \hat\cR_j ~\text{for all }j \big\} .
\$
On the event $E_n$, we know that 
$|\hat\cR_{j\to 0}| = |\hat\cR_j| = |\cR^\infty| = |\cR^{(1)}|$ for any $j\in \cR^\infty$. 
Therefore, for heterogeneous pruning, we have 
\$
\sum_{j=1}^m \ind\big\{  p_j \leq s_j, \xi_j |\hat\cR_{j\to 0}| \leq |\cR^\infty|\big\} 
= \sum_{j=1}^m \ind\big\{ j \in \cR^\infty, \xi_j |\cR^\infty| \leq |\cR^\infty| \big\} = |\cR^\infty|
\$
since $\xi_j\in [0,1]$; second,  
if we adopt the homogeneous pruning, 
\$
\sum_{j=1}^m \ind\big\{   p_j \leq s_j, \xi |\hat\cR_{j\to 0}| \leq |\cR^\infty|\big\} 
= \sum_{j=1}^m \ind\big\{ j \in \cR^\infty, \xi |\cR^\infty| \leq |\cR^\infty| \big\} = |\cR^\infty|
\$
since $\xi\in [0,1]$; finally, 
if we adopt the determinisitic pruning, 
\$
\sum_{j=1}^m \ind\big\{   p_j \leq s_j,  |\hat\cR_{j\to 0}| \leq |\cR^\infty|\big\} 
= \sum_{j=1}^m \ind\big\{ j \in \cR^\infty,  |\cR^\infty| \leq |\cR^\infty| \big\} = |\cR^\infty|.
\$
At the same time, for any $r\in \NN$, the heterogeneous pruning satisfies 
\$
\sum_{j=1}^m \ind\big\{   p_j \leq s_j, \xi_j |\hat\cR_{j\to 0}| \leq r\big\} \leq 
\sum_{j=1}^m \ind \{   p_j \leq s_j\} = |\cR^{(1)}| =|\cR^{\infty}| ,
\$
and similarly the homogeneous pruning satisfies 
$
\sum_{j=1}^m \ind \{   p_j \leq s_j, \xi |\hat\cR_{j\to 0}| \leq r \} \leq 
|\cR^{(1)}| =|\cR^{\infty}| ,
$
and the deterministic pruning satisfies 
$
\sum_{j=1}^m \ind \{   p_j \leq s_j,  |\hat\cR_{j\to 0}| \leq r \} \leq 
|\cR^{(1)}| =|\cR^{\infty}| .
$
Thus, by definition we have $r_{\hete}^* = r_{\homo}^* = r_{\dtm}^*= |\cR^\infty|$ and 
\$
\cR_\hete = \big\{j\colon   p_j \leq s_j,~ \xi_j |\hat\cR_{j\to 0}|\leq r_\hete^* \big\} = \cR^{(1)} = \cR^\infty  = \cR_\homo = \cR_{\dtm}
\$
on the event $E_n$, hence completing the proof of case (i). 

\paragraph{Case (ii): $m,n\to \infty$.}
Our first result is the convergence of the rejection threshold $\hat\tau_j$ 
and $\tau$ (c.f.~\eqref{eq:def_hattau_j} and~\eqref{eq:def_hattau}) 
depending on the following lemma.

\begin{lemma}
\label{lem:tau_conv2}
Under the conditions of Proposition~\ref{prop:asymp_equiv}, 
\$
\sup_{t\in [0,1]}\sup_{j\geq 1} \bigg|\frac{1}{m}\sum_{\ell=1}^n \ind\{  p_\ell^{(j)}  \leq t  \}
- \PP(F(V_j^0) \leq t) \bigg| ~\asto~0, \quad 
\sup_{t\in [0,1]}  \bigg|\frac{1}{m}\sum_{\ell=1}^n \ind\{  p_\ell \leq t  \}
- \PP(F(V_j^0) \leq t) \bigg| ~\asto~0,
\$
as $m,n\to \infty$, 
where the supremum is over countably many $j$, 
and $\PP$ is over the distribution of test samples. 
\end{lemma}
The proof of Lemma~\ref{lem:tau_conv2} follows exactly the same argument as 
Lemma~\ref{lem:tau_conv} in the proof of Theorem~\ref{thm:fdr_asymp}, 
relying on the 
uniform convergence of $\{ p_\ell^{(j)}\}_{\ell\neq j}$ and $\{p_j\}$ to $\{p_j^\infty\}$ 
as well as the continuous distribution of $p_j^\infty$. 
We omit the proof here for brevity. 

Given that 
there exists some $t'\in(0,1]$ such that 
$\frac{\PP(F(V_j^0) \leq t)}{t} > 1/q$, 
taking any fixed $\delta\in (0,1)$, 
we have 
\$
\sup_{t\in [\delta,1]} \bigg|\frac{\sum_{\ell=1}^n \ind\{  p_\ell \leq t  \}}{mt}
- \frac{\PP(F(V_j^0)\leq t) }{t} \bigg| ~\asto ~ 0
\$
and 
\$
\sup_{t\in [\delta,1]} \sup_{j\geq 1} \bigg|\frac{ 1 + \sum_{\ell\neq j} \ind\{  p_\ell^{(j)} \leq t  \}}{mt}
- \frac{\PP(F(V_j^0)\leq t) }{t} \bigg| ~\asto ~ 0.
\$
Since $t^* = \sup\{t\colon \frac{t}{\PP(F(V_j^0)\leq t)} \leq q\}$ 
is not an isolated point of this set, 
we know that both $\hat\tau_j$ and $\hat\tau$ are eventually 
within a neighborhood around $t^*$; to be specific, 
the above uniform convergence result guarantees 
the uniform convergence of the criteria: 
\#\label{eq:tau_unif}
\hat\tau ~\asto ~ t^*,\quad \sup_{j\geq 1} \big|\hat\tau_j - t^* \big| ~\asto ~0. 
\#
Therefore, we have 
\#\label{eq:conv_Rj}
&\limsup_{m,n\to \infty}\sup_{j\geq 1} \bigg| \frac{|\hat\cR_{j\to 0}|}{m} - \PP\big(F(V_j^0)\leq t^*\big) \bigg| \notag \\ 
&=\limsup_{m,n\to \infty}\sup_{j\geq 1}  \bigg| \frac{1}{m} + \frac{1}{m}\sum_{\ell\neq j} \ind\{p_\ell^{(j)} \leq \hat\tau_j \}
-  \PP\big(F(V_j^0)\leq t^*\big) \bigg| \notag \\ 
&\leq \limsup_{m,n\to \infty}\sup_{j\geq 1}  \bigg|  \frac{1}{m}\sum_{\ell=1}^m \ind\{ p_\ell^\infty \leq \hat\tau_j \}
-  \PP\big(F(V_j^0)\leq t^*\big) \bigg| \notag \\ 
&= \limsup_{m,n\to \infty}\sup_{j\geq 1}  \bigg| \PP\big(F(V_j^0)\leq \hat\tau_j \big) 
-  \PP\big(F(V_j^0)\leq t^*\big) \bigg| = 0,
\#
where the second line uses the uniform convergence in Lemma~\ref{lem:tau_conv2}, 
the third line uses the uniform law of large numbers for distribution functions 
of i.i.d.~random variables $p_\ell^\infty = F(V_\ell^0)$, 
and the last equality follows from the uniform convergence of $\hat\tau_j$ to $t^*$ in~\eqref{eq:tau_unif}
as well as the continuity of the distribution function of $F(V_j^0)$. 
For the convenience of reference, with a 
slight abuse of notation, we denote $G_\infty(t) = \PP(F(V_j^0)\leq t)$, 
so that~\eqref{eq:conv_Rj} 
implies $|\hat\cR_{j\to 0}|/m$ converges to $G_\infty(t^*)$ uniformly, and as $m,n\to \infty$,
\$
\sup_{j\geq 1} \big|c_{n+j} - q\cdot G_\infty(t^*)\big| =  q\cdot \sup_{j\geq 1 } \big||\hat\cR_{j\to 0}|/m - G_\infty(t^*)\big| ~ \asto ~ 0.
\$ 
From the definition of $t^*$ and the continuity of $G_\infty$, we have 
$ G_\infty(t^*)/t^* = 1/q$. 

In the following, we show that the first-step rejection set $\cR^{(1)} := \{j\colon  p_j \leq s_j\} $  
is asymptotically the same as $\cR_{\bh}$, i.e., $\limsup_{m,n\to \infty }\frac{|\cR^{(1)}\bigtriangleup \cR_\bh|}{m} = 0$
almost surely. 
We define $ {d}_j^{\cc} =   p_j - s_j$, $d_j^{\bh} = p_j - \hat\tau$ 
and $d_j^\infty = p_j^\infty - t^*$ for $j\in\{1,\dots,m\}$, 
for which we've shown that 
$\sup_{j\geq 1} \big| {d}_j^\cc - d_j^\infty\big| \to 0$ and 
$\sup_{j\geq 1} \big|d_j^\bh - d_j^\infty\big| \to 0$ almost surely. 
We also define 
\$\delta_n =\sup_{j\geq 1} \big\{ | {d}_j^\cc - d_j^\infty| \big\} \vee 
\sup_{j\geq 1} \big\{ |d_j^\bh - d_j^\infty| \big\},
\$
which obeys $\delta_{n,m}\to 0$ almost surely as $m,n\to \infty$. 
Therefore, for any fixed $\delta>0$, 
\#\label{eq:lln_step1}
\frac{|\cR^{(1)}\bigtriangleup \cR_\bh|}{m}
&\leq \frac{1}{m} \sum_{j=1}^m \big| \ind\{   p_j \leq s_j\} - \ind\{  p_j \leq \hat\tau \} \big| \\
&\leq \ind \{ \delta_n>\delta \} + \frac{\ind \{ \delta_n \leq \delta \} }{m} 
\sum_{j=1}^m \big| \ind\{  d_j^\cc \leq 0, d_j^\bh >0\} + \ind\{  d_j^\cc >0,  d_j^\bh \leq 0 \} \big|\notag \\ 
&\leq \ind \{ \delta_n>\delta \} + \frac{\ind \{ \delta_n \leq \delta \} }{m} 
\sum_{j=1}^m \big| \ind\{ -\delta < d_j^\infty \leq \delta \} + \ind\{  -\delta < d_j^\infty \leq \delta \} \big| \notag \\ 
&\leq \ind \{ \delta_n>\delta \}  + \frac{2}{m} \sum_{j=1}^m  \ind\{ -\delta < d_j^\infty \leq \delta \}, \notag
\#
where the third line uses the fact that  
$\delta_n\leq \delta$, $d_j^\cc\leq 0$ and $d_j^\bh>0$ 
implies $d_j^\infty\in(-\delta,\delta]$, and similarly for the other term. 
Above, we know $\ind\{\delta_n >\delta\}~\asto~ 0$ since 
$\delta_n~\asto~0$ as $m,n\to \infty$. 
By the strong law of large numbers, we have 
\$
\frac{2}{m} \sum_{j=1}^m  \ind\{ -\delta < d_j^\infty \leq \delta \} ~\asto ~ \PP\big(F(V_j^0)\in (-\delta + t^*, \delta + t^*]\big).
\$ 
Since the distribution of $F(V_j^0)$ is continuous, 
the arbitrariness of $\delta>0$ yields 
$
\frac{|\cR^{(1)}\bigtriangleup \cR_\bh|}{m} ~\asto~ 0.
$

We then proceed to show that the proportion of 
the hypotheses filterd out in the second step is negligible, i.e., 
$\frac{|\cR^{(1)}\bigtriangleup \cR_\homo|}{m}~\pto ~0$, 
and  $\frac{|\cR^{(1)}\bigtriangleup \cR_\homo|}{|\cR^{(1)}|}~\pto ~0$. 
For any fixed $\delta >0$, 
\$
&\limsup_{m,n\to \infty}  \bigg| \frac{|\cR^{(1)}|}{m} - \PP\big(F(V_j^0)\leq t^*\big) \bigg| \notag \\ 
&\leq \limsup_{m,n\to \infty}  \bigg|  \frac{1}{m}\sum_{j=1}^m  \ind\{  p_j \leq s_j \}
-  \frac{1}{m}\sum_{j=1}^m  \ind\{p_j^\infty \leq t^* \} \bigg| + 
\limsup_{m,n\to \infty}  \bigg|  \frac{1}{m}\sum_{j=1}^m  \ind\{p_j^\infty \leq t^* \}
-  \PP\big(F(V_j^0)\leq t^*\big) \bigg|\notag \\ 
&=\limsup_{m,n\to \infty}  \bigg|  \frac{1}{m}\sum_{j=1}^m  \ind\{  p_j \leq s_j \}
-  \frac{1}{m}\sum_{j=1}^m  \ind\{p_j^\infty \leq t^* \} \bigg| \notag \\
&\leq \limsup_{m,n\to \infty} \ind \{ \delta_n>\delta \} + \limsup_{m,n\to \infty}\frac{\ind \{ \delta_n \leq \delta \} }{m} 
\sum_{j=1}^m \big|  \ind\{  p_j \leq s_j \} - \ind\{p_j^\infty \leq t^* \} \big| \notag \\ 
&\leq \limsup_{m,n\to \infty} \ind \{ \delta_n>\delta \} + \limsup_{m,n\to \infty}\frac{\ind \{ \delta_n \leq \delta \} }{m} 
\sum_{j=1}^m \big| \ind\{ -\delta < d_j^\infty \leq \delta \} + \ind\{  -\delta < d_j^\infty \leq \delta \} \big| \notag \\ 
&\leq  \limsup_{m,n\to \infty} \frac{1}{m} 
\sum_{j=1}^m \big| \ind\{ -\delta < d_j^\infty \leq \delta \} + \ind\{  -\delta < d_j^\infty \leq \delta \} \big|  \\
&\leq 2\PP\big( t^*-\delta \leq F(V_j^0)\leq t^*+\delta\big),
\$
where the first equality uses the strong law of large numbers; 
the arguments starting from the third line follow similar ideas 
as those in~\eqref{eq:lln_step1}; 
the second-last inequality uses 
$\ind\{\delta_n >\delta\}~\asto ~ 0$, and 
the last inequality again uses the strong law of large numbers. 
Therefore, by the arbitrariness of $\delta>0$, we have 
\#  \label{eq:conv_R(1)}
&\limsup_{m,n\to \infty}  \bigg| \frac{|\cR^{(1)}|}{m} - \PP\big(F(V_j^0)\leq t^*\big) \bigg| = 0
\#
almost surely. 
Combined with the previous result that 
$\frac{|\cR^{(1)}\bigtriangleup \cR_\bh|}{m} ~\asto~ 0$, 
we also have 
$\frac{|\cR^{(1)}\bigtriangleup \cR_\bh|}{|\cR^{(1)}|} ~\asto~ 0$ and $\frac{|\cR^{(1)}\bigtriangleup \cR_\bh|}{|\cR_\bh|} ~\asto~ 0$, hence proving the first two parts 
in case (ii) of Proposition~\ref{prop:asymp_equiv}.  

We now fix any $\delta\in(0,1)$, and denote the event 
\$
E(\delta) = \Bigg\{ \sup_j\bigg|\frac{|\hat\cR_{j\to 0}|}{m} - G_\infty(t^*) \bigg| >\delta  \Bigg\} 
\bigcup \Bigg\{ \frac{1}{m} \sum_{j=1}^m \ind\big\{   p_j \leq s_j \big\} < (1-\delta/2)G_\infty(t^*)  \Bigg\}.
\$
Then the uniform concentration~\eqref{eq:conv_Rj} and~\eqref{eq:conv_R(1)} 
implies 
$\PP(E(\delta))\to 0$ 
as $m,n\to \infty$. 
Let $r = \lfloor (1-\delta)\cdot m\cdot G_\infty(t^*)\rfloor
\geq (1-2\delta) \cdot m \cdot G_\infty(t^*)$ when $m$ is 
sufficiently large. Then 
\$
&\PP\Bigg(  \sum_{j=1}^m \ind\big\{   p_j \leq s_j, \xi |\hat\cR_{j\to 0}| \leq r \big\} \geq r \Bigg) \\ 
&\geq \PP\big(E(\delta)^c \big)\cdot 
\PP\Bigg(  \sum_{j=1}^m \ind\big\{  p_j \leq s_j, \xi |\hat\cR_{j\to 0}| \leq r \big\} \geq r \Biggiven E(\delta)^c \Bigg) \\ 
&\geq \PP\big(E(\delta)^c \big)\cdot 
\PP\Bigg( \frac{1}{m} \sum_{j=1}^m \ind\big\{  p_j \leq s_j, \xi \leq {\textstyle \frac{r}{|\hat\cR_{j\to 0}|} } \big\} \geq (1-\delta/2) G_\infty(t^*) \Biggiven E(\delta)^c \Bigg) \\ 
&\geq \PP\big(E(\delta)^c \big)\cdot 
\PP\Bigg( \frac{1}{m} \sum_{j=1}^m \ind\big\{  p_j \leq s_j, \xi \leq {\textstyle \frac{(1-2\delta) m G_\infty(t^*)}{(1+\delta) m G_\infty(t^*)} } \big\} \geq (1-\delta/2) G_\infty(t^*) \Biggiven E(\delta)^c \Bigg) \\ 
&= \PP\big(E(\delta)^c \big)\cdot 
\PP\Bigg( \xi \leq {\textstyle \frac{(1-2\delta) m G_\infty(t^*)}{(1+\delta) m G_\infty(t^*)} },~\text{and}~\frac{1}{m} \sum_{j=1}^m \ind\big\{ p_j \leq s_j \big\} \geq (1-\delta/2) G_\infty(t^*) \Biggiven E(\delta)^c \Bigg)
\$
when $m$ is sufficiently large, 
where the second inequality uses 
$r\leq (1-\delta/2)\cdot m\cdot G_\infty(t)$, 
and the third inequality follows from  
$|\hat\cR_{j\to 0}|\leq (1+\delta) \cdot m\cdot G_\infty(t^*)$ 
and $\frac{1}{m}\sum_{j=1}^m \ind\{\hat p_j \leq s_j\} = |\cR^{(1)}|/m \geq (1-\delta/2)G_\infty(t^*)$ 
on the event $E(\delta)^c$. 
Since $\xi$ is independent of $E(\delta)$, we have 
\$
&\PP\Bigg(  \sum_{j=1}^m \ind\big\{  p_j \leq s_j, \xi |\hat\cR_{j\to 0}| \leq r \big\} \geq r \Bigg) 
\geq \PP\big(E(\delta)^c \big)\cdot  
\PP\Big( \xi \leq {\textstyle \frac{(1-2\delta) m G_\infty(t^*)}{(1+\delta) m G_\infty(t^*)} } \Biggiven E(\delta)^c \Big)
\geq \PP\big(E(\delta)^c \big)\cdot  \frac{1-2\delta}{1+\delta}.
\$
By the definition of $r^*_\homo$,  
for any $\delta\in (0,1/2)$, we have 
\$
\limsup_{m,n\to \infty} 
\PP\Big( r_\homo^* \geq \lfloor (1-\delta) m G_\infty(t^*)\rfloor \Big) \geq \frac{1-\delta}{1+\delta}.
\$
Since $|\cR^{(1)}|/m~\asto~ G_\infty(t^*)>0$, 
we have 
\$
\limsup_{m,n\to \infty} 
\PP\bigg( \frac{r_\homo^*}{|\cR^{(1)}|} \geq 1-\delta  \bigg) \geq \frac{1-2\delta}{1+\delta}.
\$
Finally, for any $0<\gamma<\delta$, applying the above inequality leads to 
\$
\limsup_{m,n\to \infty} 
\PP\bigg( \frac{r_\homo^*}{|\cR^{(1)}|} \geq 1-\delta  \bigg) \geq 
\limsup_{m,n\to \infty} 
\PP\bigg( \frac{r_\homo^*}{|\cR^{(1)}|} \geq 1-\gamma \bigg)\geq \frac{1-2\gamma}{1+\gamma}.
\$
Note that $r_\homo^* = |\cR_\homo|$, and that $\cR_\homo\subset \cR^{(1)}$. 
Thus, 
\$
\frac{|\cR^{(1)}\bigtriangleup \cR_\homo|}{|\cR^{(1)}|} = \frac{|\cR^{(1)}| - |\cR_\homo| }{|\cR^{(1)}|}
= 1 - \frac{r_\homo^*}{|\cR^{(1)}|}. 
\$
Sending $\gamma\to 0$, we conclude the proof of 
$\frac{|\cR^{(1)}\bigtriangleup \cR_\homo|}{|\cR^{(1)}|}~\pto ~0$ and 
$\frac{|\cR^{(1)}\bigtriangleup \cR_\homo|}{m}~\pto ~0$. 
\end{proof}

\subsection{Proof of Theorem~\ref{thm:est_w}}
\label{app:thm_est_w}

\begin{proof}[Proof of Theorem~\ref{thm:est_w}]

For clarity, throughout this proof, we denote the p-values 
used in the algorithm as
\$
\hat{p}_j &=   \frac{\sum_{i=1}^n  \hat{w}(X_i)\ind{\{V_i < \hat{V}_{n+j} \}}+  \hat{w}(X_{n+j}) }{\sum_{i=1}^n \hat{w}(X_i) + \hat{w}(X_{n+j})}, \\
\hat{p}_{\ell}^{(j)} &= \frac{\sum_{i=1}^n \hat{w}(X_i) \ind {\{V_i < \hat{V}_{n+\ell}\}} + \hat{w}(X_{n+j}) \ind {\{ \hat{V}_{n+j}<  \hat{V}_{n+\ell}\}} }{\sum_{i=1}^n \hat{w}(X_i) + \hat{w}(X_{n+j})}, \quad \ell\neq j,
\$
which uses the estimated weights $\hat{w}(\cdot)$. 
We keep the definition of $p_j$ and $p_\ell^{(j)}$ 
which uses the ground truth $w(\cdot)$ and are not available. 
We call BH($q$) the BH procedure at nominal level $q$. 
We also define 
\$
\tilde{p}_j &=   \frac{\sum_{i=1}^n  \hat{w}(X_i)\ind{\{V_i <  {V}_{n+j} \}}+  \hat{w}(X_{n+j}) }{\sum_{i=1}^n \hat{w}(X_i) + \hat{w}(X_{n+j})}, \\
\tilde{p}_{\ell}^{(j)} &= \frac{\sum_{i=1}^n \hat{w}(X_i) \ind {\{V_i < \hat{V}_{n+\ell}\}} + \hat{w}(X_{n+j}) \ind {\{  {V}_{n+j}<  \hat{V}_{n+\ell}\}} }{\sum_{i=1}^n \hat{w}(X_i) + \hat{w}(X_{n+j})}, \quad \ell\neq j
\$
as the counterparts of $\{\bar{p}_j\}$ and $\{\bar{p}_j^{\ell}\}$ 
we used in the proof of Theorem~\ref{thm:calib_ite} (c.f.~\eqref{eq:oracle_p_j} 
and~\eqref{eq:oracle_p_ell^j} for their definitions). 
Accordingly, we let  
$\tilde\cR_j$ 
be the rejection set from BH($q$) procedure with 
$\{\tilde{p}_1^{(j)},\dots,\tilde{p}_{j-1}^{(j)},\tilde{p}_j,\tilde{p}_{j+1}^{(j)},\dots,\tilde{p}_m^{(j)}\}$, and let 
$\tilde\cR_{j\to 0}$ 
be the rejection set from BH($q$) procedure with 
$\{\tilde{p}_1^{(j)},\dots,\tilde{p}_{j-1}^{(j)},0,\tilde{p}_{j+1}^{(j)},\dots,\tilde{p}_m^{(j)}\}$.

Since $\hat{w}(\cdot)$ is independent of the calibration and test sample, 
we view it as a fixed function throughout. 
The following 
Lemma~\ref{lem:fdr_ite_reduce_est_fix} is parallel to Lemma~\ref{lem:fdr_ite_reduce} 
in the proof of Theorem~\ref{thm:calib_ite} (c.f.~Appendix~\ref{app:thm_calib_ite}). 
The only difference between Lemma~\ref{lem:fdr_ite_reduce}
and Lemma~\ref{lem:fdr_ite_reduce_est_fix} is that 
we replace $w(\cdot)$ with $\hat{w}(\cdot)$. 
Since the proof of Lemma~\ref{lem:fdr_ite_reduce}
only uses the randomness in $\{\xi_j\}$ 
and $\xi$, which are independent of the observations, 
it still goes through for estimated weights. We thus omit the proof of 
Lemma~\ref{lem:fdr_ite_reduce_est_fix} 
for brevity. 
\begin{lemma}
\label{lem:fdr_ite_reduce_est_fix}
Let $\cR = \cR_{\hete}$ or $\cR = \cR_{\homo}$ or $\cR = \cR_{\dtm}$ in Algorithm~\ref{alg:bh} with estimated weights. 
Then 
\# \label{eq:cc_eq1_est_fix} 
\EE\Bigg[  \frac{\sum_{j=1}^m \ind{\{j \in \cR, Y_{n+j} \leq c_{n+j}\}} }{1\vee |\cR|}   \Bigg]  
\leq \sum_{j=1}^m  \EE\Bigg[ \frac{  \ind {\{  \hat p_j \leq s_j,Y_{n+j} \leq c_{n+j}\}} }{|\hat\cR_{j\to 0} |}  \Bigg],
\#
where $c_{n+j} = q|\hat\cR_{j\to 0}|/m$, and $\hat\cR_{j\to 0}$ is the 
rejection set in Line 8 of Algorithm~\ref{alg:bh} computed 
with $\hat{p}_j$ and $\hat{p}_\ell^{(j)}$. 
\end{lemma}

To bound the right-handed side of~\eqref{eq:cc_eq1_est_fix}, 
we note that the arguments up to~\eqref{eq:orc_fdr} 
in the proof of Theorem~\ref{thm:calib_ite} do not 
rely on $w(\cdot)$ being the true covariate shift function. 
Thus, exactly the same arguments imply
\#\label{eq:est_bd_1}
\fdr\leq 
\sum_{j=1}^m \EE\Bigg[ \frac{  \ind{\{   \tilde{p}_j \leq \frac{q |\tilde\cR_{j\to 0} |}{m}  \}} }{|\tilde\cR_{j\to 0}|}   \Bigg]. 
\#

Define the event 
$\cE_{z,j} = \big\{ [Z_1,\dots,Z_n,Z_{n+j}]=[z_1,\dots,z_n,z_{n+j}]\big\}$, 
where $[Z_1,\dots,Z_n,Z_{n+j}]$ is an unordered set of 
all calibration data 
$Z_i = (X_i,Y_i)$, $i=1,\dots,n$ 
and the unobserved test sample $n+j$, 
and $z:=[z_1,\dots,z_n,z_{n+j}]$ is the unordered set of their currently realized values. 
Since $\hat{w}(\cdot)$ is fixed and 
$\tilde\cR_{j\to 0}$ only depends on 
$\{\tilde{p}_1^{(j)},\dots,\tilde{p}_{j-1}^{(j)},0,\tilde{p}_{j+1}^{(j)},\dots,\tilde{p}_m^{(j)}\}$, we note that 
$\tilde\cR_{j\to 0}$ is invariant 
up to permutation of $\{(X_i,Y_i)\}_{i=1}^n\cup\{(X_{n+j},Y_{n+j})\}$. 
That is, $\tilde\cR_{j\to 0}$ is fully determined by $\cE_{z,j}$. 
Then for each $j$, we have 
\#\label{eq:est_bd_2}
\EE\Bigg[ \frac{  \ind{\{   \tilde{p}_j \leq \frac{q |\tilde\cR_{j\to 0} |}{m}  \}} }{|\tilde\cR_{j\to 0}|}   \Bigg] 
&= \sum_{k=1}^m \EE\Bigg[ \ind\{ |\tilde\cR_{j\to 0}|=k \}\frac{  \ind{\{   \tilde{p}_j \leq \frac{q k}{m}  \}} }{k}   \Bigg]\notag \\ 
&= \sum_{k=1}^m \EE\Bigg[ \frac{\ind\{ |\tilde\cR_{j\to 0}|=k \}}{k} \cdot \PP\Big(   \tilde{p}_j \leq \frac{q k}{m}   \Biggiven \cE_{z,j}\Big) \Bigg]
\#

In what follows, we are to bound the FDR by 
studying the difference between $\tilde{p}_j$ 
with estimated weights 
versus $\bar{p}_j$, 
i.e., what we would have should $w(\cdot)$ be available.  
For simplicity, we write 
\$
\hat{W}_- &= \sum_{i=1}^n  \hat{w}(X_i)\ind{\{V_i <  {V}_{n+j} \}}+  \hat{w}(X_{n+j}),  \quad 
\hat{W}_+  =  \sum_{i=1}^n  \hat{w}(X_i)\ind{\{V_i > {V}_{n+j} \}},\\ 
W_- &= \sum_{i=1}^n  {w}(X_i)\ind{\{V_i <  {V}_{n+j} \}}+  {w}(X_{n+j}), \quad 
W_+   =\sum_{i=1}^n  {w}(X_i)\ind{\{V_i > {V}_{n+j} \}}.
\$
Then $\hat{W}_-+\hat{W}_+ = \sum_{i=1}^n \hat{w}(X_i) + \hat{w}(X_{n+j})$ and 
$W_-+W_+ = \sum_{i=1}^n  {w}(X_i) +  {w}(X_{n+j})$. 
Recall that $\hat\gamma := \sup_{x}\{ \frac{\hat{w}(x)}{w(x)} \vee \frac{w(x)}{\hat{w}(x)}  \}$ is independent of all the data hence can be viewed as fixed. 
Since $\hat\gamma^{-1}\leq \hat{w}(x)/w(x)\leq\hat \gamma$, we know 
$W_-\leq \hat\gamma \hat{W}_-$ and $W_+\geq \hat\gamma^{-1}\hat{W}_+$. 
Also, recall that $\bar{p}_j = \frac{\sum_{i=1}^n w(X_i)\ind {\{V_i < {V}_{n+j} \}}+   w(X_{n+j})  }{\sum_{i=1}^n w(X_i) + w(X_{n+j})}$ is the 
counterpart with the true weight function.
Thus, on the event 
$
\big\{ \tilde p_j = \frac{\hat{W}_-}{\hat{W}_- + \hat{W}_+ } \leq \frac{qk}{m} \big\}
$,
it holds that 
\$
\bar{p}_j = \frac{W_-}{W_- + W_+} \leq \frac{\hat\gamma^2 }{\hat\gamma^2  +  \hat{W}_+ /\hat{W}_-} \leq \frac{\hat\gamma^2 }{\hat\gamma^2  + \frac{m}{qk}-1 }.
\$
We have shown in the proof 
of Theorem~\ref{thm:calib_ite} (c.f.~Appendix~\ref{app:thm_calib_ite})
that $\PP(\bar p_j \leq t\given \cE_{z,j})\leq t$ for 
all $t\in \RR$ that is adapted to $\cE_{z,j}$. Thus
\$
&\sum_{k=1}^m \EE\Bigg[ \frac{\ind\{ |\tilde\cR_{j\to 0}|=k \}}{k} \cdot \PP\Big(   \tilde{p}_j \leq \frac{q k}{m}   \Biggiven \cE_{z,j}\Big) \Bigg] \\
&\leq \sum_{k=1}^m \EE\Bigg[ \frac{\ind\{ |\tilde\cR_{j\to 0}|=k \}}{k} \cdot \PP\Big(   \bar{p}_j \leq \frac{\hat\gamma^2 }{\hat\gamma^2  + \frac{m}{qk}-1 }  \Biggiven \cE_{z,j}\Big) \Bigg] \\ 
&\leq \sum_{k=1}^m \EE\Bigg[ \frac{\ind\{ |\tilde\cR_{j\to 0}|=k \}}{k} \cdot   \frac{\hat\gamma^2 }{\hat\gamma^2  + \frac{m}{qk}-1 }   \Bigg] 
\leq \max_{k=1,\dots,m} \bigg\{ \frac{\hat\gamma^2 /k }{\hat\gamma^2  + \frac{m}{qk}-1 }   \bigg\} = \frac{q}{m} \cdot \frac{\hat\gamma^2}{1+ q(\hat\gamma^2-1)/m}.
\$
Finally, combining~\eqref{eq:est_bd_1},~\eqref{eq:est_bd_2} 
and the above bound, we have  
\$
\fdr\leq q\cdot \EE\bigg[ \frac{\hat\gamma^2}{1+ q(\hat\gamma^2-1)/m} \bigg],
\$
which completes the proof 
of Theorem~\ref{thm:calib_ite}.   
\end{proof}

\subsection{Proof of Theorem~\ref{thm:fdr_cond}}
\label{app:thm_outlier}

\begin{proof}[Proof of Theorem~\ref{thm:fdr_cond}]
Similar to Lemma~\ref{lem:fdr_ite_reduce}, we use 
the following lemma as an intermediate step. Its proof 
is in Appendix~\ref{app:lemma_outlier_reduce}. 

\begin{lemma}
\label{lem:fdr_out_reduce}
Let $\cR = \cR_{\hete}$ or $\cR = \cR_{\homo}$ or $\cR = \cR_{\dtm}$ in Algorithm~\ref{alg:bh}. 
Then 
\# \label{eq:out_1}
 \EE\bigg[  \frac{| \cR \cap \cH_0 | }{1\vee |\cR|} \biggiven \cH_0 \bigg]  
\leq \sum_{j\in \cH_0}  \EE\bigg[ \frac{  \ind {\{  p_j \leq s_j \}} }{|\hat\cR_{j\to 0} |} \bigggiven \cH_0  \bigg],
\#
where $s_j = q|\hat\cR_{j\to 0}|/m$, and $\hat\cR_{j\to 0}$ is the 
first-step rejection set we obtain in Line 8 of Algorithm~\ref{alg:bh_cond}. 
\end{lemma}

Throughout this proof, we always condition on $\cH_0$ 
and omit such dependence in our notations for simplicity. 
Note that $\hat\cR_{j\to 0}$ only depends on  $\{p_\ell^{(j)}\}_{\ell \neq j}$, 
which is measurable with respect to 
the unordered set $Z=[Z_1,\dots,Z_n,Z_{n+j}]$ and $\{X_{n+\ell}\}_{\ell\neq j}$, 
where $Z_i = (X_i,Y_i)$ for $i=1,\dots,n+m$. 
We denote 
\$\cE_{z,j}:=\big\{[Z_1,\dots,Z_n,Z_{n+j}]=z\big\},
\quad z=[z_1,\dots,z_n,z_{n+j}]
\$ as the event that 
the unordered set $[Z_1,\dots,Z_n,Z_{n+j}]$
takes on (the unordered set of) 
their realized values $z=[z_1,\dots,z_n,z_{n+j}]$. 
Thus, conditional on the event $\cE_{z,j}$, the 
only randomness in $p_j$ is which values among $z_1,\dots,z_{n},z_{n+j}$ 
that $Z_1,\dots,Z_n,Z_{n+j}$ take; 
such information, on the other hand, does not 
affect $\{p_\ell^{(j)}\}_{\ell\neq j}$.  
As a result, we have the conditional independence 
\#\label{eq:cond_indep_out}
\ind{\{p_j \leq c_{n+j}\}} \indep  |\hat\cR_{j\to 0}  | ~\biggiven ~\cE_{z,j},\quad \forall j\in \cH_0. 
\#
Therefore, by the tower property, 
\$
\EE\bigg[ \frac{  \ind{\{p_j \leq s_j\}} }{|\hat\cR_{j\to 0} |} \Biggiven \cE_{z,j}  \bigg]
= \EE\Bigg[ \frac{  \PP\big(p_j \leq s_j \biggiven \cE_{z,j},|\hat\cR_{j\to 0} |\big)  }{|\hat\cR_{j\to 0} |} \Biggiven \cE_{z,j}  \Bigg] 
= \EE\Bigg[ \frac{  F_j(s_j)  }{|\hat\cR_{j\to 0} |} \Biggiven \cE_{z,j}  \Bigg] ,
\$
where we denote $F_j(t)=\PP(p_j\leq t\given \cE_{z,j})$ for any fixed $t\in [0,1]$.
Finally, 
by the covariate shift condition in Assumption~\ref{assump:label_conditional}, 
for any $j\in \cH_0$ 
and any $i=1,\dots,n,n+j$, 
we have the  weighted exchangeability similar to~\eqref{eq:weighted_exch}, 
which implies 
\$
\PP(Z_{n+j} = z_i\given \cE_{z,j}) = \frac{w(z_i)}{\sum_{i=1}^n w(x_{i}) + w(x_{n+j})}.
\$
Let $[1],[2],\dots,[n+1] $ be 
a permutation of $ 1,\dots,n, n+j $ such that 
$V_{[1]}\leq V_{[2]}\leq \cdots V_{[n+1]}$. By the definition of $p_j$, 
we know that 
conditional on $\cE_{z,j}$, 
\$
F_j(t) = \sum_{k=1}^n \frac{w(z_{[k]}) }{\sum_{i=1}^n w(x_{i}) + w(x_{n+j})} \ind\bigg\{\sum_{j=1}^k \frac{w(z_{[j]}) }{\sum_{i=1}^n w(x_{i}) + w(x_{n+j})} \leq t \bigg\} \leq t 
\$
for any $t\in [0,1]$. 
Invoking Lemma~\ref{lem:fdr_out_reduce}, we thus have 
\$
\fdr \leq \sum_{j\in \cH_0} \EE\Bigg[ \frac{   s_j   }{|\hat\cR_{j\to 0} |} \Biggiven \cE_{z,j}  \Bigg] \leq q\cdot \frac{|\cH_0|}{m},
\$
which concludes the proof of Theorem~\ref{thm:fdr_cond}.
\end{proof}

\subsection{Proof of Lemma~\ref{lem:fdr_out_reduce}}
\label{app:lemma_outlier_reduce}

\begin{proof}[Proof of Lemma~\ref{lem:fdr_out_reduce}]
We prove the three pruning methods 
using similar arguments as the proof of Lemma~\ref{lem:fdr_ite_reduce}. 

\paragraph{Case 1: $\cR=\cR_{\hete}$.} By definition, 
we know that  $r^*_\hete =  |\cR_{\hete}| = |\cR|$. 
Then  
\$
\EE\bigg[  \frac{| \cR \cap \cH_0 | }{1\vee |\cR|} \biggiven \cH_0 \bigg]  
&= \EE\Bigg[ \frac{\sum_{j\in \cH_0} \ind {\{p_j \leq s_j\}} \ind{\{\xi_j \leq \frac{|\cR|}{ |\hat\cR_{j\to 0} | }\}} }{1\vee |\cR|} \Bigggiven \cH_0 \Bigg].
\$
Note that once $j\in \cR$, sending $\xi_j\to 0$ while keeping $\{\xi_\ell\}_{\ell\neq j}$
unchanged does not change the rejection set. 
That is, $\cR = \cR_{\xi_j\to 0}$ if $j\in \cR$, 
where 
\$
\cR_{\xi_j\to 0} &= \{j\}\cup \Big\{ \ell \colon  p_\ell \leq s_\ell, ~ \xi_\ell |\hat\cR_{\ell\to 0} | \leq  r_{j,\hete}^*  \Big\}, ~~
  r_{j,\hete}^* := \max\Big\{ r\colon 1+ \sum_{\ell\neq j} \ind {\{  p_\ell\leq s_\ell, \, \xi_\ell   |\hat\cR_{\ell\to 0}  | \leq r \}}   \geq r \Big\}.
\$
and we recall that $\hat\cR_{\ell\to 0}$ is the first-step rejection set. 
We then have 
\$
\fdr &= \sum_{j\in \cH_0} \sum_{k=1}^m \EE\Bigg[ \frac{\ind {\{|\cR|=k\}}}{k} \ind {\{p_j \leq s_j\}} \ind \big\{\xi_j \leq {\textstyle \frac{k}{ |\hat\cR_{j\to 0} | }}\big\}    \Bigg] \\ 
&= \sum_{j\in \cH_0} \sum_{k=1}^m \EE\Bigg[ \frac{\ind {\{|\cR_{\xi_j\to 0}|=k\}}}{k}  \ind{\{p_j \leq s_j\}} \ind\big\{\xi_j \leq {\textstyle \frac{k}{ |\hat\cR_{j\to 0} | }}\big\}\Bigg] \\ 
&\leq \sum_{j\in \cH_0} \sum_{k=1}^m \EE\Bigg[ \frac{\ind {\{|\cR_{\xi_j\to 0}|=k\}}}{k}  \ind{\{p_j \leq s_j\}}\cdot  \frac{k}{ |\hat\cR_{j\to 0} |  } \Bigg] \\ 
&= \sum_{j\in \cH_0} \sum_{k=1}^m \EE\Bigg[ \frac{\ind{\{|\cR_{\xi_j\to 0}|=k\}}}{|\hat\cR_{j\to 0} |}  \ind{\{p_j \leq s_j\}}   \Bigg]
= \sum_{j\in \cH_0}  \EE\Bigg[ \frac{  \ind{\{p_j \leq s_j\}} }{|\hat\cR_{j\to 0} |}   \Bigg],
\$
where the third line follows from the tower property, and 
that $\xi_j$ are independent of $\cR_{\xi_j\to 0}$, $p_j$, $\hat{\cR}_{j\to 0}$. 

\paragraph{Case 2: $\cR=\cR_\homo$.} 
We will again use the conditional PRDS property of $\{\xi |\hat\cR_{j\to 0}|\}_{j=1}^m$ which we proved in Lemma~\ref{lem:prds}. 
We let $\EE_\cD$ be the conditional expectation given all the 
calibration and test data (and $\cH_0$), 
so that only the randomness in $\xi$ remains, under which 
one still has $\xi\sim \textrm{Unif}[0,1]$. 
The tower property implies 
$\fdr = \EE[\fdr(\cD)\given \cH_0]$, where we define the conditional FDR as
\$
\fdr(\cD) := \EE_{\cD}\Bigg[  \frac{\sum_{j\in \cH_0} \ind{\{j \in \cR \}} }{1\vee |\cR|} \Biggiven \cH_0 \Bigg].
\$
Note that $\ind\{  p_j \leq s_j\}$ are deterministic conditional on the data. 
Recall that $\cR^{(1)} = \{j\colon   p_j \leq s_j\}$  
and $\cH_0 $ are
both deterministic given the data. 
Using the same arguments as in the proof of Lemma~\ref{lem:fdr_ite_reduce} 
up to~\eqref{eq:homo_bd},  
\$
\fdr(\cD,j)&:=\sum_{k=1}^{m} \frac{1}{k}\EE_{\cD}\big[  \ind{\{ \xi |\hat\cR_{j\to 0}| \leq k \}} \left(  \ind\{ |\cR| \leq k\} -  \ind\{ |\cR| \leq k-1\} \right) \big] \\ 
&\leq \frac{1}{m_1} + \sum_{k=1}^{m_1-1} \bigg\{\frac{\PP_\cD(\xi |\hat\cR_{j\to 0}| \leq k)}{k}
\PP_{\cD}\Big( |\cR| \leq k \Biggiven \xi |\hat\cR_{j\to 0}| \leq k \Big)  \\
&\qquad \qquad \qquad \qquad - 
  \frac{\PP_\cD(\xi |\hat\cR_{j\to 0}| \leq k+1)}{k+1}
\PP_{\cD}\Big( |\cR| \leq k \Biggiven \xi |\hat\cR_{j\to 0}| \leq k+1 \Big)\bigg\}
\$
By Lemma~\ref{lem:prds}, we know that 
conditioning on all the data, 
$\{\xi|\hat\cR_{\ell\to 0}|\}_{\ell\neq j}$ is 
PRDS in $\xi |\hat\cR_{j\to 0}|$.  
Then by the construction of $\cR=\cR_{\homo}$, we see that 
\$
\PP_{\cD}\Big( |\cR| \leq k \Biggiven \xi |\hat\cR_{j\to 0}| \leq k \Big) 
\leq \PP_{\cD}\Big( |\cR| \leq k \Biggiven \xi |\hat\cR_{j\to 0}| \leq k+1 \Big) ,
\$
hence since $\xi$ is independent of everything else, 
\$
\fdr(\cD,j) &\leq \frac{1}{m_1} + \sum_{k=1}^{m_1-1} \bigg\{\frac{\PP_\cD(\xi |\hat\cR_{j\to 0}| \leq k)}{k} -  \frac{\PP_\cD(\xi |\hat\cR_{j\to 0}| \leq k+1)}{k+1}\bigg\}
\PP_{\cD}\Big( |\cR| \leq k \Biggiven \xi |\hat\cR_{j\to 0}| \leq k \Big) \\ 
&= \frac{1}{m_1} + \sum_{k=1}^{m_1-1} \bigg\{\frac{ \min\{1, k/|\hat\cR_{j\to 0}|\} }{k} -  \frac{\min\{1, (k+1)/|\hat\cR_{j\to 0}|\} }{k+1}\bigg\}
\PP_{\cD}\Big( |\cR| \leq k \Biggiven \xi |\hat\cR_{j\to 0}| \leq k \Big) . 
\$
This is exactly the same as~\eqref{eq:prds_bd}. Hence following the same 
arguments as in the proof of Lemma~\ref{lem:fdr_ite_reduce}, we 
conclude the proof of the homogeneous case. 

\paragraph{Case 3: $\cR=\cR_{\dtm}$.} 
In this case, similar to the proof of Lemma~\ref{lem:fdr_ite_reduce}, 
we know that  $|\hat\cR_{j\to 0}|\leq |\cR|$ deterministically for all $j\in \cR$. 
Thus, 
\$
\EE\Bigg[  \frac{\sum_{j\in\cH_0} \ind{\{j \in \cR \}} }{1\vee |\cR|} \Bigg]  
&\leq \sum_{j\in\cH_0}  \EE\Bigg[  \frac{ \ind {\{ j\in \cR \}}   }{1\vee |\hat\cR_{j\to 0}|} \Bigg]
\leq \sum_{j=1}^m  \EE\Bigg[  \frac{ \ind {\{ p_j\leq s_j \}}   }{1\vee |\hat\cR_{j\to 0}|} \Bigg],
\$
where the last inequality uses the fact that 
$j\in \cR$ implies $p_j \leq s_j$ for any $j$. 
This concludes the proof of the third case, and 
therefore the proof of Lemma~\ref{lem:fdr_out_reduce}. 
\end{proof}

\section{Supporting lemmas}

The following lemma establishes the uniform strong law of 
large numbers with weights. The proof is standard which we include here for completeness. 

\begin{lemma}[Uniform strong law of large numbers]
\label{lem:ulln}
Let $P$ and $\tilde{P}$ be two distributions over $\cX\times\cY$, related 
by a covariate shift $\frac{\ud \tilde{P}}{\ud P}(x,y)=w(x)$ 
for some function $w\colon \cX\to \RR^+$, and $\frac{\ud \TPP}{\ud \PP}$ 
is the Radon–Nikodym derivative. 
Let $\{(X_i,Y_i)_{i=1}^n\}$ be a sequence of i.i.d.~samples 
from $P$, and $f\colon \cX\times \cY\to \RR$ 
be any fixed function.  
Then it holds almost surely that 
\$
\limsup_{n\to \infty} \sup_{t \in \RR} ~\bigg| \frac{1}{n} \sum_{i=1}^n w(X_i) \ind\{ f(X_i,Y_i) \leq t \}
- \tilde{P}\big( f(X,Y) \leq t \big)\bigg| = 0.
\$ 
\end{lemma}

\begin{proof}[Proof of Lemma~\ref{lem:ulln}]
For simplicity, we denote 
\$
\tilde{F}_n(t) = \frac{1}{n} \sum_{i=1}^n w(X_i) \ind\{ f(X_i,Y_i) \leq t \}, 
\quad \tilde{F}(t)=\tilde{P}(f(X,Y)\leq t),\quad t\in \RR.
\$ 
Then by the law of large numbers, for any fixed $t\in \RR$, 
$\tilde{F}_n(t)\to\tilde{F}(t)$ almost surely 
as $n\to \infty$. 

Let $\epsilon>0$ be any fixed constant. 
Then we can pick a finite sequence of real numbers 
$t_1<t_2<\cdots <t_m$ such that 
$\tilde{F}(t_1) \leq \epsilon$,
$\tilde{F}( t_m) \geq 1-\epsilon$, 
and $\tilde{F}( t_{k+1})-\tilde{F}(t_k) \leq \epsilon$. 
For any $t\in [t_1,t_m)$, 
there exists some $1\leq k\leq m$ such that 
$t\in[t_k,t_{k+1})$. Then we know 
$ 
\tilde{F}_n(t_k) \leq \tilde{F}_n(t) \leq \tilde{F}_n(t_{k+1})$ and 
$\tilde{F} (t_k) \leq \tilde{F} (t) \leq \tilde{F} (t_{k+1})
$. This implies 
\$
\tilde{F}_n(t) - \tilde{F} (t) &\leq \tilde{F}_n(t_{k+1}) - \tilde{F}(t_k) 
\leq \big| \tilde{F}_n(t_{k+1}) - \tilde{F}(t_{k+1}) \big| + \tilde{F}(t_{k+1}) - \tilde{F}(t_{k})
\leq \big| \tilde{F}_n(t_{k+1}) - \tilde{F}(t_{k+1}) \big| + \epsilon,
\$
where we use the triangular inequality in the second inequality. 
Similarly
\$
\tilde{F}_n(t) - \tilde{F} (t)\geq \tilde{F}_n(t_k) - \tilde{F}(t_{k+1}) 
\geq \tilde{F}_n(t_k) - \tilde{F}(t_k)  + \tilde{F}(t_k) - \tilde{F}(t_{k+1})
\geq - \big|\tilde{F}_n(t_k) - \tilde{F}(t_k) \big| - \epsilon.
\$
They together lead to $|\tilde{F}_n(t)-\tilde{F}(t)|\leq 
\big|\tilde{F}_n(t_k) - \tilde{F}(t_k) \big| \vee \big| \tilde{F}_n(t_{k+1}) - \tilde{F}(t_{k+1}) \big| +\epsilon$ for all $t\in [t_k,t_{k+1})$. 
Also, for $t<t_1$, since $\tilde{F}_n(t)\geq 0$ and $\tilde{F}(t)\leq \epsilon$,  we have 
\$
-\epsilon \leq \tilde{F}_n(t) - \tilde{F} (t) \leq \tilde{F}_n(t_1) - \tilde{F}(t) 
\leq \tilde{F}_n(t_1)- \tilde{F}(t_1) + \epsilon. 
\$
Similarly, for $t>t_m$, since $\tilde{F}_n(t)\leq 1$ and $\tilde{F}(t)\geq 
\tilde{F}(t_m)\geq 1-\epsilon$, we have 
\$
\epsilon \geq 
\tilde{F}_n(t) - \tilde{F} (t) \geq  \tilde{F}_n(t_m) - \tilde{F} (t ) 
\geq \tilde{F}_n(t_m) - \tilde{F}(t_m) -\epsilon.
\$
Putting the above cases together yields 
\$
\sup_{t\in \RR} \big|\tilde{F}_n(t) -\tilde{F}(t)\big| 
\leq \sup_{1\leq k\leq m} \big|\tilde{F}_n(t_k) -\tilde{F}(t_k)\big| + \epsilon.
\$
By the strong law of large numbers for $t_1,\dots,t_m$ and the arbitrariness of $\epsilon>0$, 
we conclude the proof of Lemma~\ref{lem:ulln}. 
\end{proof}

\end{document}